\documentclass[12pt]{article}

\usepackage{amsmath}
\usepackage{amssymb}
\usepackage{amsthm}
\usepackage{hyperref}
\usepackage{enumerate}

\newtheorem{thm}{Theorem}[section]

\newtheorem{lem}[thm]{Lemma}

\newtheorem{assumption}[thm]{Assumption}

\newtheorem{definition}[thm]{Definition}

\newtheorem{example}[thm]{Example}

\newtheorem{remark}[thm]{Remark}
\newenvironment{rem}{\begin{remark}\rm}{\end{remark}}

\newcommand{\E}{\mathbb{E}}
\newcommand{\eps}{\epsilon}

\newcommand{\EX}{{\Bbb{E}}}
\newcommand{\PX}{{\Bbb{P}}}

\title{On Continuous-Time Gaussian Channels~\footnote{Results in this paper have been partially presented in the 2014 IEEE ISIT~\cite{LiuHan2014}.}}

\author{\small \begin{tabular}{ccc}
Xianming Liu & Guangyue Han\\
Huazhong University of Science and Technology & The University of Hong Kong\\
email: xmliu@hust.edu.cn & email: ghan@hku.hk\\
\end{tabular}}

\date{{\normalsize \today}}

\date{{\normalsize \today}}

\begin{document} \maketitle

\begin{abstract}
A continuous-time white Gaussian channel can be formulated using a white Gaussian noise, and a conventional way for examining such a channel is the sampling approach based on the Shannon-Nyquist sampling theorem, where the original continuous-time channel is converted to an equivalent discrete-time channel, to which a great variety of established tools and methodology can be applied. However, one of the key issues of this scheme is that continuous-time feedback and memory cannot be incorporated into the channel model. It turns out that this issue can be circumvented by considering the Brownian motion formulation of a continuous-time white Gaussian channel. Nevertheless, as opposed to the white Gaussian noise formulation, a link that establishes the information-theoretic connection between a continuous-time channel under the Brownian motion formulation and its discrete-time counterparts has long been missing. This paper is to fill this gap by establishing causality-preserving connections between continuous-time Gaussian feedback/memory channels and their associated discrete-time versions in the forms of sampling and approximation theorems, which we believe will play important roles in the long run for further developing continuous-time information theory.

As an immediate application of the approximation theorem, we propose the so-called approximation approach to examine continuous-time white Gaussian channels in the point-to-point or multi-user setting. It turns out that the approximation approach, complemented by relevant tools from stochastic calculus, can enhance our understanding of continuous-time Gaussian channels in terms of giving alternative and strengthened interpretation to some long-held folklore, recovering ``long known'' results from new perspectives, and rigorously establishing new results predicted by the intuition that the approximation approach carries. More specifically, using the approximation approach complemented by relevant tools from stochastic calculus, we first derive the capacity regions of continuous-time white Gaussian multiple access channels and broadcast channels, and we then analyze how feedback affects their capacity regions: feedback will increase the capacity regions of some continuous-time white Gaussian broadcast channels and interference channels, while it will not increase capacity regions of continuous-time white Gaussian multiple access channels.
\end{abstract}


\section{Introduction}  \label{intro}

Continuous-time Gaussian channels were considered at the very inception of information theory. In his celebrated paper~\cite{sh48} birthing information theory, Shannon studied the following point-to-point continuous-time white Gaussian channels:
\begin{equation} \label{white-Gaussian-noise-channel}
Y(t)=X(t)+Z(t), \quad t \in \mathbb{R},
\end{equation}
where $X(t)$ is the channel input with average power limit $P$, $Z(t)$ is the white Gaussian noise with flat power spectral density $1$ and $Y(t)$ are the channel output. Shannon actually only considered the case that the channel has bandwidth limit $\omega$, namely, the channel input $X$ and the noise $Z$, and therefore the output $Y$ all have bandwidth limit $\omega$ (alternatively, as in ($9.54$) of~\cite{co2006}, this can be interpreted as the original channel (\ref{white-Gaussian-noise-channel}) concatenated with an ideal bandpass filter with bandwidth limit $\omega$). Using the celebrated Shannon-Nyquist sampling theorem~\cite{ny24, sh49}, the continuous-time channel (\ref{white-Gaussian-noise-channel}) can be equivalently represented by a parallel Gaussian channel:
\begin{equation} \label{sampled-white-Gaussian-noise-channel}
Y_{n}^{(\omega)}=X_{n}^{(\omega)}+Z_{n}^{(\omega)}, \quad n \in \mathbb{Z},
\end{equation}
where the noise process $\{Z_n^{(\omega)}\}$ is i.i.d. with variance $1$~\cite{co2006}. Regarding the ``space'' index $n$ as time, the above parallel channel can be interpreted as a discrete-time Gaussian channel associated with the continuous-time channel (\ref{white-Gaussian-noise-channel}). It is well known from the theory of discrete-time Gaussian channels that the capacity of the channel (\ref{sampled-white-Gaussian-noise-channel}) can be computed as
\begin{equation} \label{finite-bandwidth-capacity}
C^{(\omega)}=\omega \log \left(1+\frac{P}{2 \omega}\right).	
\end{equation}
Then, the capacity $C$ of the channel (\ref{white-Gaussian-noise-channel}) can be computed by taking the limit of the above expression as $\omega$ tends to infinity:
\begin{equation} \label{infinite-bandwidth-capacity}
C=\lim_{\omega \to \infty} C^{(\omega)}=P/2.
\end{equation}

The {\em sampling approach} consisting of (\ref{white-Gaussian-noise-channel})-(\ref{infinite-bandwidth-capacity}) as above, which serves as a link between the continuous-time channel (\ref{white-Gaussian-noise-channel}) and the discrete-time channel (\ref{sampled-white-Gaussian-noise-channel}), typifies a conventional way to examine continuous-time Gaussian channels: convert them into associated discrete-time Gaussian channels, for which we have ample ammunition at hands. Note that when $P$ tends to $0$, using the fact that when $P$ is ``close'' to $0$,
$$
\omega \log  \left(1+\frac{P}{2 \omega}\right) \mbox{ is ``close'' to } \frac{P}{2},
$$
one also reaches (\ref{infinite-bandwidth-capacity}), which roughly explains the following long-held folklore within the information theory community:
\begin{quote}
{\em a continuous-time infinite-bandwidth Gaussian channel {\bf \em without feedback or memory} is ``equivalent'' to a discrete-time Gaussian channel without feedback or memory at low signal-to-noise ratio (SNR).} \hfill (A)
\end{quote}
Moments of reflection, however, reveals that the sampling approach for the channel capacity (with bandwidth limit or not) is heuristic in nature: For one thing, a bandwidth-limited signal cannot be time-limited, which renders it infeasible to define the data transmission rate if assuming a channel has bandwidth limit. In this regard, rigorous treatments coping with this issue and other technicalities can be found in~\cite{wy66, ga68}; see also~\cite{sl76} for a relevant in-depth discussion. Another issue is that, even disregarding the above technical nuisance arising from the bandwidth limit assumption, the sampling approach only gives a lower bound for the capacity of (\ref{white-Gaussian-noise-channel}): it shows that $P/2$ is achievable via a class of special coding schemes, but it is not clear that why transmission rate higher than $P/2$ cannot be achieved by other coding schemes. The capacity of (\ref{white-Gaussian-noise-channel}) was rigorously studied in~\cite{fo61, be62}, and a complete proof establishing $P/2$ as its de facto capacity can be found in~\cite{as63,as64}.

Alternatively, the continuous-time white Gaussian channel (\ref{white-Gaussian-noise-channel}) can be examined~\cite{ih93} under the Brownian motion formulation:
\begin{equation} \label{Brownian-motion-channel}
Y(t)=\int_0^t X(s)ds + B(t),
\end{equation}
where, slightly abusing the notation, we still use $Y(t)$ to denote the output corresponding to the input $X(s)$, and $B(t)$ denotes the standard Brownian motion ($Z(t)$ can be viewed as a generalized derivative of $B(t)$); equivalently, the channel (\ref{Brownian-motion-channel}) can be seen as the original channel (\ref{white-Gaussian-noise-channel}) concatenated with an integrator circuit. As opposed to white Gaussian noises, which only exist as generalized functions~\cite{po94}, Brownian motions are well-defined stochastic processes and have been extensively studied in probability theory. Here we remark that, via a routine orthonormal decomposition argument, both of the two channels are equivalent to a parallel channel consisting of infinitely many Gaussian sub-channels~\cite{as65}.

An immediate and convenient consequence of such a formulation is that many notions in discrete time, including mutual information and typical sets, carry over to the continuous-time setting, which will rid us of the nuisances arising from the bandwidth limit assumption. Indeed, such a framework yields a fundamental formula for the mutual information of the channel (\ref{Brownian-motion-channel})~\cite{Duncan70, ka71} and a clean and direct proof~\cite{ka71} that the capacity of (\ref{Brownian-motion-channel}) is $P/2$; moreover, as evidenced by numerous results collected in~\cite{ih93} and some recent representative work~\cite{we10, we13} on point-to-point Gaussian channels, the use of Brownian motions elevates the level of rigor of our treatment, and equip us with a wide range of established techniques and tools from stochastic calculus. Here we remark that Girsanov's theorem, one of the most important theorems in stochastic calculus, lays the foundation of our rigourous treatment; for those who are interested in the technical details in our proofs, we refer to~\cite{li01,ih93}, where Girsanov's theorem (and its numerous variants) and its wide range of applications in information theory are discussed in great details.

Furthermore, as elaborated in Remark~\ref{from-continuous-to-discrete}, the Brownian motion formulation is also versatile enough to accommodate feedback and memory; in particular, the point-to-point continuous-time white Gaussian memory/feedback channel can be characterized by the following stochastic differential equation:
\begin{equation} \label{feedback-memory}
Y(t)=\int_0^t g(s, W_0^s, Y_0^{s}) ds + B(t), \quad t \in [0, T],
\end{equation}
where $g$ is a function from $[0, T] \times C[0, T] \times C[0, T]$ to $\mathbb{R}$. Note that (\ref{feedback-memory}) can be interpreted
\begin{enumerate}
\item[1)] either as a feedback channel, where $W_0^s$ can be rewritten as $M$, interpreted as the message to be transmitted through the channel, and $g(s)$ can be rewritten as $X(s)$, interpreted as the channel input, which depends on $M$ and $Y_0^{s}$, the channel output up to time $s$ that is fed back to the sender,
\item[2)] or as a memory channel, where $W_0^s$ can rewritten as $X_0^s$, interpreted as the channel input, $g$ is ``part'' of the channel, and $Y(t)$, the channel output at time $t$, depends on $X_0^t$ and $Y_0^t$, the channel input and output up to time $t$ that are present in the channel as memory, respectively.
\end{enumerate}
Note that, strictly speaking, the third parameter of $g$ in (\ref{feedback-memory}) should be $Y_0^{s-}$, which, however, can be equivalently replaced by $Y_0^s$ due to the continuity of sample paths of $\{Y(t)\}$. Note that, with the presence of feedback/memory, the existence and uniqueness of $Y$ is in fact a tricky mathematical problem, however, we will in this paper simply assume that the input $X$ is appropriately chosen such that $Y$ uniquely exists. For more detailed discussion about the Brownian motion formulation, we refer the reader to~\cite{ih93}.

As opposed to the white Gaussian noise formulation, under the Brownian motion formulation, memory and feedback can be naturally translated to the discrete-time setting: the pathwise continuity of a Brownian motion allows the inheritance of temporal causality when the channel is sampled (see Section~\ref{sampling-theorems}) or approximated (see Section~\ref{approximation-theorems}). On the other hand, the white Gaussian noise formulation is facing inherent difficulty as far as inheriting temporal causality is concerned: in converting (\ref{white-Gaussian-noise-channel}) to (\ref{sampled-white-Gaussian-noise-channel}), while $X_n^{(w)}$ are obtained as ``time'' samples of $X(t)$, $Z_n^{(w)}$ are in fact ``space'' samples of $Z(t)$, as they are merely the coefficients of the (extended) Karhunen-Loeve decomposition of $Z(t)$~\cite{ge59, huang, johnson}; see also~\cite{yhk06} for an in-depth discussion on this.

On the other hand though, as opposed to the white Gaussian noise formulation, a link that establishes the information-theoretic connection between the continuous-time channel (\ref{feedback-memory}) and its discrete-time counterparts has long been missing, which may explain why discrete-time and continuous-time information theory (under the Brownian motion formulation) have largely gone separate ways with little interaction for the past several decades. In this paper, we will fill this gap by establishing causality-preserving connections between the channel (\ref{Brownian-motion-channel}) and its associated discrete-time versions in the forms of sampling and approximation theorems, which we believe will serve as the above-mentioned missing links and play important roles in the long run for further developing continuous-time information theory, particularly for the communication scenarios when feedback/memory is present.

As an immediate application of the approximation theorem, we propose the approximation approach to examine continuous-time Gaussian feedback channels with the average power constraint and infinite bandwidth (again, by comparison, the conventional sampling approach cannot handle feedback). It turns out that this approach, when complemented by relevant tools from stochastic calculus, can greatly enhance our understanding of continuous-time Gaussian channels in terms of giving alternative and strengthened interpretations to the low SNR equivalence in (A), recovering ``long known'' results (Theorems~\ref{Theorem-MAC}(for the non-feedback case),~\ref{Theorem-IC-Without-Feedback} and~\ref{Theorem-BC-Without-Feedback}) from new and rigorous perspectives, and deriving new results (Theorems~\ref{Theorem-MAC}(for the feedback case),~\ref{Theorem-IC-With-Feedback},~\ref{Theorem-BC-With-Feedback-1} and~\ref{Theorem-BC-With-Feedback-2}) inspired by the intuition that the approximation approach carries.

Below, we summarize the contributions of this paper in greater details.

In Section~\ref{sampling-theorems}, we prove Theorems~\ref{sampling-theorem-1} and~\ref{sampling-theorem-2}, sampling theorems for a continuous-time Gaussian feedback/memory channel, which naturally connect such a channel with their sampled discrete-time versions. And in Section~\ref{approximation-theorems}, we prove Theorems~\ref{approximation-theorem-1} and~\ref{approximation-theorem-2}, the so-called approximation theorems, which connect a continuous-time Gaussian feedback/memory channel with its approximated discrete-time versions (in the sense of the Euler-Maruyama approximation~\cite{kl92}). Roughly speaking, a sampling theorem says that a time-sampled channel is ``close'' to the original channel if the sampling is fine enough, and an approximation theorem says that an approximated channel is ``close'' to the original channel if the approximation is fine enough, both in an information-theoretic sense. Note that, as elaborated in Remark~\ref{sampling-approximation-theorems}, certain version of the approximation theorem boils down to the sampling theorem when there is no memory and feedback in the channel.

Apparently a sampling theorem, whose spirit is in line with the Shannon-Nyquist sampling theorem, is of practical and theoretical value due to the fact it deals with the ``real'' values of the channel output; and, as will be elaborated later, approximation theorems seem to be surprisingly useful in a number of respects despite the fact it only deals with the ``approximated'' values of the channel output: it can certainly provide alternative rigorous tools in translating results from discrete time to continuous time; more importantly, as elaborated in Section~\ref{infinite-bandwidth-revisited}, it lays the foundation for the approximation approach, which gives us the intuition in the point-to-point continuous-time setting, which will further help us to deliver rigorous treatments of multi-user continuous-time Gaussian channels in Section~\ref{multi-user-Gaussian-channels}.

More specifically, in Section~\ref{multi-user-Gaussian-channels}, we derive the capacity regions of a continuous-time white Gaussian multiple access channel (Theorem~\ref{Theorem-MAC}), a continuous-time white Gaussian interference channel (Theorem~\ref{Theorem-IC-Without-Feedback}), and a continuous-time white Gaussian broadcast channel (Theorem~\ref{Theorem-BC-Without-Feedback}). Here, we note that when there is no feedback, as discussed in Remark~\ref{sampling-heuristics}, the results above are ``long known'' in the sense that they are roughly suggested by the conventional sampling approach, or alternatively, the low SNR equivalence in (A). However, to the best of our knowledge, explicit formulations and statements of such results are missing in the literature and their rigourous proofs are non-trivial (for instance, when establishing Theorem~\ref{Theorem-BC-Without-Feedback}, we have to resort to the continuous-time I-MMSE relationship~\cite{gu05}, which has been established only recently). By comparison, the presence of feedback necessitates the use of the approximation approach, which help us to connect relevant results and proofs in discrete time to analyze how feedback affects the capacity regions of families of continuous-time multi-user one-hop Gaussian channels: feedback will increase the capacity regions of some continuous-time Gaussian broadcast channels (Theorem~\ref{Theorem-BC-With-Feedback-2}) and interference channels (Theorem~\ref{Theorem-IC-With-Feedback}), while it will not increase capacity regions of a continuous-time physically degraded Gaussian broadcast channel (Theorem~\ref{Theorem-BC-With-Feedback-1}) and a continuous-time Gaussian multiple access channels (Theorem~\ref{Theorem-MAC}).

\section{Sampling Theorems} \label{sampling-theorems}

A very natural question is whether, similarly for the white Gaussian noise formulation, sampling theorems hold for continuous-time white Gaussian channels under Brownian motion formulation. In this section, we will establish sampling theorems for the channel (\ref{feedback-memory}), which naturally connect such channels with their discrete-time versions obtained by sampling.

Consider the following regularity conditions for channel (\ref{feedback-memory}):
\begin{itemize}
\item[(a)] The solution $\{Y(t)\}$ to the stochastic differential equation (\ref{feedback-memory}) uniquely exists;
\item[(b)]
$$
\PX\left(\int_0^T g^2(t, W_0^t, Y_0^t) dt < \infty \right)=\PX\left(\int_0^T g^2(t, W_0^t, B_0^t) dt < \infty \right)=1;
$$
\item[(c)]
$$
\int_0^T \E[|g(t, W_0^t, Y_0^t)|] dt < \infty.
$$
\end{itemize}
Note that all the three above conditions are rather weak: Condition (a) is necessary for the channel to be meaningful, and Conditions (b) and (c) are very mild integrability assumptions.

Now, for any $n \in \mathbb{N}$, choose {\it time points} $t_{n, 0}, t_{n, 1}, \ldots, t_{n, n} \in \mathbb{R}$ such that
$$
0=t_{n, 0} < t_{n, 1} < \cdots < t_{n, n-1} < t_{n, n}=T,
$$
and let $\Delta_n \triangleq \{t_{n, 0}, t_{n, 1}, \ldots, t_{n, n}\}$. Sampling the channel (\ref{feedback-memory}) over the time interval $[0, T]$ with respect to $\Delta_n$, we obtain its sampled discrete-time version as follows:
\begin{equation}  \label{after-sampling-with-feedback}
Y(t_{n, i})=\int_0^{t_{n, i}} g(s, W_0^s, Y_0^s) ds + B(t_{n, i}), \quad i=0, 1, \ldots, n.
\end{equation}
For any time point sequence $\Delta_n$, we will use $\delta_{\Delta_n}$ to denote its minimal stepsize, namely,
$$
\delta_{\Delta_n} \triangleq \max_{i=1, 2, \dots, n} (t_{n, i}-t_{n, i-1}).
$$
$\Delta_n$ is said to be {\it evenly spaced} if $t_{n, i}-t_{n, i-1}=T/n$ for all feasible $i$, and we will use the shorthand notation $\delta_n$ to denote its stepsize, i.e., $\delta_n \triangleq t_{n, 1}-t_{n, 0}=T/n$. Apparently, evenly spaced time point sequences are natural candidates with respect to which a continuous-time Gaussian channel can be sampled.

We are primarily concerned with the mutual information for the channel (\ref{feedback-memory}), whose standard definition (see, e.g., ~\cite{pi64,ih93}) is given below:
\begin{equation} \label{definition-mutual-information}
I(W_0^T; Y_0^T)=\begin{cases}
\E\left[\log \frac{d \mu_{W Y}}{d \mu_W \times \mu_Y}(W_0^T, Y_0^T)\right], &\mbox{ if } \frac{d \mu_{W Y}}{d \mu_W \times \mu_Y} \mbox{ exists },\\
\infty, &\mbox{ otherwise },
\end{cases}
\end{equation}
where the subscripted $\mu$ denotes the measure induced on $C[0, T]$ or $C[0, T] \times C[0, T]$ by the corresponding stochastic process and $d \mu_{W Y}/d \mu_W \times \mu_Y$ denotes the Radon-Nikodym derivative of $\mu_{WY}$ with respect to $d \mu_W \times \mu_Y$.

Roughly speaking, the following sampling theorem states that for any sequence of ``increasingly refined'' samplings, the mutual information of the sampled discrete-time channel (\ref{after-sampling-with-feedback}) will converge to that of the original channel (\ref{feedback-memory}).
\begin{thm} \label{sampling-theorem-1}
Assume Conditions (a)-(c). Suppose that $\Delta_n \subset \Delta_{n+1}$ for all $n$ and that $\delta_{\Delta_n} \to 0$ as $n$ tends to infinity. Then, we have
$$
\lim_{n \to \infty} I(W_0^T; Y(\Delta_n))=I(W_0^T; Y_0^T),
$$
where $Y(\Delta_n) \triangleq \{Y(t_{n, 0}), Y(t_{n, 1}), \ldots, Y(t_{n, n})\}$.
\end{thm}

\begin{proof}
The proof is rather technical and thereby postponed to Appendix~\ref{proof-sampling-theorem-1}.
\end{proof}

Regarding the assumptions of Theorem~\ref{sampling-theorem-1}, as mentioned before, Conditions (a)-(c) are very weak, but the condition that ``$\Delta_n \subset \Delta_{n+1}$ for all $n$'' is somewhat restrictive, which, in particular, is not satisfied by the set $\{\Delta_n\}$ of all evenly spaced time point sequences. We next show that this condition can be replaced by some extra regularity conditions: The same theorem holds as long as the stepsize of the sampling tends to $0$, which, in particular, is satisfied by the set of all evenly spaced sampling sequences.

Below and hereafter, defining the distance $\|U_0^s-V_0^t\|$ between $U_0^s$ and $V_0^t$ with $0 \leq s \leq t$ as
\begin{equation} \label{UV-distance}
\|U_0^s-V_0^t\| \triangleq \sup_{r \in [0, s]} |U(r)-V(r)|+ \sup_{r \in [s, t]} |U(s)-V(r)|,
\end{equation}
we may assume the following three regularity conditions for the channel (\ref{feedback-memory}):
\begin{itemize}
\item[(d)] \textbf{Uniform Lipschitz condition:} There exists a constant $L > 0$ such that for any $0 \leq s_1, s_2, s_3, t_1, t_2, t_3 \leq T$, any $U_0^T, V_0^T, Y_0^T$ and $Z_0^T$,
$$
|g(s_1, U_0^{s_2}, Y_0^{s_3})-g(t_1, V_0^{t_2}, Z_0^{t_3})| \leq L (|s_1-t_1|+\|U_0^{s_2}-V_0^{t_2}\|+ \|Y_{0}^{s_3}-Z_0^{t_3}\|);
$$

\item[(e)] \textbf{Uniform linear growth condition:} There exists a constant $L > 0$ such that for any $W_0^T$ and any $Y_0^T$,
$$
|g(t, W_0^t, Y_0^t)| \leq L (1+\|W_0^t\|+\|Y_0^t\|),
$$
where
$$
\|W_0^t\| = \sup_{r \in [0, t]} |W(r)|, \quad \|Y_0^t\|= \sup_{r \in [0, t]} |Y(r)|;
$$

\item[(f)] \textbf{Regularity conditions on $W$:} There exists $\varepsilon > 0$ such that
$$
\E[e^{\varepsilon \|W_0^T\|^2}] < \infty,
$$
and for any $K > 0$, there exists $\varepsilon' > 0$ such that
$$
\E [e^{K \sup_{|s-t| \leq \varepsilon'} (W(s)-W(t))^2}] < \infty,
$$
and there exists a constant $L > 0$ such that for any $\varepsilon'' > 0$,
$$
\E[\sup\nolimits_{|s-t| \leq \varepsilon''} (W(s)-W(t))^4] \leq L (\varepsilon'')^2.
$$
\end{itemize}

The following lemma, whose proof is postponed to Appendix~\ref{proof-improved-liptser-1}, says that Conditions (d)-(f) are stronger than Conditions (a)-(c). We however remark that Conditions (d)-(f) are still rather mild assumptions: The uniform Lipschitz condition, uniform linear growth condition and their numerous variants are typical assumptions that can guarantee the existence and uniqueness of the solution to a given stochastic differential equation. In theory, these two conditions are considered mild in the sense there are examples that the corresponding stochastic differential equation may not have solutions at all if these two conditions are not satisfied (see, e.g.,~\cite{mao97}). Note that the third condition is a mild integrability condition; as a matter of fact, for a feedback channel where $W$ is interpreted as the message, this condition is trivially satisfied. All three conditions above will be taken for granted in most practical communication situations: as might be expected, the signals employed in practice will be much better-behaving.
\begin{lem} \label{improved-liptser-1}
Assume Conditions (d)-(f). Then, there exists a unique strong solution of (\ref{feedback-memory}) with initial value $Y(0)=0$. Moreover, there exists $\varepsilon > 0$ such that
\begin{equation} \label{exponential-finiteness-1}
\E [e^{\varepsilon \|Y_0^T\|^2}] < \infty,
\end{equation}
which immediately implies Conditions (b) and (c).
\end{lem}

Roughly speaking, the following sampling theorem states that if the stepsizes of the samplings tend to $0$, the mutual information of the channel (\ref{after-sampling-with-feedback}) will converge to that of the channel (\ref{feedback-memory}). Note that in this theorem, we do not need the assumption that ``$\Delta_n \subset \Delta_{n+1}$ for all $n$'', which is required in Theorem~\ref{sampling-theorem-1}.

\begin{thm} \label{sampling-theorem-2}
Assume Conditions (d)-(f). For any sequence $\{\Delta_n\}$ with $\delta_{\Delta_n} \to 0$ as $n$ tends to infinity, we have
$$
\lim_{n \to \infty} I(W_0^T; Y(\Delta_n))=I(W_0^T; Y_0^T).
$$
\end{thm}

\begin{proof}
The proof is rather technical and lengthy, and thereby postponed to Appendix~\ref{proof-sampling-theorem-2}. We note that, as detailed in Remark~\ref{mmse-1}, the arguments in the proof can be adapted to yield a sampling theorem in estimation theory.
\end{proof}

\section{Approximation Theorems} \label{approximation-theorems}

In this section, we will establish approximation theorems for the channel (\ref{feedback-memory}), which naturally connect such channels with their discrete-time versions obtained by approximation. As elaborated in later sections, the approximation theorem will underpin the approximation approach that will be introduced in Section~\ref{infinite-bandwidth-revisited}.

An application of the Euler-Maruyama approximation~\cite{kl92} with respect to $\Delta_n$ to (\ref{feedback-memory}) will yield a discrete-time sequence $\{Y^{(n)}(t_{n, i}): i=0, 1, \dots, n\}$ and a continuous-time process $\{Y^{(n)}(t): t \in [0, T]\}$, a linear interpolation of $\{Y(t_{n, i})\}$, as follows: Initializing with $Y^{(n)}(0)=0$, we recursively compute, for each $i=0, 1, \dots, n-1$,
\begin{equation} \label{Euler-Maruyama-Sequence}
Y^{(n)}(t_{n, i+1})=Y^{(n)}(t_{n, i})+ \int_{t_{n, i}}^{t_{n, i+1}} g(s, W_0^{t_{n, i}}, Y_0^{(n), t_{n, i}}) ds + B(t_{n, i+1})-B(t_{n, i}),
\end{equation}
\begin{equation} \label{linear-interpolation}
Y^{(n)}(t)=Y^{(n)}(t_{n, i})+\frac{t-t_{n, i}}{t_{n, i+1}-t_{n, i}} (Y^{(n)}(t_{n, i+1})-Y^{(n)}(t_{n, i})), \quad t_{n, i} \leq t \leq t_{n, i+1}.
\end{equation}

We are now ready to prove the following theorem:
\begin{thm} \label{approximation-theorem-1}
Assume Conditions (d)-(f). Then, we have
$$
\lim_{n \to \infty} I(W_0^T; Y^{(n)}(\Delta_n))=I(W_0^T; Y_0^T),
$$
where $Y^{(n)}(\Delta_n) \triangleq \{Y^{(n)}(t_{n, 0}), Y^{(n)}(t_{n, 1}), \ldots, Y^{(n)}(t_{n, n})\}$.
\end{thm}

\begin{proof}
The proof is rather technical and lengthy, and thereby postponed to Appendix~\ref{proof-approximation-theorem-1}. We note that, as detailed in Remark~\ref{mmse-2}, the arguments in the proof can be adapted to yield an approximation theorem in estimation theory.
\end{proof}

For any $\{\Delta_n\}$, let $W^{(n)}(t)$ denote the piecewise linear version of $W_0^T$ with respect to $\Delta_n$; more precisely, for any $i=0, 1, \dots, n$, $W^{(n)}(t_{n,i})=W(t_{n,i})$, and for any $t_{n, i-1} < s < t_{n, i}$ with $s=\lambda t_{n, i-1}+(1-\lambda) t_{n, i}$ where $0 < \lambda < 1$, $W^{(n)}(s)=\lambda W(t_{n, i-1})+(1-\lambda) W(t_{n, i})$. The following modified Euler-Maruyama approximation with respect to $\Delta_n$ applied to the channel (\ref{feedback-memory}) yields a discrete-time sequences $\{Y^{(n)}(t_{n, i}): i=0, 1, \dots, n\}$ and a continuous-time processes $\{Y^{(n)}(t): t \in [0, T]\}$ as follows: Initializing with $Y^{(n)}(0)=0$, we recursively compute, for each $i=0, 1, \dots, n-1$,
\begin{equation} \label{Modified-Euler-Maruyama-Sequence}
Y^{(n)}(t_{n, i+1})=Y^{(n)}(t_{n, i})+ \int_{t_{n, i}}^{t_{n, i+1}} g(s, W_0^{(n), t_{n, i}}, Y_0^{(n), t_{n, i}}) ds + B(t_{n, i+1})-B(t_{n, i}),
\end{equation}
\begin{equation} \label{modified-linear-interpolation}
Y^{(n)}(t)=Y^{(n)}(t_{n, i})+\frac{t-t_{n, i}}{t_{n, i+1}-t_{n, i}} (Y^{(n)}(t_{n, i+1})-Y^{(n)}(t_{n, i})), \quad t_{n, i} \leq t \leq t_{n, i+1}.
\end{equation}
Now, using a parallel argument in the proof of Theorem~\ref{approximation-theorem-1}, we have the following approximation theorem.
\begin{thm} \label{approximation-theorem-2}
Assume Conditions (d)-(f). Then, we have
$$
\lim_{n \to \infty} I(W^{(n)}(\Delta_n); Y^{(n)}(\Delta_n))=I(W_0^T; Y_0^T).
$$
\end{thm}

\begin{rem}
When the channel (\ref{feedback-memory}) is interpreted as a feedback channel, both $W^{(n)}$ and $W$ are precisely $M$. When the channel (\ref{feedback-memory}) is interpreted as a memory channel, Theorem~\ref{approximation-theorem-2} states that the mutual information between its input and output is the limit of that of its approximated input and output (in the sense of the above-mentioned modified Euler-Maruyama approximation).
\end{rem}

Other variants of the Euler-Maruyama approximation can also be applied to the channel to yield variants of the approximation theorem. For instance, under Conditions (d)-(f), for the following variant of the Euler-Maruyama approximation,
\begin{equation}  \label{variant-1}
Y^{(n)}(t_{n, i+1})=Y^{(n)}(t_{n, i})+ \int_{t_{n, i}}^{t_{n, i+1}} g(t_{n, i}, W_0^{t_{n, i}}, Y_0^{(n), t_{n, i}}) ds + B(t_{n, i+1})-B(t_{n, i}),
\end{equation}
a parallel argument as in the proof of Theorem~\ref{approximation-theorem-1} will give the following variant of Theorem~\ref{approximation-theorem-1}:
\begin{thm} \label{approximation-theorem-3}
Assume Conditions (d)-(f). Then, we have
\begin{equation} \label{theorem-variant-1}
\lim_{n \to \infty} I(W_0^T; Y^{(n)}(\Delta_n))=I(W_0^T; Y_0^T).
\end{equation}
\end{thm}
\noindent Moreover, for
\begin{equation}  \label{variant-2}
Y^{(n)}(t_{n, i+1})=Y^{(n)}(t_{n, i})+ \int_{t_{n, i}}^{t_{n, i+1}} g(t_{n, i}, W_0^{(n), t_{n, i}}, Y_0^{(n), t_{n, i}}) ds + B(t_{n, i+1})-B(t_{n, i}),
\end{equation}
we have the following variant of Theorem~\ref{approximation-theorem-2}:
\begin{thm} \label{approximation-theorem-4}
Assume Conditions (d)-(f). Then, we have
\begin{equation} \label{theorem-variant-2}
\lim_{n \to \infty} I(W^{(n)}(\Delta_n); Y^{(n)}(\Delta_n))=I(W_0^T; Y_0^T).
\end{equation}
\end{thm}

Regarding the approximation theorem and its variants, we make the following several remarks.

\begin{rem}  \label{directed-information}
Continuous-time directed information has been defined in~\cite{we13} for continuous-time white Gaussian channels with positively delayed feedback. In this remark, we show that our approximation theorem can be used to give an alternative definition of continuous-time directed information, even for the case that the feedback is instantaneous.

Consider the following continuous-time Gaussian feedback channel:
\begin{equation} \label{g-f}
Y(t)=\int_0^t X(s, M, Y_0^{s}) ds + B(t), \quad t \in [0, T].
\end{equation}
For any $\Delta_n$, we define $\tilde{X}^{(n)}(\cdot)$ as follows: for any $t$ with $t_{n, i} \leq t < t_{n, i+1}$,
$$
\tilde{X}^{(n)}(t)=\sum_{j=0}^{i-1} \int_{t_{n, j}}^{t_{n, j+1}} X(s, M, Y_0^{(n), t_{n, j}}) ds + \int_{t_{n, i}}^{t} X(s, M, Y_0^{(n), t_{n, i}}) ds,
$$
Writing $\tilde{X}(t, M, Y_0^{(n), t})$ as $\tilde{X}^{(n)}(t)$ for simplicity, (\ref{Euler-Maruyama-Sequence}) can be rewritten as
$$
Y^{(n)}(t_{n, i+1})=Y^{(n)}(t_{n, i})+ \tilde{X}^{(n)}(t_{n, i+1})-\tilde{X}^{(n)}(t_{n, i}) + B(t_{n, i+1})-B(t_{n, i}),
$$
for which it can be readily checked that
\begin{equation} \label{input-message}
I(\tilde{X}^{(n)}(\Delta_n) \to Y^{(n)}(\Delta_n))=I(M; Y^{(n)}(\Delta_n)).
\end{equation}
Theorem~\ref{approximation-theorem-1} and the above observation can be used to define continuous-time directed mutual information. To be more precise, the continuous-time directed information from $X_0^T$ to $Y_0^T$ of the channel (\ref{g-f}) can be defined as
\begin{equation}  \label{our-def}
I(X_0^T \to Y_0^T) \triangleq \lim_{n \to \infty} I(\tilde{X}^{(n)}(\Delta_n) \to Y^{(n)}(\Delta_n)).
\end{equation}

Consider the following continuous-time Gaussian channel with possibly delayed feedback:
\begin{equation}
Y(t)=\int_0^t X(s, M, Y_0^{s-D}) ds + B(t), \quad t \in [0, T],
\end{equation}
where $D \geq 0$ denotes the delay of the feedback. In~\cite{we13}, the notion of continuous-time directed information from $X_0^T$ to $Y_0^T$ is defined as follows:
\begin{equation}  \label{Tsachy-Definition}
I_D(X_0^T \to Y_0^T) = \inf_{\Delta_n} \sum_{i=1}^n I(X_{t_{n, 0}}^{t_{n, i}}; Y_{t_{n, i-1}}^{t_{n, i}}|Y_{t_{n, 0}}^{t_{n, i-1}}).
\end{equation}
It is proven that for the case $D > 0$, using this notion, a connection between information theory and estimation theory can be established as follows:
\begin{equation}  \label{Tsachy-Claim}
I_D(X_0^T \to Y_0^T) = \frac{1}{2} \int_0^T \E[(X(t)-\E[X(t)|Y_0^t])^2] dt.
\end{equation}
On the other hand though, it is easy to see that for the case $D=0$, i.e., there is no delay in the feedback as in (\ref{g-f}), the definition in (\ref{Tsachy-Definition}) and the equality
as in (\ref{Tsachy-Claim}) may run into some problems: Consider the extreme scenario and choose $X(t)=-Y(t)$ for any feasible $t$, then clearly the right hand side of (\ref{Tsachy-Claim}) should be equal to $0$.  On the other hand though, for the left hand side, each small interval in (\ref{Tsachy-Definition}) will yield
$$
I_{D=0}(X_{t_{n, 0}}^{t_{n, i+1}}; Y_{t_{n, i}}^{t_{n, i+1}}|Y_{t_{n, 0}}^{t_{n, i}})=I_{D=0}(Y_{t_{n, i}}^{t_{n, i+1}}; Y_{t_{n, i}}^{t_{n, i+1}}|Y_{t_{n, 0}}^{t_{n, i}})=I_{D=0}(Y_{t_{n, i}}^{t_{n, i+1}}; Y_{t_{n, i}}^{t_{n, i+1}}|Y(t_{n, i})),
$$
where for the last equality, we have used the fact that under the assumption that $X(t)=-Y(t)$, $\{Y(t)\}$ is an Ornstein-Uhlenbeck process, which is a Gaussian Markov process. Noting that given $Y(t_i)=y(t_i)$, the Radon-Nikodym derivative $d \mu_{Y_{t_i}^{t_{i+1}} Y_{t_i}^{t_{i+1}}}/d \mu_{Y_{t_i}^{t_{i+1}}} \times \mu_{Y_{t_i}^{t_{i+1}}}$ does not exist, we conclude, by referring to the definition in (\ref{definition-mutual-information}), that
$$
I_{D=0}(Y_{t_{n, i}}^{t_{n, i+1}}; Y_{t_{n, i}}^{t_{n, i+1}}|Y(t_{n, i})) = \infty, \mbox{ and thereby, } I_{D=0}(X_{t_{n, 0}}^{t_{n, i+1}}; Y_{t_{n, i}}^{t_{n, i+1}}|Y_{t_{n, 0}}^{t_{n, i}})=\infty,
$$
which further implies that, $I_{D=0}(X_0^T \to Y_0^T)$, the left-hand side of (\ref{Tsachy-Claim}) at $D=0$ is infinite, a contradiction. On the other hand, be it the case $D > 0$ or $D=0$, with the definition in (\ref{our-def}), Theorem~\ref{approximation-theorem-1} however promises:
$$
I(X_0^T \to Y_0^T) = I(M_0^T; Y_0^T)=\frac{1}{2} \int_0^T \E[(X(t)-\E[X(t)|Y_0^t])^2] dt.
$$
\end{rem}

\begin{rem} \label{sampling-approximation-theorems}
When there is no feedback or memeory, Theorem~\ref{approximation-theorem-1} boils down to Theorem~\ref{sampling-theorem-2}: obviously we will have for any feasible $i$
$$
Y^{(n)}(t_{n, i})=Y(t_{n, i}),
$$
which means that Theorem~\ref{approximation-theorem-1} actually states
$$
\lim_{n \to \infty} I(W^T_0; Y(\Delta_n))=I(W_0^T; Y_0^T),
$$
which is precisely the conclusion of Theorem~\ref{sampling-theorem-2}. And moreover, by Remark~\ref{directed-information}, we also have
$$
\lim_{n \to \infty} I(\tilde{X}^{(n)}(\Delta_n); Y(\Delta_n))=\lim_{n \to \infty} I(\tilde{X}^{(n)}(\Delta_n) \to Y(\Delta_n))=I(W_0^T; Y_0^T).
$$
\end{rem}

\begin{rem}  \label{from-continuous-to-discrete}
In this remark, we briefly discuss the possible applications of our sampling and approximation theorems, both of which we believe will important roles in the long run for further developing continuous-time information theory, particularly for scenarios where feedback and memory are present.

Taking advantage of the pathwise continuity of a Brownian motion, our sampling theorems, Theorems~\ref{sampling-theorem-1} and~\ref{sampling-theorem-2}, naturally connect continuous-time Gaussian memory/feedback channels with their discrete-time counterparts, whose outputs are precisely sampled outputs of the original continuous-time Gaussian channel. In discrete time, the Shannon-McMillan-Breiman theorem provides an effective way to approximate the entropy rate of a stationary ergodic process, and numerical computation and optimization of mutual information of discrete-time channel using the Shannon-McMillan-Breiman theorem and its extensions have been extensively studied (see, e.g.,~\cite{Han-Limit-Theorem, Han-Randomized} and references therein), which suggests our sampling theorems may well serve as a bridge to capitalize on relevant results in discrete time to numerically compute and optimize the mutual information of continuous-time Gaussian channels. In short, despite numerous technical barriers that one needs to overcome, we believe that in the long run the sampling theorems can help us in terms of numerically computing the mutual information and capacity of continuous-time Gaussian channels.

By comparison, our approximation theorems, Theorems~\ref{approximation-theorem-1} and~\ref{approximation-theorem-2}, are  somewhat ``artificial'' in the sense that the outputs of the associated discrete-time channels are only approximated outputs of the original continuous-time channels. Nonetheless, as the Euler-Maruyama approximation of a continuous-time channel yields the form of a discrete-time channel typically takes, our approximation theorems allow a smooth translation from the results and ideas from the discrete-time setting to the continuous-time setting. As a result, the approximation theorems underpin the so-called approximation approach (to be introduced in Section~\ref{infinite-bandwidth-revisited}) and readily yield results for continuous-time Gaussian channels in the multi-user setting, which will be elaborated in the following sections.
\end{rem}

\section{The Approximation Approach}  \label{infinite-bandwidth-revisited}

Consider the following continuous-time white Gaussian channel with feedback
\begin{equation}  \label{i-b}
Y(t)=\int_0^t X(s, M, Y_0^s) ds+ B(t), \quad t \geq 0,
\end{equation}
satisfying the power constraint: there exists $P > 0$ such that for any $T$, with probability $1$
\begin{equation}  \label{power}
\frac{1}{T} \int_0^T X^2(s, M, Y_0^s) ds \leq P.
\end{equation}
As mentioned in Section~\ref{intro}, it is well-known that the capacity of the above channel is $P/2$ (The same result can be established under alternative power constraints; see, e.g., ~\cite{ih93}).

As elaborated in Section~\ref{intro}, when there is no feedback in the channel, i.e., the channel (\ref{i-b}) is actually equivalent to (\ref{white-Gaussian-noise-channel}), and one can ``derive'' the non-feedback capacity heuristically using the conventional sampling approach as in (\ref{white-Gaussian-noise-channel})-(\ref{infinite-bandwidth-capacity}). But this approach is unable to tackle feedback since an application of the Shannon-Nyquist sampling theorem will destroy the temporal causality.

In this section, we use our approximation theorems to give an alternative way to ``derive'' the capacity of (\ref{i-b}), which will be referred to as {\it the approximation approach} in the remainder of the paper. Compared to the sampling approach, the approximation approach can handle feedback due to the fact the Euler-Maruyama approximation preserves temporal causality. Below we briefly explain this new approach, which will be further developed and used, either heuristically or rigorously, in Section~\ref{multi-user-Gaussian-channels}, where multiple users may be involved in a communication system.

For fixed $T > 0$, consider the evenly spaced sequence $\Delta_n$ with stepsize $\delta_n = T/n$. Applying the Euler-Maruyama approximation (\ref{variant-1}) to the channel (\ref{i-b}) over the time window $[0, T]$, we obtain
\begin{equation} \label{revisited}
Y^{(n)}(t_{n, i+1})=Y^{(n)}(t_{n, i})+ \int_{t_{n, i}}^{t_{n, i+1}} X(t_{n, i}, M, Y_0^{(n), t_{n, i}}) ds + B(t_{n, i+1})-B(t_{n, i}).
\end{equation}
By Theorem~\ref{approximation-theorem-3}, we have
\begin{equation}  \label{another-heuristic}
I(M; Y_0^T) = \lim_{n \to \infty} I(M; Y^{(n)}(\Delta_n)).
\end{equation}
Our strategy is to ``establish'' the capacity for the discrete-time channel (\ref{revisited}) first, and then the capacity for the continuous-time channel (\ref{i-b}) using the ``closeness'' between the two channels, as claimed by approximation theorems.

For the converse part, we first note that
\begin{align*}
I(M; Y^{(n)}(\Delta_n)) &=\sum_{i=1}^{n} h(Y^{(n)}(t_{n, i})-Y^{(n)}(t_{n, i-1})|Y^{(n),t_{n, i-2}}_{t_{n, 0}})-\sum_{i=1}^{n} h(B(t_{n, i})-B(t_{n, i-1}))\\
& \leq \sum_{i=1}^{n} h(Y^{(n)}(t_{n, i})-Y^{(n)}(t_{n, i-1}))-\sum_{i=1}^{n} h(B(t_{n, i})-B(t_{n, i-1})).
\end{align*}
It then follows from the fact
\begin{align*}
Var(Y^{(n)}(t_{n, i})-Y^{(n)}(t_{n, i-1})) &= \E[(Y^{(n)}(t_{n, i})-Y^{(n)}(t_{n, i-1}))^2]\\
&=\E[\delta_n^2 (X^{(n)}(t_{n, i-1}))^2]+\E[(B(t_{n, i})-B(t_{n, i-1}))^2]\\
&=\E[\delta_n^2 (X^{(n)}(t_{n, i-1}))^2]+\delta_n,
\end{align*}
that
\begin{align}
\label{ineq-1} I(M; Y_{\Delta_n}^{(n)}) & \leq \frac{1}{2} \sum_{i=0}^{n} \log (1+\delta_n \EX[(X^{(n)}(t_{n, i}))^2]) \\
\label{ineq-2} & \leq \frac{1}{2} \sum_{i=0}^{n} \delta_n \EX[(X^{(n)}(t_{n, i}))^2],
\end{align}
which, by (\ref{another-heuristic}), immediately yields
\begin{equation} \label{ineq-3}
I(M; Y_0^T) \leq \frac{1}{2} \int_0^T \EX[X^2(s)] ds \leq \frac{P T}{2},
\end{equation}
which establishes the converse part.

For the availability part, note that if we assume all $X^{(n)}(t_{n, i-1})$ are independent of the Brownian motion $B$ with $\E[(X^{(n)}(t_{n, i-1}))^2]=P$, then the inequalities in (\ref{ineq-1}) and (\ref{ineq-2}) will become equalities. The part then follows from a usual random coding argument with codes generated by the distribution of $X^{(n)}$ (or more precisely, a linear interpolation of $X^{(n)}$). It is clear that as $n$ tends to infinity, the process $X^{(n)}$ behaves more and more like a white Gaussian process. This observation echoes Theorem $6.4.1$ in~\cite{ih93}, whose proof rigorously shows that an Ornstein-Uhlenbeck process that oscillates ``extremely'' fast will achieve the capacity of (\ref{i-b}).

Roughly speaking, similar to the conventional sampling approach, the above approximation approach establishes a continuous-time Gaussian feedback channel as the limit of the associated discrete-time channels as the SNR for each channel use shrink to zero proportionately (note that in the above arguments, the SNR for each channel use is $\delta_n$). In other words, we have strengthened the low SNR equivalence in (A) as follows:
\begin{quote}
{\em a continuous-time infinite-bandwidth Gaussian channel {\bf \em with feedback} is ``equivalent'' to a discrete-time Gaussian channel with feedback at low SNR.} \hfill (B)
\end{quote}
We remark however that for the purpose of deriving the capacity of (\ref{i-b}) though, the approximation approach, like the conventional sampling approach, is heuristic in nature: Theorem~\ref{approximation-theorem-1} does require Conditions (d)-(f), which are much stronger than the power constraint (\ref{power}). Nevertheless, this approach is of fundamental importance to our treatment of continuous-time Gaussian channels: as elaborated in Section~\ref{multi-user-Gaussian-channels}, not only can it channel the ideas and techniques in discrete time to rigorously establish new results in continuous time, more importantly, it can also provide insights and intuition for our rigorous treatments where we will employ established tools and develop new tools in stochastic calculus.

\section{Continuous-Time Multi-User Gaussian Channels} \label{multi-user-Gaussian-channels}

Extending Shannon's fundamental theorems on point-to-point communication channels to general networks with multiple sources and destinations, network information theory aims to establish the fundamental limits on information flows in networks and the optimal coding schemes that achieve these limits. The vast majority of researches on network information theory to date have been focusing on networks in discrete time. In a way, this phenomenon can find its source from Shannon's original treatment of continuous-time point-to-point channels, where such channels were examined through their associated discrete-time versions. This insightful viewpoint has exerted major influences on the bulk of the related literature on continuous-time Gaussian channels, oftentimes prompting a model shift from the continuous-time setting to the discrete-time one right from the beginning of a research attempt.

The primary focus of this section is to illustrate the possible applications of the approximation approach: 1) Guided by this approach, we will rigorously derive the capacity regions of families of continuous-time multi-user one-hop white Gaussian channels, including continuous-time multi-user white Gaussian multiple access channels (MACs) and broadcast channels (BCs). To deliver the rigourous proofs of our results, we will directly work within the continuous-time setting, employing established tools and developing new tools (see Theorems~\ref{independent-OUs} and~\ref{Xianming-Lemma}) in stochastic calculus to complement the approximation approach. 2) We can also rigorously apply this approach to examine how feedback affects the capacity regions of the above-mentioned channels via translations of results and techniques in discrete time. It turns out that some results can be translated from the discrete-time setting to the continuous-time setting, such as that feedback increases the capacity region of Gaussian BCs, and that feedback does not increase the capacity region of physically degraded BCs. Nevertheless, there is a seeming ``exception'': as opposed to discrete-time Gaussian MACs, feedback does not increase the capacity region some of continuous-time Gaussian MACs, which, somewhat surprisingly, can also be explained by the approximation approach as well.

Below, we summarize the results in this section. To put our results into a relevant context, we will first list some related results in discrete time, and for obvious reasons, we can only list those that are most relevant to ours.

{\bf Gaussian MACs.} When there is no feedback, the capacity region of a discrete-time memoryless MAC is relatively better understood: a single-letter characterization has been established by Ahlswede~\cite{ah73} and the capacity region of a Gaussian MAC was explicitly derived in Wyner~\cite{wy74} and Cover~\cite{co75}. On the other hand, the capacity region of MACs with feedback still demands more complete understanding, despite several decades of great effort by many authors: Cover and Leung~\cite{co81} derived an achievable region for a memoryless MAC with feedback. In~\cite{wi82}, Willems showed that Cover and Leung's region is optimal for a class of memoryless MACs with feedback where one of the inputs is a deterministic function of the output and the other input. More recently, Bross and Lapidoth~\cite{br05} improved Cover and Leung's region, and Wu {\em et al.}~\cite{wu} extended Cover and Leung's region for the case where non-causal state information is available at both senders. An interesting result has been obtained by Ozarow~\cite{oz84}, who derived the capacity region of a memoryless Gaussian MAC with two users via a modification of the Schalkwijk-Kailath scheme~\cite{sc66}; moreover, Ozarow's result showed that in general, the capacity region for a discrete memoryless MAC is increased by feedback.

In Section~\ref{MAC}, guided by the approximation approach, we first establish Lemma~\ref{independent-OUs}, a key lemma which roughly says that ``other users can be simply treated as noises'', and we then employ established tools from stochastic calculus to derive the capacity region of a continuous-time white Gaussian MAC with $m$ senders and with/without feedback. It turns out that for such a channel, feedback does not increase the capacity region, which, at first sight, may seem at odds with the aforementioned Ozarow's result and the conclusion of our approximation theorems. This however can be roughly explained by the well-known fact that ``$>$'' may become ``$=$'' when taking the limit (indeed, $a_n > b_n$ does not necessarily imply $\lim_{n \to \infty} a_n > \lim_{n \to \infty} b_n$); see Remark~\ref{approximation-heuristics} for a more detailed explanation. \hfill $\blacksquare$

{\bf Gaussian ICs.} The capacity regions of discrete-time Gaussian ICs are largely unknown except for certain special scenarios: The capacity region of Gaussian ICs with strong interference has been established in Sato~\cite{sa78}, Han and Kobayashi~\cite{ha81}. The sum-capacity of Gaussian ICs with weak interference has been simultaneously derived in~\cite{sh09, an09, mo09}. The half-bit theorem on the tightness of the Han-Kobayashi bound~\cite{ha81} was proven in~\cite{et08}. The approximation of the Gaussian IC by the $q$-ary expansion deterministic channel was first proposed by Avestimehr, Diggavi, and Tse~\cite{av11}. Outer and inner bounds on the feedback capacity region of Gaussian interference channels are established by Suh and Tse~\cite{su11}. Note that all the above-mentioned work deal with ICs with two pairs of senders and receivers. For more than two user pairs, special classes of Gaussian ICs have been examined using the scheme of interference alignment; see an extensive list of references in~\cite{el11}.

In Section~\ref{IC}, using a similar approach that we developed for continuous-time Gaussian MACs, we derive the capacity region of a continuous-time white Gaussian IC with $m$ pairs of senders and receivers and without feedback. And we also use a translated version of the argument in~\cite{su11} and the approximation approach to show that feedback does increase the capacity region of certain continuous-time white Gaussian IC.\hfill $\blacksquare$

{\bf Gaussian BCs.} The capacity regions of discrete-time Gaussian BCs without feedback are well known~\cite{co72, be73}. And it has been shown by El Gamal~\cite{el81} that feedback cannot increase the capacity region of a physically degraded Gaussian BC. On the other hand, it was shown by Ozarow and Leung~\cite{ozle84} that feedback can increase the capacity of stochastically degraded Gaussian BCs, whose capacity regions are far less understood.

In Section~\ref{BC}, we first establish a continuous-time version of entropy power inequality (Theorem~\ref{Xianming-Lemma}) and then derive the capacity region of a continuous-time Gaussian non-feedback BC with $m$ receivers. Employing the approximation approach, we use a modified argument in~\cite{el81} to show that feedback does not increase the capacity region of a physically degraded continuous-time Gaussian BC, and on the other hand, a translated version of the argument in~\cite{ozle84} to show that feedback does increase the capacity region of certain continuous-time Gaussian BC. \hfill $\blacksquare$

Here we remark that the above-mentioned capacity results for the non-feedback case (Theorems~\ref{Theorem-MAC} (for the non-feedback case) and~\ref{Theorem-BC-Without-Feedback}) are ``long known'' in the sense that they are ``predicted'' by the conventional sampling approach and their proofs follow from the usual framework. On the other hand though, explicit formulations and statements of these results and their rigorous and complete proofs, to the best of our knowledge, do not exist in the literature. The reason, we believe, is that there are a number of technical difficulties that one has to overcome to prove such results: Lemma~\ref{independent-OUs} and Lemma~\ref{Xianming-Lemma} (which is based on the I-MMSE relationship that has only been established in~\cite{gu05}) are newly developed in this work and their proofs are non-trivial.

In contrast, the approximation approach can be applied to continuous-time Gaussian feedback channels, either heuristically or rigorously. More specifically, it can be heuristically applied to ``explain'' Theorems~\ref{Theorem-MAC} (on feedback capacity) and~\ref{Theorem-BC-Without-Feedback}) and give us intuition (as elaborated in Remark~\ref{approximation-heuristics}, it helps ``predict'' the optimal channel input distribution, which has been made rigorous in Lemma~\ref{independent-OUs}), and it can also be applied to establish Theorems~\ref{Theorem-BC-With-Feedback-1} and~\ref{Theorem-BC-With-Feedback-2}.

\subsection{Gaussian MACs}\label{MAC}

Consider a continuous-time white Gaussian MAC with $m$ users, which can be characterized by
\begin{equation}  \label{Equation-MAC}
Y(t)=\int_0^t X_1(s, M_1, Y_0^s) ds+\int_0^t X_2(s, M_2, Y_0^s) ds + \cdots +\int_0^t X_m(s, M_m, Y_0^s) ds+B(t), \quad t \geq 0,
\end{equation}
where $X_i$ is the channel input from sender $i$, which depends on $M_i$, the message sent from sender $i$, which is independent of all messages from other senders, and possibly on the feedback $Y_0^{s}$, the channel output up to time $s$.

For $T, R_1, \ldots, R_m, P_1, \ldots, P_m > 0$, a $(T, (e^{T R_1}, \ldots, e^{T R_m}), (P_1, \ldots, P_m))$-code for the MAC (\ref{Equation-MAC}) consists of $m$ sets of integers $\mathcal{M}_i=\{1, 2, \ldots, e^{T R_i}\}$, the {\it message alphabet} for user $i$, $i=1, 2, \ldots, m$, and $m$ {\it encoding functions}, $X_i: \mathcal{M}_i \rightarrow C[0, T]$, which satisfy the following power constraint: for any $i=1, 2, \ldots, m$, with probability $1$,
\begin{equation}  \label{PowerConstraint-MAC}
\frac{1}{T} \int_0^T X^2_i(s, M_i, Y_0^s) ds \leq P_i,
\end{equation}
and a {\it decoding function},
$$
g: C[0, T] \rightarrow \mathcal{M}_1 \times \mathcal{M}_2 \times \cdots \times \mathcal{M}_m.
$$
The average probability of error for the above code is defined as
$$
\hspace{-1cm} P_e^{(T)}=\frac{1}{e^{T (\sum_{i=1}^m R_i)}} \sum_{(M_1, M_2, \ldots, M_m) \in \mathcal{M}_1 \times \mathcal{M}_2 \times \cdots \times \mathcal{M}_m} P\{g(Y_0^T) \neq (M_1, M_2, \ldots, M_m)~|~(M_1, M_2, \ldots, M_m) \mbox{ sent}\}.
$$
A rate tuple $(R_1, R_2, \ldots, R_m)$ is said to be {\bf achievable} for the MAC if there exists a sequence of $(T, (e^{T R_1}, \ldots, e^{T R_m}), (P_1, \ldots, P_m))$-codes with $P_e^{(T)} \rightarrow 0$ as $T \rightarrow \infty$. The {\bf capacity region} of the MAC is the closure of the set of all the achievable $(R_1, R_2, \ldots, R_m)$ rate tuples.

The following theorem, whose proof is postponed to Appendix~\ref{proof-Theorem-MAC}, gives an explicit characterization of the capacity region of (\ref{Equation-MAC}).
\begin{thm}  \label{Theorem-MAC}
Whether there is feedback or not, the capacity region of the continuous-time white Gaussian MAC (\ref{Equation-MAC}) is
$$
\{(R_1, R_2, \ldots, R_m) \in \mathbb{R}_+^m: R_i \leq P_i/2, \quad i=1, 2, \ldots, m\}.
$$
\end{thm}

\begin{rem} \label{sampling-heuristics}
When there is no feedback, Theorem~\ref{Theorem-MAC} can be heuristically explained using the sampling approach as in (\ref{sampled-white-Gaussian-noise-channel})-(\ref{infinite-bandwidth-capacity}) (this heuristical approach in this example should be well-known; see, e.g., Exercise $15.26$ in~\cite{co2006}).

For simplicity only, we consider the following continuous-time white Gaussian multiple access channel with two senders:
\begin{equation} \label{example-mac}
Y(t)=X_1(t)+X_2(t)+Z(t), \quad t \in \mathbb{R},
\end{equation}
where $X_i$, $i=1, 2$, is the input from the $i$-th user with average power limit $P_i$. Similarly as before, consider its associated discrete-time version corresponding to bandwidth limit $\omega$:
$$
Y_n=X_{1, n}^{(\omega)}+X_{2, n}^{(\omega)}+Z_n^{(\omega)}, \quad n \in \mathbb{Z}.
$$
Then, it is well known~\cite{el11} that the outer bound on the capacity region can be computed as
$$
\left\{(R_1, R_2) \in \mathbb{R}_+^2: R_1 \leq W \log \left(1+\frac{P_1}{2 \omega} \right), R_2 \leq W \log \left(1+\frac{P_2}{2 \omega} \right)\right\},
$$
and the inner bound as
$$
\hspace{-0.6cm} \left\{(R_1, R_2) \in \mathbb{R}_+^2: R_1 \leq \omega \log \left(1+\frac{P_1}{2 \omega} \right), R_2 \leq \omega \log \left(1+\frac{P_2}{2 \omega} \right), R_1+R_2 \leq \omega \log \left(1+\frac{P_1+P_2}{2 \omega} \right)\right\}.
$$
(Here, it is known~\cite{wy74, co75} that the outer bound can be tightened to coincide with the inner bound, which, however, is not needed for this example.) It is easy to verify that the two bounds also collapse into the same region as $\omega$ tends to infinity:
$$
\left\{(R_1, R_2) \in \mathbb{R}_+^2: R_1 \leq P_1/2, R_2 \leq P_2/2\right\},
$$
which is ``expected'' to be the capacity region of (\ref{example-mac}); or alternatively, one can apply the low SNR equivalence in (A) and take the limit as $P$ tends to $0$, reaching the same conclusion. Note that similar arguments hold for more than two senders as well through a parallel extension.
\end{rem}

\begin{rem} \label{approximation-heuristics}
When the feedback is present in the channel, the approximation approach, rather than the conventional sampling approach, is necessary to explain Theorem~\ref{Theorem-MAC}.

Again, for simplicity only, we consider the following continuous-time Gaussian MAC with two senders:
\begin{equation} \label{heur-1}
Y(t)=\int_0^t X_1(s, M_1, Y_0^s) ds+\int_0^t X_2(s, M_2, Y_0^s) ds + B(t), \quad t \geq 0,
\end{equation}
with the power constraints: there exist $P_1, P_2 > 0$ such that for all $T$,
\begin{equation}  \label{powers}
\int_0^T X_1^2(s, M_1, Y_0^s) ds \leq P_1 T, \quad \int_0^T X_2^2(s, M_2, Y_0^s) ds \leq P_2 T.
\end{equation}
Applying the Euler-Maruyama approximation to the above channel over the time window $[0, T]$ with respect to the evenly spaced $\Delta_n$ with $\delta_n=T/n$, we obtain
$$
\hspace{-1cm} Y^{(n)}(t_{n, i+1})=Y^{(n)}(t_{n, i})+ \int_{t_{n, i}}^{t_{n, i+1}} X_1(s, M_{1}, Y_0^{(n), t_{n, i}}) ds + \int_{t_{n, i}}^{t_{n, i+1}} X_2(s, M_{2}, Y_0^{(n), t_{n, i}}) ds + B(t_{n, i+1})-B(t_{n, i}).
$$
Now, straightforward computations and a usual concavity argument then yields that for large $n$,
{\small \begin{align*}
I(M_1; Y^{(n)}(\Delta_n)|M_2) & = h(Y^{(n)}(\Delta_n)|M_2)-h(Y^{(n)}(\Delta_n)|M_1, M_2) \\
&=\sum_{i=1}^{n} h(Y^{(n)}(t_{n, i})|Y_{t_{n, 0}}^{(n), t_{n, i-1}}, M_2)-\sum_{i=1}^n h(Y^{(n)}(t_{n, i})|Y_{t_{n, 0}}^{(n), t_{n, i-1}}, M_1, M_2)\\
& \leq \sum_{i=1}^{n} \frac{1}{2} \log \left(\E\left(\int_{t_{n, i-1}}^{t_{n, i}} X_1(s, M_1, Y_0^{(n), t_{n, i-1}}) ds\right)^2+\delta_n \right) - \frac{1}{2} \log (\delta_n)\\
& \leq \sum_{i=1}^{n} \frac{1}{2} \log \left(\left(\int_{t_{n, i-1}}^{t_{n, i}} \E X_1(s, M_1, Y_0^{(n), t_{n, i-1}})^2 ds\right) \delta_n+\delta_n \right) - \frac{1}{2} \log (\delta_n)\\
& = \sum_{i=1}^{n} \frac{1}{2} \log \left(\int_{t_{n, i-1}}^{t_{n, i}} \E X_1(s, M_1, Y_0^{(n), t_{n, i-1}})^2 ds+1 \right)\\
& \leq \sum_{i=1}^{n} \sum_{i=1}^{n} \frac{1}{2} \int_{t_{n, i-1}}^{t_{n, i}} \E X_1(s, M_1, Y_0^{(n), t_{n, i-1}})^2 ds\\
& \approx \sum_{i=1}^{n} \sum_{i=1}^{n} \frac{1}{2} \int_{t_{n, i-1}}^{t_{n, i}} \E X_1(s, M_1, Y_0^{s})^2 ds\\
& \leq \frac{P_1 T}{2}.
\end{align*}}
A completely parallel argument will yield that
$$
I(M_2; Y^{(n)}(\Delta_n)|M_1) \leq \frac{P_2 T}{2}.
$$
It then follows from Theorem~\ref{approximation-theorem-1} the region below give an outer bound of the capacity region:
\begin{equation} \label{heur-2}
\{(R_1, R_2): 0 \leq R_1 \leq P_1/2, \quad 0 \leq R_2 \leq P_2/2\}.
\end{equation}

To see that this outer bound can be achieved, set $X_1(s), X_2(s)$, $t_{n, i} \leq s \leq t_{n, i+1}$, in (\ref{heur-1}) to be independent Gaussian random variables with variances $P_1$, $P_2$, respectively. Then, one verifies that for large $n$,
\begin{align}
I(M_1; Y(\Delta_n)) & = h(Y(\Delta_n))-h(Y(\Delta_n)|M_1) \nonumber \\
&=\sum_{i=1}^{n} h(Y(t_{n, i}))-\sum_{i=1}^n h(Y(t_{n, i})|M_1) \nonumber\\
&=\sum_{i=1}^{n} \frac{1}{2} \log (P_1 \delta_n^2 + P_2 \delta_n^2+ \delta_n)-\frac{1}{2} \log (P_2 \delta_n^2 + \delta_n) \label{ignore-1}\\
&=\sum_{i=1}^{n} \frac{1}{2} \log \left(1+ \frac{P_1 \delta_n^2}{P_2 \delta_n^2+\delta_n}\right) \label{ignore-2}\\
&\approx P_1 T/2, \label{new-low-SNR-equivalence}
\end{align}
where we have used the fact that $\delta_n$ is ``close'' to $0$ for large enough $n$ for (\ref{new-low-SNR-equivalence}), which, parallelly as in Remark~\ref{sampling-heuristics}, can be alternatively explained by taking $P_1$ to $0$ and then applying the low SNR equivalence in (B). With a similar argument, one can prove that
$$
I(M_2; Y(\Delta_n)) \approx P_2 T/2.
$$
It then follows that the outer bound in (\ref{heur-2}) can be achieved.

Here we remark that similarly as in Section~\ref{infinite-bandwidth-revisited}, for $n$ large enouch, the constructed processes $X_1$ and $X_2$ behave like ``fast-oscillating'' Ornstein-Uhlenbeck processes, and moreover, from (\ref{ignore-1}) and (\ref{ignore-2}), one can tell that for one user to achieve the maximum transmission rate, the other user can simply be ignored. Predicting the optimal channel input, these facts echo Remark~\ref{xianming-explanation} and give another explanation to Lemma~\ref{independent-OUs}, a key lemma in our rigourous proof of Theorem~\ref{Theorem-MAC} in Appendix~\ref{proof-Theorem-MAC}.
\end{rem}

\subsection{Gaussian ICs} \label{IC}

Consider the following continuous-time white Gaussian interference channel having no feedback and with $m$ pairs of senders and receivers: for $i=1, 2, \ldots, m$,
\begin{align}  \label{Equation-IC}
Y_i(t) =a_{i1} \int_0^t X_1(s, M_1) ds + a_{i2} \int_0^t X_2(s, M_2) ds + \cdots + a_{im} \int_0^t X_m(s, M_m) ds+ B_i(t), \quad t \geq 0,
\end{align}
where $X_i$ is the channel input from sender $i$, which depends on $M_i$, the message sent from sender $i$, which is independent of all messages from other senders, and $a_{ij} \in \mathbb{R}$, $i, j=1, 2, \ldots, m$, is the channel gain from sender $j$ to receiver $i$, all $B_i(t)$ are (possibly correlated) standard Brownian motions.

For $T, R_1, \ldots, R_m, P_1, \ldots, P_m > 0$, a $(T, (e^{T R_1}, \ldots, e^{T R_m}), (P_1, \ldots, P_m))$-code for the IC (\ref{Equation-IC}) consists of $m$ sets of integers $\mathcal{M}_i=\{1, 2, \ldots, e^{T R_i}\}$, the {\it message alphabet} for user $i$, $i=1, 2, \ldots, m$, and $m$ {\it encoding functions}, $X_i: \mathcal{M}_i \rightarrow C[0, T]$ satisfying the following power constraint: for any $i=1, 2, \ldots, m$, with probability $1$,
\begin{equation} \label{PowerConstraint-IC}
\frac{1}{T} \int_0^T X_i^2(s, M_i) ds \leq P_i,
\end{equation}
and $m$ {\it decoding functions}, $g_i: C[0, T] \rightarrow \mathcal{M}_i$, $i=1, 2, \ldots, m$.

The average probability of error for the $(T, (e^{T R_1}, \ldots, e^{T R_m}), (P_1, \ldots, P_m))$-code is defined as
$$
\hspace{-1cm} P_e^{(T)}=\frac{1}{e^{T(\sum_{i=1}^m R_i)}} \sum_{(M_1, M_2, \ldots, M_m) \in \mathcal{M}_1 \times \mathcal{M}_2 \times \cdots \times \mathcal{M}_m} P\{g_i(Y_{i, 0}^T) \neq M_i, i=1, 2, \ldots, m~|~(M_1, M_2, \ldots, M_m) \mbox{ sent}\}.
$$
A rate tuple $(R_1, R_2, \ldots, R_m)$ is said to be {\bf achievable} for the IC if there exists a sequence of $(T, (e^{T R_1}, \ldots, e^{T R_m}),(P_1, \ldots, P_m))$-codes with $P_e^{(T)} \rightarrow 0$ as $T \rightarrow \infty$. The {\bf capacity region} of the IC is the closure of the set of all the achievable $(R_1, R_2, \ldots, R_m)$ rate tuples.

The following theorem explicitly characterizes the capacity region of the above IC, whose proof has been postponed to Appendix~\ref{proof-Theorem-IC-Without-Feedback}.
\begin{thm} \label{Theorem-IC-Without-Feedback}
The capacity region of the continuous-time white Gaussian IC (\ref{Equation-IC}) is
$$
\{(R_1, R_2, \ldots, R_m) \in \mathbb{R}_+^m: R_i \leq a_{ii}^2 P_i/2, \quad i=1, 2, \ldots, m\}.
$$
\end{thm}

\begin{rem}
Theorem~\ref{Theorem-IC-Without-Feedback} can be heuristically derived using a similar argument employing the approximation approach as in Remark~\ref{approximation-heuristics}.
\end{rem}

With the explicit non-feedback capacity region stated in Theorem~\ref{Theorem-IC-Without-Feedback}, we are now ready to use the approximation approach to analyze the effects of feedback on continuous-time Gaussian ICs.

The following theorem says that feedback does help continuous-time Gaussian ICs, whose proof uses a translated version of the argument in~\cite{su11} coupled with the approximation approach as in Section~\ref{infinite-bandwidth-revisited}, and so we only provide a sketch of the proof.
\begin{thm} \label{Theorem-IC-With-Feedback}
Feedback strictly increases the capacity region of certain continuous-time Gaussian interference channel.
\end{thm}

\begin{proof}
Consider the following symmetric continuous-time Gaussian interference channel with two pairs of senders and receivers:
$$
Y_1(t)=\sqrt{snr} \int_0^t X_1(s) ds + \sqrt{inr} \int_0^t X_2(s) ds + B_1(t),
$$
$$
Y_2(t)=\sqrt{inr} \int_0^t X_1(s) ds + \sqrt{snr} \int_0^t X_2(s) ds + B_2(t),
$$
where $snr, inr$ denote the signal-to-noise, interference-to-noise ratios, respectively, $B_1(t), B_2(t)$ are independent standard Brownian motions, and the average power of $X_1, X_2$ are assumed to be $1$.

Following~\cite{su11}, we consider the following coding scheme over two stages, each of length $T_0$.
In the first stage, transmitters $1$ and $2$ send codewords $X_{1, 0}^{T_0}$ and $X_{2, 0}^{T_0}$ with rates $R_1$ and $R_2$, respectively. In the second stage, using feedback, transmitters $1$ and $2$ decode $X_{2, 0}^{T_0}$ and $X_{1, 0}^{T_0}$, respectively. This can be decoded if
$$
R_1, R_2 \leq \frac{inr}{2}.
$$
Then, transmitters $1$ and $2$ send $X_{1, T_0}^{2 T_0}$ and $X_{2, T_0}^{2 T_0}$, respectively such that for any $0 \leq t \leq T_0$,
$$
X_1(T_0+t)=X_2(t), \quad X_2(T_0+t)=-X_1(t).
$$
Then during the two stages, receiver $1$ receives
$$
Y_1(t)=\sqrt{snr} \int_0^t X_1(s) ds+\sqrt{inr} \int_0^t X_2(s) ds + B_1(t), \quad 0 \leq t \leq T_0,
$$
and
$$
Y_1(T_0+t)=\sqrt{snr} \int_0^{T_0+t} X_1(s) ds + \sqrt{inr} \int_0^{T_0+t} X_2(s) ds + B_1(T_0+t), \quad  0 \leq t \leq T_0,
$$
which immediately gives rise to
\begin{align*}
Y_1(T_0+t)-Y_1(T_0) &= \sqrt{snr} \int_{T_0}^{T_0+t} X_1(s) ds + \sqrt{inr} \int_{T_0}^{T_0+t} X_2(s) ds + B_1(T_0+t)-B_1(T_0)\\
                &= \sqrt{snr} \int_0^t X_2(s) ds - \sqrt{inr} \int_0^t X_1(s) ds + B_1(T_0+t)-B_1(T_0).
\end{align*}
We then have that for any $0 \leq t \leq T_0$,
$$
\hspace{-1.5cm} \sqrt{snr} Y_1(t)-\sqrt{inr} (Y_1(T_0+t)-Y_1(T_0))=(snr+inr) \int_0^t X_1(s) ds + \sqrt{snr} B_1(t)-\sqrt{inr} (B_1(T_0+t)-B(T_0)),
$$
which means the codeword $X_{1, 0}^{T_0}$ can be decoded at the second stage if
$$
R_1 \leq \frac{snr+inr}{2}.
$$
A completely parallel argument yields that the codeword $X_{2, 0}^{T_0}$ can be decoded at the second stages if
$$
R_2 \leq \frac{snr+inr}{2}.
$$
All in all, after the two stages, the two codewords $X_{1, 0}^{T_0}$ and $X_{2, 0}^{T_0}$ can be decoded as long as
$$
R_1, R_2 \leq \frac{inr}{2};
$$
in other words, coding rate $(\frac{inr}{2}, \frac{inr}{2})$ is achievable, which, if assuming $inr > snr$, will imply that feedback strictly increases the capacity region.
\end{proof}

\subsection{Gaussian BCs} \label{BC}

In this section, we consider a continuous-time white Gaussian BC with $m$ receivers, which is characterized by: for $i=1, 2, \ldots, m$,
\begin{equation} \label{Equation-BC}
Y_i(t) = \sqrt{snr_i} \int_0^t X(s, M_1, M_2, \ldots, M_m) ds + B_i(t), \quad t \geq 0,
\end{equation}
where $X$ is the channel input, which depends on $M_i$, the message sent from sender $i$, which is uniformly distributed over a finite alphabet $\mathcal{M}_i$ and independent of all messages from other senders, $snr_i$ is the signal-to-noise ratio in the channel for user $i$, $B_i(t)$ are (possibly correlated) standard Brownian motions.

For $T, R_1, R_2, \ldots, R_m, P> 0$, a $(T, (e^{T R_1}, \ldots, e^{T R_m}), P)$-code for the BC (\ref{Equation-BC}) consists of $m$ set of integers $\mathcal{M}_i=\{1, 2, \ldots, e^{T R_i}\}$, the {\it message set} for receiver $i$, $i=1, 2, \ldots, m$,  and an {\it encoding function}, $X: \mathcal{M}_1 \times \mathcal{M}_2 \times \cdots \times \mathcal{M}_m \rightarrow C[0, T]$, which satisfies the following power constraint: with probability $1$,
\begin{equation} \label{PowerConstraint-BC}
\frac{1}{T} \int_0^T X^2(s, M_1, M_2, \ldots, M_m) ds \leq P,
\end{equation}
and $m$ {\it decoding functions}, $g_i: C[0, T] \rightarrow \mathcal{M}_i$, $i=1, 2, \ldots, m$.

The average probability of error for the $(T, (e^{T R_1}, e^{T R_2}, \ldots, e^{T R_m}), P)$-code is defined as
$$
\hspace{-1cm} P_e^{(T)}=\frac{1}{e^{T(\sum_{i=1}^m R_i)}} \sum_{(M_1, M_2, \ldots, M_m) \in \mathcal{M}_1 \times \mathcal{M}_2 \times \cdots \times \mathcal{M}_m} P\{g_i(Y_0^T) \neq M_i, i=1, 2, \ldots, m~|~(M_1, M_2, \ldots, M_m) \mbox{ sent}\}.
$$
A rate tuple $(R_1, R_2, \ldots, R_m)$ is said to be {\bf achievable} for the BC if there exists a sequence of $(T, (e^{T R_1}, e^{T R_2}, \ldots, e^{T R_m}), P)$-codes with $P_e^{(T)} \rightarrow 0$ as $T \rightarrow \infty$. The {\bf capacity region} of the BC is the closure of the set of all the achievable $(R_1, R_2, \ldots, R_m)$ rate tuples.

The following theorem explicitly characterizes the capacity region of the above BC, whose proof is postponed to Appendix~\ref{proof-Theorem-BC-Without-Feedback}.
\begin{thm} \label{Theorem-BC-Without-Feedback}
The capacity region of the continuous-time white Gaussian BC (\ref{Equation-BC}) is
$$
\left\{(R_1, R_2, \ldots, R_m) \in \mathbb{R}_+^m: \frac{R_1}{snr_1}+\frac{R_2}{snr_2}+\cdots+\frac{R_m}{snr_m} \leq \frac{P}{2}\right\}.
$$
\end{thm}

\begin{rem}
Theorem~\ref{Theorem-BC-Without-Feedback} can be heuristically derived using a similar argument employing the approximation approach as in Remark~\ref{approximation-heuristics}.
\end{rem}

We are now ready to use the approximation approach to analyze the effects of feedback on continuous-time Gaussian BCs.

The following theorem says that feedback does not help physically degraded Gaussian BCs, whose proof is inspired by the ideas in Section~\ref{infinite-bandwidth-revisited} and parallels the argument in~\cite{el81}.
\begin{thm}  \label{Theorem-BC-With-Feedback-1}
Consider the following continuous-time physically degraded Gaussian broadcast channel with one sender and two receivers:
$$
Y_1(t)=\int_0^t X(s, M_1, Y_{1, 0}^s, Y_{2, 0}^s) ds+\sqrt{N_1} B_1(t),
$$
$$
Y_2(t)=\int_0^t X(s, M_2, Y_{1, 0}^s, Y_{2, 0}^s) ds+\sqrt{N_1} B_1(t)+\sqrt{N_2} B_2(t),
$$
where $N_1, N_2 > 0$, and $B_1, B_2$ are independent standard Brownian motions, and the channel input $X(s)$ is assumed to satisfy Conditions (d)-(f). Then, feedback does not increase the capacity region of the above channel.
\end{thm}

\begin{proof}
Let $X$ be a $(T, (e^{T R_1}, e^{T R_2}), P)$-code. By the code construction, for $i=1, 2$, it is possible to estimate the messages $M_i$ from the channel output $Y_{i, 0}^T$ with an arbitrarily low probability of error. Hence, by Fano's inequality, for $i=1, 2$,
$$
H(M_i|Y_{i, 0}^T) \leq T R_i P^{(T)}_e +H(P^{(T)}_e) = T \varepsilon_{i, T},
$$
where $\varepsilon_{i, T} \rightarrow 0$ as $T \rightarrow \infty$. It then follows that
$$
T R_1 = H(M_1) = H(M_1|M_2) \leq I(M_1; Y_{1, 0}^T|M_2) + T \varepsilon_{1, T},
$$
$$
T R_2 = H(M_2) \leq I(M_2; Y_{2, 0}^T) + T \varepsilon_{2, T}.
$$
Now the Euler-Maruyama approximation with respect to the evenly spaced $\Delta_n$ of stepsize $\delta_n = T/n$ applied to the continuous-time physically degraded Gaussian BC yields:
{\small $$
\hspace{-1cm} Y_1^{(n)}(t_{n, i})-Y_1^{(n)}(t_{n, i-1})=\int_{t_{n, i-1}}^{t_{n, i}} X(s, M, Y_{1, t_{n, 0}}^{(n), t_{n, i-1}}, Y_{2, t_{n, 0}}^{(n), t_{n, i-1}}) ds+\sqrt{N_1} B_1(t_{n, i})-\sqrt{N_1} B_1(t_{n, i-1}),
$$
$$
\hspace{-2.4cm} Y_2^{(n)}(t_{n, i})-Y_2^{(n)}(t_{n, i-1})=\int_{t_{n, i-1}}^{t_{n, i}} X(s, M, Y_{1, t_{n, 0}}^{(n), t_{n, i-1}}, Y_{2, t_{n, 0}}^{(n), t_{n, i-1}}) ds+\sqrt{N_1} B_1(t_{n, i})-\sqrt{N_1} B_1(t_{n, i-1})+ \sqrt{N_2} B_2(t_{n, i})- \sqrt{N_2} B_2(t_{n, i-1}).
$$}
Then, by Theorem~\ref{approximation-theorem-1}, we have
\begin{align*}
I(M_2; Y_{2,0}^T) & = \lim_{n \to \infty} I(M_2; Y^{(n)}_{2}(\Delta_n))\\
& = \lim_{n \to \infty} I(M_2; \Delta Y^{(n)}_{2}(\Delta_n))\\
& = \lim_{n \to \infty} h(\Delta Y^{(n)}_{2}(\Delta_n))-h(\Delta Y^{(n)}_{2}(\Delta_n)|M_2),
\end{align*}
where $\Delta Y^{(n)}_{2}(\Delta_n) \triangleq \{Y_2^{(n)}(t_{n, i})-Y_2^{(n)}(t_{n, i-1}): i=1, 2, \cdots, n\}$. Note that
$$
H(\Delta Y^{(n)}_{2}(\Delta_n)) \leq \sum_{i=1}^n \log (2 \pi e(P \delta_n^2+N_2 \delta_n)),
$$
and
\begin{align*}
H(\Delta Y^{(n)}_{2}(\Delta_n)|M_2)& =\sum_{i=1}^n h(Y^{(n)}_{2}(t_{n, i})-Y^{(n)}_{2}(t_{n, i-1})|Y_{2, t_{n, 0}}^{(n), t_{n, i-1}}, M_2)\\
& \geq \sum_{i=1}^n h(\sqrt{N_2} B_{2}(t_{n, i})-\sqrt{N_2} B_{2}(t_{n, i-1}))\\
&=\sum_{i=1}^n \log (2 \pi e N_2 \delta_n),
\end{align*}
which implies that there exists an $\alpha \in [0, 1]$ such that
$$
h(\Delta Y^{(n)}_{2}(\Delta_n)|M) = \sum_{i=1}^n \frac{n}{2} \log (2 \pi e (\alpha P \delta_n^2 + N_2 \delta_n)).
$$
It then follows from Theorem~\ref{approximation-theorem-1} that
$$
I(M_2; Y_{2, 0}^T) \leq \frac{1}{2} \lim_{n \to \infty} \sum_{i=1}^n \log \frac{P \delta_n^2+N_2 \delta_n}{\alpha P \delta_n^2+N_2 \delta_n} = \frac{(1-\alpha)P T}{2 N_2}.
$$
Next we consider
\begin{align*}
I(M_1; Y^{(n)}_{1}(\Delta_n)|M_2) & = h(Y^{(n)}_{1}(\Delta_n)|M_2)-h(Y^{(n)}_{1}(\Delta_n)|M_1, M_2)\\
&=h(Y^{(n)}_{1}(\Delta_n)|M_2)-\sum_{i=1}^n h(Y^{(n)}_{1}(t_{n, i})|M_1, M_2, Y_{1, t_{n, 0}}^{(n), t_{n, i-1}})\\
&\leq h(Y^{(n)}_{1}(\Delta_n)|M_2)-\sum_{i=1}^n h(Y^{(n)}_{1}(t_{n, i})|M_1, M_2, Y_{1, t_{n, 0}}^{(n), t_{n, i-1}}, Y_{1, t_{n, 0}}^{(n), t_{n, i-1}})\\
& = h(Y^{(n)}_{1}(\Delta_n)|M_2)-\frac{1}{2} \sum_{i=1}^n \log(2 \pi e N_1 \delta_n).
\end{align*}
Now, using Lemma $1$ in~\cite{el81} (an extension of the entropy power inequality), we obtain
\begin{align*}
h(Y^{(n)}(\Delta_n)|M_2) \geq \frac{n}{2} \log (2^{2 h(Y^{(n)}_{1}(\Delta_n))|M_2)/n}+2 \pi e (N_2-N_1) \delta_n),
\end{align*}
which immediately implies that
$$
h(Y_{1, t_{n, 0}}^{(n), t_{n, i}}|M_2) \leq \frac{1}{2} \sum_{i=1}^n \log (2 \pi e (\alpha P \delta_n^2 + N_1 \delta_n))
$$
and furthermore, by Theorem~\ref{approximation-theorem-1},
\begin{align*}
I(M_1; Y_{1, 0}^T|M_2) & \leq \lim_{n \to \infty} \frac{1}{2} \sum_{i=1}^n \log (2 \pi e (\alpha P \delta_n^2 + N_1 \delta_n)) - \frac{1}{2} \sum_{i=1}^n \log (2 \pi e N_1 \delta_n)\\
& = \lim_{n \to \infty} \frac{1}{2} \sum_{i=1}^n \log \left(1+ \frac{\alpha P \delta_n}{N_1}\right)\\
&=\frac{\alpha P T}{2N_1}.
\end{align*}
Now, by Theorem~\ref{Theorem-BC-Without-Feedback}, we conclude that feedback capacity region is exactly the same non-feedback capacity region; in other words, feedback does not increase the capacity region of a physically degraded continuous-time Gaussian BC.
\end{proof}

The following theorem says that feedback does help some stochastically degraded Gaussian BCs, whose proof, instead of directly employing the approximation theorem, uses the connections between continuous-time and discrete-time Gaussian channels and the notion of continuous-time directed information in Remark~\ref{directed-information}, both of which can find their source from the approximation theorem. We only provide the sketch of the proof, since it is largely based on a translated version of the argument in~\cite{ozle84}.
\begin{thm} \label{Theorem-BC-With-Feedback-2}
Feedback increases the capacity region of certain continuous-time stochastically degraded Gaussian broadcast channel.
\end{thm}

\begin{proof}
Consider the following symmetric continuous-time Gaussian broadcast channel:
$$
Y_1(t)=\int_0^t X(s) ds + B_1(t),
$$
$$
Y_2(t)=\int_0^t X(s) ds + B_2(t),
$$
where $B_1, B_2$ are independent standard Brownian motions, and $X$ satisfies the average power constraint $P$. By Theorem~\ref{Theorem-BC-Without-Feedback}, without feedback, the capacity region is the set of rate pairs $(R_1, R_2)$ such that
\begin{equation} \label{R1R2-1}
R_1 + R_2 \leq \frac{P}{2}.
\end{equation}
With feedback, one can use the following variation~\cite{ozle84} of the Schalkwijk-Kailath coding scheme~\cite{sc66} over $[0, T]$ at discrete time points $\{t_{n, i}\}$ that form an evenly spaced $\Delta_n$ of stepsize $\delta_n$: For the channel input, after some proper initialization, at time $t \in [t_{n, i}, t_{n, i+1})$, we send $X^{(n)}(t)=X_1^{(n)}(t)+X_2^{(n)}(t)$, where
$$
X_1^{(n)}(t)=\gamma_i (X_1^{(n)}(t_{n, i-1})-\E[X_1^{(n)}(t_{n, i-1})|Y_1^{(n)}(t_{n, i-1})]),
$$
$$
X_2^{(n)}(t)=-\gamma_i (X_2^{(n)}(t_{n, i-1})-\E[X_2^{(n)}(t_{n, i-1})|Y_2^{(n)}(t_{n, i-1})]),
$$
where $\gamma_i$ is chosen so that $\E[X_i^2(t)] = P$ for each $i$; and for the channel outputs, we have, for any $t \in [t_{n, i}, t_{n, i+1}]$,
\begin{equation} \label{BC-approximiated-1}
Y_1^{(n)}(t)=Y_1^{(n)}(t_{n, i})+ \int_{t_{n, i}}^{t} X^{(n)}(t_{n, i}) ds + B_1(t)-B_1(t_{n, i}),
\end{equation}
and
\begin{equation} \label{BC-approximiated-2}
Y_2^{(n)}(t)=Y_2^{(n)}(t_{n, i})+ \int_{t_{n, i}}^{t} X^{(n)}(t_{n, i}) ds + B_2(t)-B_2(t_{n, i}).
\end{equation}
Going through a completely parallel argument as in~\cite{ozle84} and capitalizing on the fact that the SNR in the channels (\ref{BC-approximiated-1}) and (\ref{BC-approximiated-2}) tend to $0$ as $n$ tends to infinity, we derive that
$$
\lim_{n \to \infty} I(X^{(n)}(\Delta_n) \to Y_1^{(n)}(\Delta_n))= \frac{1}{2}\sum_{n \to \infty} \sum_{i=1}^n \frac{\log\left(1+ \frac{P \delta_n (1+\rho^*)/2}{(1+P \delta_n (1-\rho^*))/2}\right)}{T}=\frac{P(1+\rho^*)}{4},
$$
and parallelly,
$$
\lim_{n \to \infty} I(X^{(n)}(\Delta_n) \to Y_2^{(n)}(\Delta_n))=\frac{P(1+\rho^*)}{4},
$$
where $\rho^* > 0$ satisfies the condition
$$
\rho^* (1+(P+1)(1+P(1-\rho^*)/2))=\frac{P(P+2)(1-\rho^*)}{2}.
$$
Note that, by Remark~\ref{directed-information}, we have
$$
I(X_0^{(n), T} \to Y_{1, 0}^{(n), T}) \geq I(X^{(n)}(\Delta_n) \to Y_1^{(n)}(\Delta_n)), \quad I(X_0^{(n), T} \to Y_{2, 0}^{(n), T}) \geq I(X^{(n)}(\Delta_n) \to Y_2^{(n)}(\Delta_n)),
$$
which immediately implies that
\begin{equation} \label{R1R2-2}
R_1=R_2=\frac{P(1+\rho^*)}{4}
\end{equation}
are achievable. The claim that feedback strictly increases the capacity region then follows from (\ref{R1R2-1}), and (\ref{R1R2-2}) and the fact that $\rho^* > 0$.
\end{proof}

\section{Conclusions and Future Work}

For a continuous-time white Gaussian channel without feedback, the classical Shannon-Nyquist sampling theorem can convert it to a discrete-time Gaussian channel, however such a link has long been missing when feedback/memory is present in the channel. In this paper, we establish sampling and approximation theorems as the missing links, which we believe will play important roles in the long run for further developing continuous-time information theory, particularly for the communication scenarios where feedback/memory is present.

As an immediate application of our approximation theorem, we propose the approximation approach, an analog of the conventional sampling approach, for Gaussian feedback channels. It turns out that, like its non-feedback counterpart, the approximation approach can bring insights and intuition to investigation of continuous-time Gaussian channels with possible feedback, and moreover, when complemented with relevant tools from stochastic calculus, can deliver rigorous treatments in the point-to-point or multi-user setting.

On the other hand though, there are many questions that remain unanswered and a number of directions that need to be further explored. Below we list a number of research directions that look promising in the near future.

1) The first direction is to strengthen and generalize our sampling and approximation theorems.

Note that both Theorem~\ref{sampling-theorem-2} and Theorem~\ref{approximation-theorem-1} require Conditions (d)-(f), which are stronger than the typical average power constraint. While Conditions (d)-(f) are rather mild for practical considerations, the stronger assumptions in our theorems will narrow their reach in some theoretical situations. For instance, despite the fact that our approximation theorem gives intuitive explanations to the rigorous treatment of continuous-time multi-user Gaussian channels in Section~\ref{multi-user-Gaussian-channels}, it fails to rigorously establish Theorems~\ref{Theorem-MAC} and~\ref{Theorem-BC-Without-Feedback}. The stochastic calculus approach employed in Section~\ref{multi-user-Gaussian-channels} requires only the power constraints, which can be loosely explained by the fact that Girsanov's theorem (or, more precisely, its several variants) only requires as weak conditions. It is certainly worthwhile to explore whether the assumptions in our sampling and approximation theorems can be relaxed either in general or for some special settings.

Another topic in this direction is the rate of convergence in the sampling and approximation theorems. While the current versions of our theorems have merely established some limits, the rate of convergence will certainly yield a more quantitative description of how fast those limits will be approached.

One can also consider generalizing these two theorems to general Gaussian channels~\cite{rozanov, hida93}. For this topic, note that there exist in-depth studies~\cite{hi74, ja74, hi75, ih80, ih85, ih90, ba91, ih94} on continuous-time point-to-point general Gaussian channels with possible feedback, for which  information-theoretic connections with the discrete-time setting are somehow lacking. A first step in this direction can be establishing sampling or approximation theorems for stationary Gaussian processes. Obviously, such theorems for stationary Gaussian processes can connect continuous-time stationary Gaussian channels to their discrete-time counterparts, for which the variational formulation of discrete-time stationary Gaussian feedback capacity in~\cite{kim10} proves to rather effective.

2) The second direction is to further explore the possible applications of our sampling and approximation theorems in the following respects.

We have shown that feedback may increase the capacity region of some continuous-time Gaussian BC, but the capacity regions of such channels remain unknown in general. An immediate problem is to explicitly find the exact capacity regions of continuous-time Gaussian BCs using the approach employed in this work, as we have done for continuous-time Gaussian MACs. Of course, further topics also include exploring whether the ideas and techniques in this paper can be applied to other families of continuous-time multi-user Gaussian channels with possible feedback.

So far we have implicitly assumed infinite bandwidth and average power constraints, but our theorems can certainly go beyond these assumptions. For instance, one can consider examining continuous-time Gaussian channels with both bandwidth limit and peak power constraint, which are more reasonable assumptions for many practical communication scenarios as they give a more accurate description of the limitations of the communication system. Little is known about the capacity of continuous-time Gaussian channels with such constraints except some upper and lower bounds established in~\cite{ozarow85, Shamai-David-1989}. In stark contrast, discrete-time peak power constrained channels (including, but not limited to Gaussian channels) have been better investigated: there has been a series of work on their capacity, such as~\cite{smith71, Shamai-David-1995, abou01, chan05, raginsky08, sharma10, fahs12, Dytso}, which feature relatively thorough discussions about different aspects of channel capacity including capacity achieving distribution, bounds and asymptotics of capacity, and numerical computation of capacity. An immediate question is to explore whether the approximation approach can translate the aforementioned existing results in discrete time, or more probably, help channel the ideas and techniques therein to the continuous-time setting. A next question is to explore whether there exists any randomized algorithm for computation of the capacity of such a channel, for which, as discussed in Remark~\ref{from-continuous-to-discrete}, we believe our sampling theorems can be particularly helpful in terms of numerically computing and optimizing the mutual information of a continuous-time Gaussian channel with bandwidth limit and peak power constraint.

\bigskip
{\bf Acknowledgement.} We would like to thank Ronit Bustin, Jun Chen, Young-Han Kim, Haim Permuter, Shlomo Shamai, Tsachy Weissman and Wenyi Zhang for insightful suggestions and comments, and for pointing out relevant references.

\section*{Appendices} \appendix

\section{Proof of Theorem~\ref{sampling-theorem-1}} \label{proof-sampling-theorem-1}

First of all, an application of Theorem $7.14$ of~\cite{li01} with Conditions (b) and (c) yields that
\begin{equation} \label{crazy}
P\left(\int_0^T \EX^2[g(t, W_0^t, Y_0^t)|Y_0^t] dt < \infty \right)=1.
\end{equation}
Then one verifies that the assumptions of Lemma $7.7$ of~\cite{li01} are all satisfied (this lemma is stated under very general assumptions, which are exactly Conditions (b), (c) and (\ref{crazy}) when restricted to our settings), which implies that for any $w$,
$$
\mu_Y \sim \mu_{Y|W=w} \sim \mu_B,
$$
where ``$\sim$'' is the standard notation for two measures being equivalent (i.e., one is absolutely continuous with respect to the other and vice versa), and moreover, with probability $1$,
{\small \begin{equation}  \label{RN-1}
\hspace{-0.5cm} \frac{d\mu_{Y|W}}{d\mu_B}(Y_0^T)=\frac{1}{\E[e^{-\int_0^T g(s) dY(s)+\frac{1}{2} \int_0^T g(s)^2 ds}|Y_0^T, W_0^T]}, \quad \frac{d\mu_{Y}}{d\mu_B}(Y_0^T)=\frac{1}{\E[e^{-\int_0^T g(s) dY(s)+\frac{1}{2} \int_0^T g(s)^2 ds}|Y_0^T]},
\end{equation}}
where we have rewritten $g(s, W_0^s, Y_0^s)$ as $g(s)$ for notational simplicity. Here we remark that $\E[e^{-\int_0^T g(s) dY(s)+\frac{1}{2} \int_0^T g(s)^2 ds}|Y_0^T, W_0^T]$ is in fact equal to $e^{-\int_0^T g(s) dY(s)+\frac{1}{2} \int_0^T g(s)^2 ds}$, but we keep it the way it is as above for an easy comparison.

Note that it follows from $\E[d\mu_B/d\mu_Y(Y_0^T)]=1$ that
$$
\E[e^{-\int_0^T g(s) dY(s)+\frac{1}{2} \int_0^T g(s)^2 ds}]=1,
$$
which is equivalent to
$$
\E[e^{-\int_0^T g(s) dB(s)-\frac{1}{2} \int_0^T g(s)^2 ds}]=1.
$$
Then, a parallel argument as in the proof of Theorem $7.1$ of~\cite{li01} further implies that for any $\Delta_n$, with probability $1$,
{\small \begin{equation}  \label{RN-2}
\hspace{-1.5cm} \frac{d\mu_{Y|W}}{d\mu_B}(Y(\Delta_n))=\frac{1}{\E[e^{-\int_0^T g(s) dY(s)+\frac{1}{2} \int_0^T g(s)^2 ds}|Y(\Delta_n), W_0^T]}, \quad \frac{d\mu_{Y}}{d\mu_B}(Y(\Delta_n))=\frac{1}{\E[e^{-\int_0^T g(s) dY(s)+\frac{1}{2} \int_0^T g(s)^2 ds}|Y(\Delta_n)]},
\end{equation}}
where we have defined
$$
\quad B(\Delta_n) \triangleq \{B(t_{n, 0}), B(t_{n, 1}), \dots, B(t_{n, n})\}, $$
and moreover,
$$
\frac{d\mu_{Y|W}}{d\mu_B}(Y(\Delta_n)) \triangleq \frac{d\mu_{Y(\Delta_n)|W}}{d\mu_{B(\Delta_n)}}(Y(\Delta_n)), \quad \frac{d\mu_{Y}}{d\mu_B}(Y(\Delta_n)) \triangleq
\frac{d\mu_{Y(\Delta_n)}}{d\mu_{B(\Delta_n)}}(Y(\Delta_n)).
$$
Then, by definition, we have
$$
I(W_0^T; Y(\Delta_n)) = \E\left[\log \frac{d\mu_{Y|W}}{d\mu_B}(Y(\Delta_n))\right]-\E\left[\log \frac{d\mu_{Y}}{d\mu_B}(Y(\Delta_n))\right].
$$
Notice that it can be easily checked that $e^{-\int_0^T g(s) dY(s)+\frac{1}{2} \int_0^T g(s)^2 ds}$ integrable, which, together with the fact that $\Delta_n \subset \Delta_{n+1}$ for all $n$, further implies that
$$
\E[e^{-\int_0^T g(s) dY(s)+\frac{1}{2} \int_0^T g(s)^2 ds}|Y(\Delta_n), W_0^T], \quad \E[e^{-\int_0^T g(s) dY(s)+\frac{1}{2} \int_0^T g(s)^2 ds}|Y(\Delta_n)]
$$
are both martingales, and therefore, by Doob's martingale convergence theorem~\cite{Durrett},
$$
\frac{d\mu_{Y|W}}{d\mu_B}(Y(\Delta_n)) \to \frac{d\mu_{Y|W}}{d\mu_B}(Y_0^T), \quad \frac{d\mu_{Y}}{d\mu_B}(Y(\Delta_n)) \to \frac{d\mu_{Y}}{d\mu_B}(Y_0^T), \mbox{ a.s.}
$$

Now, by Jensen's inequality, we have
\begin{equation} \label{xianming-inequality-1}
\E\left[ \left. -\int^T_0 g(s) dY(s)+\frac{1}{2}\int^T_0
g(s)^2ds \right| Y(\Delta_n), W_0^T \right] \leq \log\E[e^{-\int^T_0 g(s)
dY(s)+\frac{1}{2}\int^T_0 g(s)^2ds}|Y(\Delta_n), W_0^T],
\end{equation}
and, by the fact that $\log x \leq x$ for any $x > 0$, we have
\begin{equation} \label{xianming-inequality-2}
\log\E[e^{-\int^T_0 g(s) dY(s)+\frac{1}{2}\int^T_0 g(s)^2ds}|Y(\Delta_n), W_0^T] \leq
\E[e^{-\int^T_0 g(s) dY(s)+\frac{1}{2}\int^T_0 g(s)^2ds}|Y(\Delta_n), W_0^T].
\end{equation}
It then follows from (\ref{xianming-inequality-1}) and (\ref{xianming-inequality-2}) that
\begin{align*}
\left|\log\E[e^{-\int^T_0 g(s) dY(s)+\frac{1}{2}\int^T_0 g(s)^2ds}|Y(\Delta_n), W_0^T]\right| & \leq
\left|\E\left[\left. -\int^T_0 g(s) dY(s)+\frac{1}{2}\int^T_0 g(s)^2 ds \right|Y(\Delta_n), W_0^T \right]\right|\\
&+\E[e^{-\int^T_0 g(s) dY(s)+\frac{1}{2}\int^T_0 g(s)^2ds}|Y(\Delta_n), W_0^T].
\end{align*}
Applying the general Lebesgue dominated convergence theorem (see, e.g., Theorem $19$ on Page $89$ of~\cite{ro10}), we then have
$$
\lim_{n \to \infty} \E\left[\log \frac{d\mu_{Y|W}}{d\mu_B}(Y(\Delta_n))\right] = \E [\log  \E[e^{-\int^T_0 g(s) dY(s)+\frac{1}{2}\int^T_0 g(s)^2ds}|Y^{T}_{0}, W_0^T]] = \E\left[\log \frac{d\mu_{Y|W}}{d\mu_B}(Y_{0}^{T})\right].
$$
A completely parallel argument yields that
$$
\lim_{n \to \infty} \E\left[\log \frac{d\mu_{Y}}{d\mu_B}(Y(\Delta_n))\right] = \E[\log  \E[e^{-\int^T_0 g(s) dY(s)+\frac{1}{2}\int^T_0 g(s)^2ds}|Y^{T}_{0}]] = \E\left[\log \frac{d\mu_{Y}}{d\mu_B}(Y_{0}^{T})\right].
$$
So, with the definition
$$
I(W_0^T; Y_0^T)=\E\left[\log \frac{d\mu_{Y|W}}{d\mu_B}(Y_0^T)\right]-\E\left[\log \frac{d\mu_{Y}}{d\mu_B}(Y_0^T)\right],
$$
we conclude that
$$
\lim_{n \to \infty} I(W_0^T; Y(\Delta_n))=\E\left[\log \frac{d\mu_{Y|W}}{d\mu_B}(Y_0^T)\right]-\E\left[\log \frac{d\mu_{Y}}{d\mu_B}(Y_0^T)\right]=I(W_0^T;Y_0^T).
$$

\section{Proof of Lemma~\ref{improved-liptser-1}} \label{proof-improved-liptser-1}

With Conditions (d)-(f), the proof of the existence and uniqueness of the solution to (\ref{feedback-memory}) is somewhat standard; see, e.g., Section $5.4$ in~\cite{mao97}. So, in the following, we will only prove (\ref{exponential-finiteness-1}).

For the stochastic differential equation (\ref{feedback-memory}), applying Condition (e), we deduce that there exists $L_1 > 0$ such that
\begin{align*}
\|Y_0^T\| & \leq \int^T_0 L_1(1+\|W_0^t\| + \|Y_0^t\|) dt+ \|B_0^T\|\\
& \leq L_1 T+L_1 T \|W_0^T\| + \|B_0^T\| + \int_0^T L_1 \|Y_0^t\| dt.
\end{align*}
Then, applying the Gronwall inequality followed by a straightforward bounding analysis, we deduce that there exists $L_2 > 0$ such that
\begin{align*}
\|Y_0^T\| & \leq (L_1 T+L_1 T \|W_0^T\| + \|B_0^T\|) e^{\int_0^T L_1 dt} \\
& = e^{L_1 T} (L_1 T+L_1 T \|W_0^T\| + \|B_0^T\|)\\
& = L_2 + L_2 \|W_0^T\| + L_2 \|B_0^T\|.
\end{align*}
Now, for any $\varepsilon > 0$, applying Doob's submartingale inequality, we have
\begin{align*}
\E[e^{\varepsilon \|Y_0^T\|^2}] & \leq \E[e^{\varepsilon(L_2 + L_2 \|W_0^T\| + L_2 \|B_0^T\|)^2}] \\
&\leq \E[e^{3\varepsilon(L_2^2 + L_2^2 \|W_0^T\|^2 + L_2^2 \|B_0^T\|^2)}] \\
&=e^{3\varepsilon L_2^2} \E[e^{3 \varepsilon L_2^2 \|W_0^T\|^2}] \E[e^{3 \varepsilon L_2^2 \|B_0^T\|^2}]\\
&=e^{3\varepsilon L_2^2} \E[e^{3 \varepsilon L_2^2 \|W_0^T\|^2}] \E[\sup\nolimits_{0\leq t\leq T} e^{3 \varepsilon L_2^2 B(t)^2}]\\
&\leq 4 e^{3\varepsilon L_2^2} \E[e^{3 \varepsilon L_2^2 \|W_0^T\|^2}] \E[e^{3 \varepsilon L_2^2 B(T)^2}],
\end{align*}
which, by Condition (f), is finite provided that $\varepsilon$ is small enough.

\section{Proof of Theorem~\ref{sampling-theorem-2}}  \label{proof-sampling-theorem-2}

We proceed in the following steps.

{\bf Step $\bf 1$.} In this step, we establish the theorem assuming that there exists $C > 0$ such that for all $w_0^T \in C[0, T]$ and all $y_0^T \in C[0, T]$,
\begin{equation} \label{case-1}
\int_0^T g^2(s, w_0^s, y_0^s) ds < C.
\end{equation}
By the definition of mutual information, (\ref{RN-1}) and (\ref{RN-2}), we have
{\small \begin{align*}
I(W_0^T; Y(\Delta_n)) &= \E\left[\log \frac{d\mu_{Y|W}}{d\mu_B}(Y(\Delta_n))\right]-\E\left[\log \frac{d\mu_{Y}}{d\mu_B}(Y(\Delta_n))\right]\\
&= -\E[\log \E[e^{-\int_0^T g(s) dY(s)+\frac{1}{2} \int_0^T g(s)^2 ds}|Y(\Delta_n), W_0^T]] + \E[\log \E[e^{-\int_0^T g(s) dY(s)+\frac{1}{2} \int_0^T g(s)^2 ds}|Y(\Delta_n)]]\\
&= -\E[F_n]+\E[G_n],
\end{align*}}
where, for notational simplicity, we have rewritten $g(s, W_0^s, Y_0^s)$ as $g(s)$.

{\bf Step $\bf 1.1$.} In this step, we prove that as $n$ tends to infinity,
\begin{equation} \label{Fn-Conv}
F_n \to -\int_0^T g(s) dY(s)+\frac{1}{2} \int_0^T g(s)^2 ds,
\end{equation}
in probability.

Let $\bar{Y}_{\Delta_n, 0}^T$ denote the piecewise linear version of $Y_0^T$ with respect to $\Delta_n$; more precisely, for any $i=0, 1, \dots, n$, $\bar{Y}_{\Delta_n}(t_{n,i})=Y(t_{n,i})$, and for any $t_{n, i-1} < s < t_{n, i}$ with $s=\lambda t_{n, i-1}+(1-\lambda) t_{n, i}$ for some $0 < \lambda < 1$, $\bar{Y}_{\Delta_n}(s)=\lambda Y(t_{n, i-1})+(1-\lambda) Y(t_{n, i})$. Let $\bar{g}_{\Delta_n}(s, W_0^s, \bar{Y}_{\Delta_n, 0}^s)$ denote the piecewise ``flat'' version of $g(s, W_0^s, \bar{Y}_{\Delta_n, 0}^s)$ with respect to $\Delta_n$; more precisely, for any $t_{n, i-1} \leq s < t_{n, i}$, $\bar{g}_{\Delta_n}(s, W_0^s, \bar{Y}_{\Delta_n, 0}^s)=g(t_{n, i-1}, W_0^{t_{n, i-1}}, \bar{Y}_{\Delta_n, 0}^{t_{n, i-1}})$.

Rewriting $\bar{g}_{\Delta_n}(s, W_0^s, \bar{Y}_{\Delta_n, 0}^s)$ as $\bar{g}_{\Delta_n}(s)$, we have
\begin{align*}
F_n & = -\log \E[e^{-\int_0^T g(s) dY(s)+\frac{1}{2} \int_0^T g^2(s) ds}|Y(\Delta_n), W_0^T]\\
    & = -\log \E[e^{-\int_0^T \bar{g}_{\Delta_n}(s) dY(s)+\frac{1}{2} \int_0^T \bar{g}_{\Delta_n}^2(s) ds-\int_0^T (g(s)-\bar{g}_{\Delta_n}(s)) dY(s) + \frac{1}{2} \int_0^T (g^2(s)-\bar{g}_{\Delta_n}^2(s)) ds}|Y(\Delta_n), W_0^T]\\
    & = -\log e^{-\int_0^T \bar{g}_{\Delta_n}(s) dY(s)+\frac{1}{2} \int_0^T \bar{g}_{\Delta_n}^2(s) ds} \E[e^{-\int_0^T (g(s)-\bar{g}_{\Delta_n}(s)) dB(s) - \frac{1}{2} \int_0^T (g(s)-\bar{g}_{\Delta_n}(s))^2(s) ds}|Y(\Delta_n), W_0^T]\\
    & = -\int_0^T \bar{g}_{\Delta_n}(s) dY(s)-\frac{1}{2} \int_0^T \bar{g}_{\Delta_n}^2(s) ds -\log \E[e^{-\int_0^T (g(s)-\bar{g}_{\Delta_n}(s)) dB(s) - \frac{1}{2} \int_0^T (g(s)-\bar{g}_{\Delta_n}(s))^2 ds}|Y(\Delta_n), W_0^T],
\end{align*}
where we have used the fact that
$$
\E[e^{-\int_0^T \bar{g}_{\Delta_n}(s) dY(s)+\frac{1}{2} \int_0^T \bar{g}_{\Delta_n}^2(s) ds}|Y(\Delta_n), W_0^T]=e^{-\int_0^T \bar{g}_{\Delta_n}(s) dY(s)+\frac{1}{2}\int_0^T \bar{g}_{\Delta_n}^2(s) ds},
$$
since $\bar{g}_{\Delta_n}(s)$ is a function depending only on $W_0^T$ and $Y(\Delta_n)$.

We now prove the following convergence:
\begin{equation} \label{conv-1}
\E\left[\left( \left(-\int_0^T \bar{g}_{\Delta_n}(s) dY(s)-\frac{1}{2} \int_0^T \bar{g}_{\Delta_n}^2(s) ds\right)-\left(-\int_0^T g(s) dY(s)-\frac{1}{2} \int_0^T g^2(s) ds\right)\right)^2\right] \to 0,
\end{equation}
which will imply that
$$
-\int_0^T \bar{g}_{\Delta_n}(s) dY(s)-\frac{1}{2} \int_0^T \bar{g}_{\Delta_n}^2(s) ds \to -\int_0^T g(s) dY(s)-\frac{1}{2} \int_0^T g^2(s) ds
$$
in probability. Apparently, to prove (\ref{conv-1}), we only need to prove that
\begin{equation} \label{conv-1a}
\E\left[\left(-\int_0^T (g(s)-\bar{g}_{\Delta_n}(s)) dB(s) - \frac{1}{2} \int_0^T (g(s)-\bar{g}_{\Delta_n}(s))^2 ds \right)^2 \right] \to 0.
\end{equation}
To establish (\ref{conv-1a}), notice that, by the It\^{o} isometry~\cite{ok95}, we have
$$
\E\left[\left(\int_0^T (g(s)-\bar{g}_{\Delta_n}(s)) dB(s) \right)^2\right] = \E\left[\int_0^T (g(s)-\bar{g}_{\Delta_n}(s))^2 ds \right],
$$
which means we only need to prove that as $n \to \infty$,
\begin{equation} \label{conv-1b}
\E\left[\left(\int_0^T (g(s)-\bar{g}_{\Delta_n}(s))^2 ds\right)^2 \right] \to 0.
\end{equation}
To see this, we note that, by Conditions (d) and (e), there exists $L_1 > 0$ such that for any $s \in [0, T]$ with $t_{n, i-1} \leq s < t_{n, i}$,
\begin{align}
& \hspace{-1cm} |g(s, W_0^s, \bar{Y}_{\Delta_n, 0}^s)-\bar{g}_{\Delta_n}(s, W_0^s, \bar{Y}_{\Delta_n, 0}^s)| \nonumber \\
& = |g(s, W_0^s, \bar{Y}_{\Delta_n, 0}^s)-g(t_{n, i-1}, W_0^{t_{n, i-1}}, \bar{Y}_{\Delta_n, 0}^{t_{n, i-1}})| \nonumber\\
& \leq L_1 (|s-t_{n, i-1}| + \|W_0^s-W_0^{t_{n, i-1}}\|+\| \bar{Y}_{\Delta_n, 0}^s - \bar{Y}_{\Delta_n, 0}^{t_{n, i-1}} \|) \nonumber\\
& \leq L_1 (|s-t_{n, i-1}| + \|W_0^s-W_0^{t_{n, i-1}}\|+ |Y(t_{n, i})-Y(t_{n, i-1})|) \label{diff-1} \\
& \leq  L_1 \delta_{\Delta_n}+ L_1 \sup\nolimits_{r \in [t_{n, i-1}, t_{n, i}]} |W(r)-W(t_{n, i-1})| \nonumber \\
& +L_1 \delta_{\Delta_n}+L_1 \delta_{\Delta_n} \|W_0^T\|+L_1 \delta_{\Delta_n} \|Y_0^T\|+|B(t_{n, i})-B(t_{n, i-1})|. \label{diff-2}
\end{align}
Moreover, by Lemma~\ref{improved-liptser-1} and Condition (f), both $\|Y_0^T\|^4$ and $\|W_0^T\|^4$ are integrable. And furthermore, by Condition (f), we deduce that for any $t_{n, i-1} \leq s < t_{n, i}$,
\begin{equation} \label{conv-1c}
\E[\sup\nolimits_{r \in [t_{n, i-1}, t_{n, i}]} (W(r)-W(t_{n, i-1}))^4] \leq L_2 \delta_{\Delta_n}^2,
\end{equation}
for some $L_2 > 0$, and one easily verifies that
\begin{equation} \label{conv-1d}
\EX[(B(t^{(n)}_i)-B(t^{(n)}_{i-1}))^4] = 3 (t^{(n)}_i-t^{(n)}_{i-1})^2 \leq 3 \delta^2_{\Delta_n}.
\end{equation}
It can be readily checked that (\ref{diff-2}), (\ref{conv-1c}) and (\ref{conv-1d}) imply (\ref{conv-1b}), which in turn implies (\ref{conv-1}), as desired.

We now prove that as $n$ tends to infinity,
\begin{equation} \label{conv-2}
\E[|\E[e^{-\int_0^T (g(s)-\bar{g}_{\Delta_n}(s)) dB(s) - \frac{1}{2} \int_0^T (g(s)-\bar{g}_{\Delta_n}(s))^2
ds}|Y(\Delta_n), W_0^T]-1|] \to 0,
\end{equation}
which will imply that
$$
\log \E[e^{-\int_0^T (g(s)-\bar{g}(s)) dB(s) - \frac{1}{2} \int_0^T (g(s)-\bar{g}(s))^2 ds}|Y(\Delta_n), W_0^T] \to 0
$$
in probability and furthermore (\ref{Fn-Conv}). To establish (\ref{conv-2}), we first note that
{\small \begin{align*}
&\hspace{-5mm} \E[|\E[e^{-\int_0^T (g(s)-\bar{g}_{\Delta_n}(s)) dB(s) - \frac{1}{2} \int_0^T (g(s)-\bar{g}_{\Delta_n}(s))^2
ds}|Y(\Delta_n), W_0^T]-1|]\\
&\hspace{-5mm} \leq \E[\E[|e^{-\int_0^T (g(s)-\bar{g}_{\Delta_n}(s)) dB(s) - \frac{1}{2} \int_0^T (g(s)-\bar{g}_{\Delta_n}(s))^2 ds}-1||Y(\Delta_n), W_0^T]]\\
&\hspace{-5mm} = \E[|e^{-\int_0^T (g(s)-\bar{g}_{\Delta_n}(s)) dB(s) - \frac{1}{2} \int_0^T (g(s)-\bar{g}_{\Delta_n}(s))^2
ds}-1|]\\
&\hspace{-5mm} \leq \E\left[\left|-\int_0^T (g(s)-\bar{g}_{\Delta_n}(s)) dB(s) - \frac{1}{2} \int_0^T (g(s)-\bar{g}_{\Delta_n}(s))^2
ds \right| e^{|-\int_0^T (g(s)-\bar{g}_{\Delta_n}(s)) dB(s) - \frac{1}{2} \int_0^T (g(s)-\bar{g}_{\Delta_n}(s))^2
ds|}\right]\\
&\hspace{-5mm} \leq \E\left[\left|-\int_0^T (g(s)-\bar{g}_{\Delta_n}(s)) dB(s) - \frac{1}{2} \int_0^T (g(s)-\bar{g}_{\Delta_n}(s))^2
ds\right|^2\right] \E\left[e^{2 |-\int_0^T (g(s)-\bar{g}_{\Delta_n}(s)) dB(s) - \frac{1}{2} \int_0^T
(g(s)-\bar{g}_{\Delta_n}(s))^2 ds|}\right].
\end{align*}}
By (\ref{conv-1}), we have that as $n$ tends to infinity,
$$
\E\left[\left|-\int_0^T (g(s)-\bar{g}_{\Delta_n}(s)) dB(s) - \frac{1}{2} \int_0^T (g(s)-\bar{g}_{\Delta_n}(s))^2
ds\right|^2\right] \to 0.
$$
It then follows that, to prove (\ref{conv-2}), we only need to prove that if $\delta_{\Delta_n}$ is small enough,
\begin{equation} \label{conv-2a}
\E\left[e^{2 |-\int_0^T (g(s)-\bar{g}_{\Delta_n}(s)) dB(s) - \frac{1}{2} \int_0^T (g(s)-\bar{g}_{\Delta_n}(s))^2 ds|}\right] < \infty.
\end{equation}
Since
$$
\E\left[e^{2 |-\int_0^T (g(s)-\bar{g}_{\Delta_n}(s)) dB(s) - \frac{1}{2} \int_0^T (g(s)-\bar{g}_{\Delta_n}(s))^2 ds|}\right]
$$
\begin{equation} \label{two-terms}
\leq \E\left[e^{2 (-\int_0^T (g(s)-\bar{g}_{\Delta_n}(s)) dB(s) - \frac{1}{2} \int_0^T (g(s)-\bar{g}_{\Delta_n}(s))^2 ds)}\right]+ \E\left[e^{2( \int_0^T (g(s)-\bar{g}_{\Delta_n}(s)) dB(s) + \frac{1}{2} \int_0^T (g(s)-\bar{g}_{\Delta_n}(s))^2 ds)}\right],
\end{equation}
we only have to prove that the two terms in the above upper bound are both finite provided that $\delta_{\Delta_n}$ is small enough. Note that for the first term, applying the Cauchy-Schwarz inequality, we have
\begin{align*}
\hspace{-1cm} \E[e^{2 (-\int_0^T (g(s)-\bar{g}_{\Delta_n}(s)) dB(s) - \frac{1}{2} \int_0^T (g(s)-\bar{g}_{\Delta_n}(s))^2 ds)}] & = \E[e^{ \int_0^T 2 (g(s)-\bar{g}_{\Delta_n}(s)) dB(s) - \int_0^T 4 (g(s)-\bar{g}_{\Delta_n}(s))^2 ds+3 \int_0^T (g(s)-\bar{g}_{\Delta_n}(s))^2 ds}]\\
& \leq \E[e^{ \int_0^T 4 (g(s)-\bar{g}_{\Delta_n}(s)) dB(s) - \int_0^T 8 (g(s)-\bar{g}_{\Delta_n}(s))^2 ds}] \E[e^{6 \int_0^T (g(s)-\bar{g}_{\Delta_n}(s))^2 ds}].
\end{align*}
It is well known that an application of Fatou's lemma yields that
\begin{equation}  \label{Fatou-Lemma}
\E[e^{ \int_0^T 4 (g(s)-\bar{g}_{\Delta_n}(s)) dB(s) - \int_0^T 8 (g(s)-\bar{g}_{\Delta_n}(s))^2 ds}] \leq 1,
\end{equation}
and by (\ref{diff-2}), we deduce that there exists $L_3 > 0$ such that
$$
\hspace{-1.5cm} \E[e^{6 \int_0^T (g(s)-\bar{g}_{\Delta_n}(s))^2 ds}] \leq e^{L_3 \delta_{\Delta_n}^2} \E[e^{L_3 \|B_0^{\delta_{\Delta_n}}\|^2}] \E[e^{L_3 \delta^2_{\Delta_n} \|Y_0^T\|^2}] \E[e^{L_3 \delta^2_{\Delta_n} \|W_{0}^{T}\|^2}] \E[e^{L_3 \sup_{|s-t| \leq \delta_{\Delta_n}} |W(s)-W(t)|^2}].
$$
Note that it follows from Doob's submartingale inequality that if $\delta_{\Delta_n}$ is small enough,
$$
\E[e^{L_3 \|B_0^{\delta_{\Delta_n}}\|^2}] < \infty,
$$
and by Lemma~\ref{improved-liptser-1}, we also deduce that if $\delta_{\Delta_n}$ is small enough,
$$
\E[e^{L_3 \delta^2_{\Delta_n} \|Y_0^T\|^2}] < \infty,
$$
which, together with Condition (f), yields that for the first term in (\ref{two-terms})
\begin{equation} \label{SongJianMethod}
\E[e^{2 (-\int_0^T (g(s)-\bar{g}_{\Delta_n}(s)) dB(s) - \frac{1}{2} \int_0^T (g(s)-\bar{g}_{\Delta_n}(s))^2 ds)}] < \infty.
\end{equation}
A completely parallel argument will yield that for the second term in (\ref{two-terms})
$$
\E\left[e^{(2 \int_0^T (g(s)-\bar{g}_{\Delta_n}(s)) dB(s) + \frac{1}{2} \int_0^T (g(s)-\bar{g}_{\Delta_n}(s))^2 ds)}\right] < \infty,
$$
which, together with (\ref{SongJianMethod}), immediately implies (\ref{conv-2a}), which in turn implies (\ref{conv-2}), as desired.

{\bf Step $\bf 1.2$.} In this step, we prove that as $n$ tends to infinity,
\begin{equation} \label{Gn-Conv}
G_n \to \log \E[e^{-\int_0^T g(s) dY(s)+\frac{1}{2} \int_0^T g(s)^2 ds}|Y_0^T],
\end{equation}
in probability.

First, note that by Theorem $7.23$ of~\cite{li01}, we have,
$$
\frac{d \mu_{Y}}{d\mu_B}(Y_0^T)  = \int \frac{d \mu_{Y|W=w}}{d\mu_B}(Y_0^T) d\mu_W(w),
$$
where
$$
\frac{d \mu_{Y|W=w}}{d\mu_B}(Y_0^T) = e^{\int_0^T g(w_0^s) dY(s)-\frac{1}{2} \int_0^T g^2(w_0^s) ds},
$$
where we have rewritten $g(s, w_0^s, Y_0^s)$ as $g(w_0^s)$ for notational simplicity. It then follows from (\ref{RN-2}) that
\begin{align*}
\log \E[e^{-\int_0^T g(s) dY(s)+\frac{1}{2} \int_0^T g^2(s) ds}|Y_0^T] & = -\log \int \frac{d \mu_{Y|W=w}}{d\mu_B}(Y_0^T) d\mu_W(w)\\
& = -\log \int e^{\int_0^T g(w_0^s) dY(s)-\frac{1}{2} \int_0^T g^2(w_0^s) ds} d \mu_W(w).
\end{align*}
Similarly, we have
\begin{align*}
\frac{d \mu_{Y}}{d\mu_B}(Y(\Delta_n)) &= \int \frac{d \mu_{Y|W=w}}{d\mu_B}(Y(\Delta_n)) d\mu_W(w)\\
                                      &= \int \frac{1}{\E[e^{-\int_0^T g(s) dY(s)+\frac{1}{2} \int_0^T g^2(s) ds}|Y(\Delta_n), W]|_{W=w}} d\mu_W(w).
\end{align*}
It then again follows from (\ref{RN-2}) that
\begin{align*}
G_n & =  \log \E[e^{-\int_0^T g(s) dY(s)+\frac{1}{2} \int_0^T g^2(s) ds}|Y(\Delta_n)]\\
    & = -\log \int \frac{1}{\E[e^{-\int_0^T g(s) dY(s)+\frac{1}{2} \int_0^T g^2(s) ds}|Y(\Delta_n), W]|_{W=w}} d\mu_W(w).
\end{align*}

Now, we consider the following difference:
\begin{align*}
&\hspace{-1cm} \int e^{\int_0^T g(w_0^s) dY(s)-\frac{1}{2} \int_0^T g^2(w_0^s) ds} d\mu_W(w)-\int \frac{1}{\left.\E\left[ \left. e^{-\int_0^T g(w_0^s) dY(s)+\frac{1}{2} \int_0^T g^2(w_0^s) ds} \right| Y(\Delta_n), W\right]\right|_{W=w}} d\mu_W(w)\\
&=\int e^{\int_0^T g(w_0^s)dY(s)-\frac{1}{2} \int_0^T g^2(w_0^s)ds}-e^{\int_0^T \bar{g}_{\Delta_n}(w_0^s) dY(s)-\frac{1}{2} \int_0^T \bar{g}_{\Delta_n}(w_0^s)^2 ds} d\mu_W(w) \\
&+\int e^{\int_0^T \bar{g}_{\Delta_n}(w_0^s) dY(s)-\frac{1}{2} \int_0^T \bar{g}_{\Delta_n}(w_0^s)^2 ds} \\
&\hspace{1cm} \times \frac{\E[e^{-(\int_0^T g(w_0^s)dY(s)-\frac{1}{2} \int_0^T g^2(w_0^s)ds)+(\int_0^T \bar{g}_{\Delta_n}(w_0^s) dY(s)-\frac{1}{2} \int_0^T \bar{g}_{\Delta_n}^2(w_0^s) ds)}|Y(\Delta_n), W]|_{W=w}-1}{\E[e^{-(\int_0^T g(w_0^s)dY(s)-\frac{1}{2} \int_0^T g^2(w_0^s)ds)+(\int_0^T \bar{g}_{\Delta_n}(w_0^s) dY(s)-\frac{1}{2} \int_0^T \bar{g}_{\Delta_n}^2(w_0^s) ds)}|Y(\Delta_n), W]|_{W=w}} d\mu_W(w)\\
&=I_n+J_n.
\end{align*}
Applying the inequality that for any $x, y \in \mathbb{R}$,
\begin{equation} \label{expo-inequality}
|e^x-e^y|=|e^y(e^{x-y}-1)| \leq e^y (|x-y| e^{x-y}+ |x-y| e^{y-x}) = |x-y| (e^x+e^{2y-x}),
\end{equation}
we have
\begin{align*}
\E[|I_n|] & \leq \int \E\left[\left|e^{\int_0^T g(w_0^s)dY(s)-\frac{1}{2} \int_0^T g^2(w_0^s)ds}-e^{\int_0^T \bar{g}_{\Delta_n}(w_0^s) dY(s)-\frac{1}{2} \int_0^T \bar{g}_{\Delta_n}^2(w_0^s) ds} \right| \right] d\mu_W(w)\\
& \hspace{-1cm} \leq \int \E\left[\left| \int_0^T g(w_0^s)-\bar{g}_{\Delta_n}(w_0^s) dY(s)-\frac{1}{2} \int_0^T g^2(w_0^s)-\bar{g}_{\Delta_n}^2(w_0^s) ds\right| \right.\\
& \hspace{-1cm} \times \left. \left(e^{\int_0^T g(w_0^s)dY(s)-\frac{1}{2} \int_0^T g^2(w_0^s)ds}+e^{(2\int_0^T \bar{g}_{\Delta_n}(w_0^s) dY(s)-\int_0^T \bar{g}_{\Delta_n}^2(w_0^s) ds)-(\int_0^T g(w_0^s)dY(s)-\frac{1}{2} \int_0^T g^2(w_0^s)ds)} \right) \right] d\mu_W(w)\\
& \hspace{-1cm} \leq \int \E\left[\left| \int_0^T (g(w_0^s)-\bar{g}_{\Delta_n}(w_0^s)) dB(s) \right|+\left|\int_0^T (g(w_0^s)-\bar{g}_{\Delta_n}(w_0^s))(g(s)-\frac{1}{2}g(w_0^s)-\frac{1}{2}\bar{g}_{\Delta_n}(w_0^s)) ds\right| \right.\\
& \hspace{-1cm} \times \left. \left(e^{\int_0^T g(w_0^s)dY(s)-\frac{1}{2} \int_0^T g^2(w_0^s)ds}+e^{(2\int_0^T \bar{g}_{\Delta_n}(w_0^s) dY(s)-\int_0^T \bar{g}_{\Delta_n}^2(w_0^s) ds)-(\int_0^T g(w_0^s)dY(s)-\frac{1}{2} \int_0^T g^2(w_0^s)ds)} \right) \right] d\mu_W(w)\\
& \hspace{-1cm} \leq \int \E\left[\left| \int_0^T (g(w_0^s)-\bar{g}_{\Delta_n}(w_0^s)) dB(s) \right|+ (L \delta_{\Delta_n}+ L \sup_{|s-t| \leq \delta_{\Delta_n}}|w(s)-w(t)|+L \delta_{\Delta_n} \right.\\
& \hspace{-1cm} \left.+L \delta_{\Delta_n} \|w_0^T\|+L \delta_{\Delta_n} \|Y_0^T\|+\sup_{|s-t| \leq \delta_{\Delta_n}} |B(s)-B(t)|) \left(\int_0^T \left| g(s)-\frac{1}{2}g(w_0^s)-\frac{1}{2}\bar{g}_{\Delta_n}(w_0^s) \right| ds\right) \right.\\
& \hspace{-1cm} \times \left. \left(e^{\int_0^T g(w_0^s)dY(s)-\frac{1}{2} \int_0^T g^2(w_0^s)ds}+e^{(2\int_0^T \bar{g}_{\Delta_n}(w_0^s) dY(s)-\int_0^T \bar{g}_{\Delta_n}^2(w_0^s) ds)-(\int_0^T g(w_0^s)dY(s)-\frac{1}{2} \int_0^T g^2(w_0^s)ds)} \right) \right] d\mu_W(w).
\end{align*}
Now, using (\ref{diff-2}), Condition (f) and the It\^{o} isometry, we deduce that as $n \to \infty$,
\begin{equation} \label{prev-1}
\int \E\left[\left| \int_0^T g(w_0^s)-\bar{g}_{\Delta_n}(w_0^s) dB(s)\right|^2 \right] d\mu_{W}(w) \to 0,
\end{equation}
and as $n$ tends to infinity,
$$
\int \E[(L \delta_{\Delta_n}+ L \sup_{|s-t| \leq \delta_{\Delta_n}}|w(s)-w(t)|+L \delta_{\Delta_n}
$$
\begin{equation} \label{prev-2}
+L \delta_{\Delta_n} \|w_0^T\| +L \delta_{\Delta_n} \|Y_0^T\|+\sup_{|s-t| \leq \delta_{\Delta_n}} |B(s)-B(t)|)^2] d\mu_{W}(w) \to 0.
\end{equation}
Now, using a similar argument as above with (\ref{case-1}) and Lemma~\ref{improved-liptser-1}, we can show that for any constant $K$,
\begin{equation} \label{prev-3}
\E[e^{\int_0^T K \bar{g}_{\Delta_n}^2(s) ds}] = \E[e^{\int_0^T K (\bar{g}_{\Delta_n}(s)-g(s)+g(s))^2 ds}] = \E[e^{\int_0^T K (2(\bar{g}_{\Delta_n}(s)-g(s))^2+2g^2(s)) ds}]< \infty,
\end{equation}
provided that $n$ is large enough, which, coupled with a similar argument as in the derivation of (\ref{SongJianMethod}), proves that for $n$ large enough,
\begin{equation}  \label{prev-4}
\hspace{-1cm} \int \E\left[ \left(e^{\int_0^T g(w_0^s)dY(s)-\frac{1}{2} \int_0^T g^2(w_0^s)ds}+e^{(2\int_0^T \bar{g}(w_0^s) dY(s)-\int_0^T \bar{g}^2(w_0^s) ds)-(\int_0^T g(w_0^s)dY(s)-\frac{1}{2} \int_0^T g^2(w_0^s)ds)} \right)^2 \right] d\mu_W(w) < \infty,
\end{equation}
and furthermore
$$
\int \left(\int_0^T \left| g(s)-\frac{1}{2}g(w_0^s)-\frac{1}{2}\bar{g}_{\Delta_n}(w_0^s) \right| ds\right)^2
$$
\begin{equation} \label{prev-5}
\hspace{-5mm} \times \left(e^{\int_0^T g(w_0^s)dY(s)-\frac{1}{2} \int_0^T g^2(w_0^s)ds}+e^{(2\int_0^T \bar{g}_{\Delta_n}(w_0^s) dY(s)-\int_0^T \bar{g}_{\Delta_n}^2(w_0^s) ds)-(\int_0^T g(w_0^s)dY(s)-\frac{1}{2} \int_0^T g^2(w_0^s)ds)} \right)^2 d\mu_W(w) < \infty,
\end{equation}
which further implies that as $n$ tends to infinity,
\begin{equation}  \label{S-prime}
\E[|I_n|] \to 0.
\end{equation}

Now, using the shorthand notations $A_n$, $A$ for $\int_0^T \bar{g}_{\Delta_n}(w_0^s) dY(s)-\frac{1}{2} \int_0^T \bar{g}_{\Delta_n}(w_0^s)^2 ds$, $\int_0^T g(w_0^s) dY(s)-\frac{1}{2} \int_0^T g(w_0^s)^2 ds$ respectively, we have
\begin{align*}
\E[|J_n|] & =\E\left[\left| \int e^{A_n} \frac{\E[e^{-A+A_n}|Y(\Delta_n), W]|_{W=w}-1}{\E[e^{-A+A_n}|Y(\Delta_n), W]|_{W=w}} d\mu_W(w) \right| \right] \\
& =\E\left[\left|\int \frac{\E[e^{-A+A_n}-1|Y(\Delta_n), W]|_{W=w}}{\E[e^{-A}|Y(\Delta_n), W]|_{W=w}} d\mu_W(w) \right|\right] \\
& \leq \E\left[\int \frac{\E[|e^{-A+A_n}-1||Y(\Delta_n), W]|_{W=w}}{\E[e^{-A}|Y(\Delta_n), W]|_{W=w}} d\mu_W(w)\right] \\
& \leq \E\left[\int \E[|A-A_n| e^{|A-A_n|} | Y(\Delta_n), W]|_{W=w} \E[e^{A}|Y(\Delta_n), W]|_{W=w} d\mu_W(w)\right]\\
& = \E\left[\E[|A-A_n| e^{|A-A_n|} | Y(\Delta_n), W] \E[e^{A}|Y(\Delta_n), W] \left(\frac{d\mu_Y}{d\mu_B}(Y_0^T)\right)/\left(\frac{d\mu_{Y|W}}{d\mu_B}(Y_0^T)\right)\right]\\
& = \E\left[|A-A_n| e^{|A-A_n|} \E[e^{A}|Y(\Delta_n), W] \E\left[\left(\frac{d\mu_Y}{d\mu_B}(Y_0^T)\right)/\left(\frac{d\mu_{Y|W}}{d\mu_B}(Y_0^T)\right) | Y(\Delta_n), W \right]\right].
\end{align*}
Now, a similar argument as in (\ref{prev-1})-(\ref{prev-5}), together with the well-known fact (see, e.g., Theorem $6.2.2$ in~\cite{ih93}) that
$$
\frac{d\mu_Y}{d\mu_B}(Y_0^T)=e^{\int_0^T \hat{g}(s) dY(s)-\frac{1}{2} \int_0^T \hat{g}^2(s) ds}, \quad \frac{d\mu_{Y|W}}{d\mu_B}(Y_0^T)=e^{\int_0^T g(s) dY(s)-\frac{1}{2} \int_0^T g^2(s) ds},
$$
where $\hat{g}(s)=\E[g(s)|Y_0^T]$, yields that
$$
\E[|J_n|] \to 0.
$$
Now, we are ready to conclude that as $n$ tends to infinity,
$$
\E[e^{-\int_0^T g(s) dY(s)+\frac{1}{2} \int_0^T g^2(s) ds}|Y(\Delta_n)] \to \E[e^{-\int_0^T g(s) dY+\frac{1}{2} \int_0^T g^2(s) ds}|Y_0^T]
$$
in probability and furthermore (\ref{Gn-Conv}), as desired.

{\bf Step $\bf 1.3$.} In this step, we show the convergence of $\{\E[F_n]\}$ and $\{\E[G_n]\}$ and further establish the theorem under the condition (\ref{case-1}).

Now, using the concavity of the $\log$ function and the fact that $\log x \leq x$, we can obtain the upper bounds and lower bounds of $F_n$ and $G_n$ as follows:
$$
F_n \leq \left|\E\left[\left.-\int_0^T g(s) dY(s)+\frac{1}{2} \int_0^T g^2(s) ds \right| Y_{\Delta_n}, W_0^T \right]\right|+\E\left[\left. e^{-\int_0^T g(s) dY(s)+\frac{1}{2} \int_0^T g^2(s) ds}\right|Y_{\Delta_n}, W_0^T \right],
$$
$$
F_n \geq -\left| \E\left[\left.-\int_0^T g(s) dY(s)+\frac{1}{2} \int_0^T g^2(s) ds \right|Y_{\Delta_n}, W_0^T \right] \right|- \E\left[\left.e^{-\int_0^T g(s) dY(s)+\frac{1}{2} \int_0^T g^2(s) ds} \right| Y_{\Delta_n}, W_0^T \right],
$$
and
$$
G_n \leq \left|\E\left[\left.-\int_0^T g(s) dY(s)+\frac{1}{2} \int_0^T g^2(s) ds \right| Y_{\Delta_n} \right]\right|+\E\left[\left. e^{-\int_0^T g(s) dY(s)+\frac{1}{2} \int_0^T g^2(s) ds}\right|Y_{\Delta_n}\right],
$$
$$
G_n \geq -\left| \E\left[\left.-\int_0^T g(s) dY(s)+\frac{1}{2} \int_0^T g^2(s) ds \right|Y_{\Delta_n}\right] \right|- \E\left[\left.e^{-\int_0^T g(s) dY(s)+\frac{1}{2} \int_0^T g^2(s) ds} \right| Y_{\Delta_n}\right].
$$
And furthermore, using a similar argument as in {\bf Step} $\bf 1.1$, we can show that as $n$ tends to infinity,
\begin{align*}
\E\left[\left.-\int_0^T g(s) dY(s)+\frac{1}{2} \int_0^T g^2(s) ds \right| Y(\Delta_n), W_0^T \right] & =\left(-\int_0^T \bar{g}_{\Delta_n}(s) dY(s)+\frac{1}{2} \int_0^T \bar{g}_{\Delta_n}^2(s) ds\right)\\
& \hspace{-6cm} \times \E\left[\left.-\int_0^T (g(s)- \bar{g}_{\Delta_n})(s) dY(s)+\frac{1}{2} \int_0^T (g^2(s)-\bar{g}_{\Delta_n}^2(s)) ds \right| Y(\Delta_n), W_0^T \right]\\
& \hspace{-6cm} \to \left(-\int_0^T g(s) dY(s)+\frac{1}{2} \int_0^T g^2(s) ds\right)
\end{align*}
and
\begin{align*}
&\hspace{-1cm} \E[e^{-\int_0^T g(s) dY(s)+\frac{1}{2} \int_0^T g^2(s) ds}|Y(\Delta_n), W_0^T] \\
&= \E[e^{-\int_0^T \bar{g}_{\Delta_n}(s) dY(s)+\frac{1}{2} \int_0^T \bar{g}^2_{\Delta_n}(s) ds-\int_0^T (g(s)-\bar{g}_{\Delta_n}(s)) dY(s) + \frac{1}{2} \int_0^T (g^2(s)-\bar{g}_{\Delta_n}^2(s)) ds}|Y(\Delta_n), W_0^T]\\
&= e^{-\int_0^T \bar{g}_{\Delta_n}(s) dY(s)+\frac{1}{2} \int_0^T \bar{g}^2_{\Delta_n}(s) ds} \E[e^{-\int_0^T (g(s)-\bar{g}_{\Delta_n}(s)) dB(s) - \frac{1}{2} \int_0^T (g(s)-\bar{g}_{\Delta_n}(s))^2 ds}|Y(\Delta_n), W_0^T]\\
& \to e^{-\int_0^T g(s) dY(s)+\frac{1}{2} \int_0^T g^2(s) ds}.
\end{align*}
It then follows from the general Lebesgue dominated convergence theorem that
$$
\lim_{n \to \infty} \E[F_n] \to \E\left[-\int_0^T g(s) dY(s)+\frac{1}{2} \int_0^T g(s)^2 ds\right].
$$
A parallel argument can be used to show that
$$
\lim_{n \to \infty} \E[G_n] = \E[\log \E[e^{-\int_0^T g(s) dY(s)+\frac{1}{2} \int_0^T g(s)^2 ds}|Y_0^T]].
$$
So, under the condition (\ref{case-1}), we have shown that
$$
\lim_{n \to \infty} I(W_0^T; Y(\Delta_n))=I(W_0^T; Y_0^T).
$$

{\bf Step $\bf 2$.} In this step, we will use the convergence in \textbf{Step} $\bf 1$ and establish the theorem without the condition (\ref{case-1}).

Following Page $264$ of~\cite{li01}, we define, for any $k$,
\begin{equation} \label{stopping-time}
\tau_k = \begin{cases}
\inf \{t \leq T: \int_0^t g^2(s, W_0^s, Y_0^s) ds \geq k\}, \mbox{ if } \int_0^T g^2(s, W_0^s, Y_0^s) ds \geq k\\
T, \mbox{ if } \int_0^T g^2(s, W_0^s, Y_0^s) ds < k.
\end{cases}
\end{equation}
Then, we again follow~\cite{li01} and define a truncated version of $g$ as follows:
\begin{equation*}
g_{(k)}(t, \gamma_0^t, \phi_0^t)=g(t, \gamma_0^t, \phi_0^t) \mathbf{1}_{\int_0^t g^2(s, \gamma_0^t, \phi_0^s) ds < k}.
\end{equation*}
Now, define a truncated version of $Y$ as follows:
\begin{equation*}
Y_{(k)}(t)=\rho \int_0^t g_{(k)}(s, W_0^s, Y_0^s) ds +B(t), \quad t \in [0, T],
\end{equation*}
which, as elaborated on Page $265$ in~\cite{li01}, can be rewritten as
\begin{equation} \label{fixed-n}
Y_{(k)}(t)=\rho \int_0^t g_{(k)}(s, W_0^s, Y_{(k), 0}^s) ds +B(t), \quad t \in [0, \tau_k \wedge T].
\end{equation}
Note that for fixed $k$, the system in (\ref{fixed-n}) satisfies the condition (\ref{case-1}), and so the theorem holds true. To be more precise, note that
$$
I(W_0^T; Y_0^{\tau_k})=\E\left[\log \frac{d\mu_{\tau_k, Y|W}}{d\mu_{\tau_k, B}}(Y_0^{\tau_k})\right]-\E\left[\log \frac{d\mu_{\tau_k, Y}}{d\mu_{\tau_k, B}}(Y_0^{\tau_k})\right],
$$
where $\mu_{\tau_k, Y}$ and $\mu_{\tau_k, B}$ respectively denote the truncated versions of $\mu_{Y}$ and $\mu_{B}$ (from time $0$ to time $\tau_n$). Applying Theorem $7.10$ in~\cite{li01}, we obtain
$$
\frac{d\mu_{\tau_k, Y|W}}{d\mu_{\tau_k, B}}(Y_0^{\tau_n})=e^{\int_0^{\tau_k} g(s) dY(s)-\frac{1}{2} \int_0^{\tau_k} g^2(s)ds},
$$
and
$$
\frac{d\mu_{\tau_k, Y}}{d\mu_{\tau_k, B}}(Y_0^{\tau_k})=e^{\int_0^{\tau_k} \hat{g}(s) dY(s)-\frac{1}{2} \int_0^{\tau_k} \hat{g}^2(s) ds},
$$
where
$$
\hat{g}(s)=\E[g(s, W_0^s, Y_0^s)|Y_0^s].
$$
It then follows that
$$
I(W_0^T; Y_0^{\tau_k}) =\frac{1}{2} \E\left[\int_0^{\tau_k} (g(s)-\hat{g}(s))^2 ds\right].
$$
Notice that it can be easily verified that $\tau_k \to T$ as $k$ tends to infinity, which, together with the monotone convergence theorem, further yields that monotone increasingly,
$$
I(W_0^T; Y_0^{\tau_k})=\frac{1}{2} \E\left[\int_0^{\tau_k} (g(s)-\hat{g}(s))^2 ds\right] \to I(W_0^T; Y_0^T)=\frac{1}{2} \E\left[\int_0^T (g(s)-\hat{g}(s))^2 ds\right],
$$
as $k$ tends to infinity. By \textbf{Step} $\bf 1$, for any fixed $k_i$,
$$
\lim_{n \to \infty} I(W_0^T; Y(\Delta_n \cap [0, \tau_{k_i}]))=I(W_0^T; Y_0^{\tau_{k_i}}),
$$
which means that there exists a sequence $\{n_i\}$ such that, as $i$ tends to infinity, we have, monotone increasingly,
$$
I(W_0^T; Y(\Delta_{n_i} \cap [0, \tau_{k_i}])) \to I(W_0^T; Y_{0}^{T}).
$$
Since, by the fact that $Y_0^{\tau_k}$ coincides with $Y_0^T$ on the interval $[0, \tau_k \wedge T]$, we have
$$
I(W_0^T; Y(\Delta_{n_i})) \geq I(W_0^T; Y(\Delta_{n_i} \cap [0, \tau_{k_i}])).
$$
Now, using the fact that
$$
I(W_0^T; Y(\Delta_{n_i})) \leq I(W_0^T; Y_0^T),
$$
we conclude that as $i$ tends to infinity,
$$
\lim_{i \to \infty} I(W_0^T; Y(\Delta_{n_i})) = I(W_0^T; Y_{0}^T).
$$
A similar argument can be readily applied to any subsequence of $\{I(W_0^T; Y(\Delta_n))\}$, which will establish the existence of its further subsubsequence that converges to $I(W_0^T; Y_0^T)$, which implies that
$$
\lim_{n \to \infty} I(W_0^T; Y(\Delta_n))=I(W_0^T; Y_0^T).
$$
The proof of the theorem is then complete.

\begin{rem}  \label{mmse-1}
The arguments in the proof of Theorem~\ref{sampling-theorem-2} can be adapted to yield a sampling theorem for continuous-time minimum mean square error (MMSE), a quantity of central importance in estimation theory.

More precisely, consider the following continuous-time Gaussian feedback channel under the assumptions of Theorem~\ref{sampling-theorem-2}:
$$
Y(t)=\int_0^t X(s, M, Y_0^{s}) ds + B(t), \quad t \in [0, T].
$$
The MMSE is the limit of the MMSE based on the samples with respect to $\Delta_n$, namely,
$$
\int_0^T \E[(X(s)-\E[X(s)|Y_0^T])^2] ds=\lim_{n \to \infty} \int_0^T \E[(X(s)-\E[X(s)|Y(\Delta_n))^2] ds.
$$
To see this, note that the above-mentioned convergence follows from the fact that
$$
\E[\E^2[X(s)|Y_0^T]]=\E\left[\left(\frac{\int X(s, m_0^s, Y_0^s) d\mu_{Y|M=m}(Y_0^T)/d\mu_B d\mu_M(m)}{d\mu_{Y}(Y_0^T)/d\mu_B}\right)^2\right],
$$
and
$$
\E[\E^2[X(s)|Y(\Delta_n)]]=\E\left[\left(\frac{\int X(s, m_0^s, Y_0^s) d\mu_{Y|M=m}(Y(\Delta_n))/d\mu_{B} d\mu_M(m)}{d\mu_{Y}(Y(\Delta_n))/d\mu_{B}}\right)^2\right],
$$
and the proven fact that $d\mu_{Y}(Y(\Delta_n))/d\mu_{B}$ and $d\mu_{Y|M}(Y(\Delta_n))/d\mu_{B}$ respectively converge to $d\mu_{Y}(Y_{0}^T)/d\mu_{B} \mbox{ and } d\mu_{Y|M}(Y_{0}^T)/d\mu_{B}$ and a parallel argument as in establishing the convergence of $\{\E[F_n]\}$ and $\{\E[G_n]\}$ in the proof of Theorem~\ref{sampling-theorem-2}.

Similarly, we can also conclude that under the assumptions of Theorem~\ref{sampling-theorem-2}, the causal MMSE is the limit of the sampled causal MMSE, namely,
$$
\int_0^T \E[(X(s)-\E[X(s)|Y_0^s])^2] ds=\lim_{n \to \infty} \int_0^T \E[(X(s)-\E[X(s)|Y(\Delta_n \cap [0, s])])^2] ds.
$$
\end{rem}

\section{Proof of Theorem~\ref{approximation-theorem-1}} \label{proof-approximation-theorem-1}

In this section, we give the detailed proof of Theorem~\ref{approximation-theorem-1}.

We will first need the following lemma, which is parallel to Lemma~\ref{improved-liptser-1}.
\begin{lem} \label{improved-liptser-2}
Assume Conditions (d)-(f). Then, there exists $\varepsilon > 0$ and a constant $C > 0$ such that for all $n$,
\begin{equation} \label{exponential-finiteness-2}
\E [e^{\varepsilon \|Y_0^{(n), T}\|^2}] < C.
\end{equation}
\end{lem}

\begin{proof}
A discrete-time version of the proof of Lemma~\ref{improved-liptser-1} implies that there exists $\varepsilon > 0$ and a constant $C > 0$ such that for all $n$
$$
\E [e^{\varepsilon \sup_{i \in \{0, 1, \dots, n\}} (Y^{(n)}(t_{n, i}))^2}] < C,
$$
which, together with (\ref{linear-interpolation}), immediately implies (\ref{exponential-finiteness-2}).
\end{proof}

We also need the following lemma, which is parallel to Theorem $10.2.2$ in~\cite{kl92}.
\begin{lem}  \label{Y-Y}
Assume Conditions (d)-(f). Then, there exists a constant $C > 0$ such that for all $n$,
$$
\E[\|Y_0^{(n), T}-Y_0^T\|^2] \leq C \delta_{\Delta_n}.
$$
\end{lem}

\begin{proof}
Note that for any $n$, we have
$$
Y(t_{n, i+1})=Y(t_{n, i})+\int_{t_{n, i}}^{t_{n, i+1}} g(s, W_0^s, Y_0^s) ds + B(t_{n, i+1})-B(t_{n, i}),
$$
and
$$
Y^{(n)}(t_{n,i+1})=Y^{(n)}(t_{n,i})+\int_{t_{n,i}}^{t_{n,i+1}} g(s, W_0^{t_{n, i}}, Y_0^{(n),t_{n,i}}) ds+ B(t_{n,i+1})-B(t_{n,i}).
$$
It then follows that
\begin{equation} \label{Y-diff}
Y(t_{n,i+1})-Y^{(n)}(t_{n, i+1})=Y(t_{n, i})-Y^{(n)}(t_{n, i})+\int_{t_{n, i}}^{t_{n, i+1}} (g(s, W_0^s, Y_0^s)-g(s, W_0^{t_{n, i}}, Y_0^{(n), t_{n, i}})) ds.
\end{equation}
Now, for any $t$, choose $n_0$ such that $t_{n, n_0} \leq t < t_{n, n_0+1}$. Now, a recursive application of (\ref{Y-diff}), coupled with Conditions (d) and (e), yields that for some $L > 0$,
{\small \begin{align*}
\hspace{-2.2cm} Y(t)-Y^{(n)}(t)&=\sum_{i=0}^{n_0} \int_{t_{n,i}}^{t_{n,i+1}} (g(s, W_0^s, Y_0^s)-g(t_{n,i}, W_0^{t_{n, i}}, Y_0^{(n),t_{n,i}})) ds+\int_{t_{n, n_0+1}}^{t} (g(s, W_0^s, Y_0^s)-g(t_{n,i}, W_0^{t_{n, n_0+1}}, Y_0^{(n),t_{n, n_0+1}})) ds\\
&\leq \sum_{i=0}^{n_0} \int_{t_{n,i}}^{t_{n,i+1}} L |s-t_{n, i}| + L \|W_0^s-W_0^{t_{n, i}}\|+ L \|Y_0^s-Y_0^{(n),s}\|+L \|Y_0^{(n),s}-Y_0^{(n),t_{n,i}}\| ds\\
&+\int_{t_{n, n_0+1}}^{t} L |s-t_{n, n_0+1}| + L \|W_0^s-W_0^{t_{n, n_0+1}}\|+ L \|Y_0^s-Y_0^{(n),s}\|+L \|Y_0^{(n),s}-Y_0^{(n),t_{n, n_0+1}}\| ds.
\end{align*}}
Noticing that for any $s$ with $t_{n, i} \leq s < t_{n, i+1}$, we have
$$
\hspace{-1cm} \|Y_0^{(n),s}-Y_0^{(n),t_{n,i}}\|^2 \leq |Y^{(n)}(t_{n,i+1})-Y^{(n)}(t_{n,i})|^2 \leq 2 \left| \int_{t_{n,i}}^{t_{n,i+1}} g(s, W_0^{t_{n,i}}, Y_0^{(n),t_{n,i}}) ds \right|^2+ 2 |B(t_{n,i+1})-B(t_{n,i})|^2,
$$
which, together with Condition (e) and the fact that for all $n$ and $i$,
\begin{equation} \label{BB}
\E[|B(t_{n,i+1})-B(t_{n,i})|^2] = O(\delta_{\Delta_n}),
\end{equation}
implies that
\begin{equation} \label{YY}
\E[\|Y_0^{(n),s}-Y_0^{(n),t_{n,i}}\|^2] = O(\delta_{\Delta_n}).
\end{equation}
Noting that the constants in the two terms $O(\delta_{\Delta_n})$ in (\ref{BB}) and (\ref{YY}) can be chosen uniform over all $n$, a usual argument with the Gronwall inequality and Condition (f) applied to $\E[\|Y_0^t-Y_0^{(n), t}\|^2]$ completes the proof of the theorem.
\end{proof}

We are now ready for the proof of Theorem~\ref{approximation-theorem-1}.
\begin{proof}[Proof of Theorem~\ref{approximation-theorem-1}]
We proceed in two steps.

{\bf Step $\bf 1$.} In this step, we establish the theorem assuming that there exists a constant $C > 0$ such that for all $w_0^T \in C[0, T]$ and all $y_0^T \in C[0, T]$,
\begin{equation} \label{approximating-case-1}
\int_0^T g^2(s, w_0^s, y_0^s) ds < C.
\end{equation}

We first note that straightforward computations yield
\begin{align*}
\hspace{-2cm} f_{Y^{(n)}(\Delta_n)|W}(y^{(n)}(\Delta_n)|w_0^T) &=\prod_{i=1}^n f(y^{(n)}_{t_{n,i}}|y_{t_{n,0}}^{(n), t_{n,i-1}},w_0^{t_{n, i-1}})\\
&=\prod_{i=1}^{n} \frac{1}{\sqrt{2\pi(t_{n, i}-t_{n, i-1})}}\exp \left({-\frac{(y^{(n)}_{t_{n, i}}-y^{(n)}_{t_{n, i-1}}-  \int_{t_{n, i-1}}^{t_{n, i}} g(s, w_0^{t_{n, i-1}}, y_0^{(n), t_{n, i-1}})ds)^2}{2(t_{n, i}-t_{n, i-1})}} \right),
\end{align*}
(here we have used the shorter notations $y^{(n)}_{t_{n, i}}, y^{(n)}_{t_{n, i-1}}$ for $y^{(n)}(t_{n, i}), y^{(n)}(t_{n, i-1})$, respectively) and
$$
\hspace{-2cm} f_{Y^{(n)}(\Delta_n)}(y^{(n)}(\Delta_n))=\int \prod_{i=1}^{n} \frac{1}{\sqrt{2\pi(t_{n, i}-t_{n, i-1})}}\exp \left({-\frac{(y^{(n)}_{t_{n, i}}-y^{(n)}_{t_{n, i-1}}- \int_{t_{n, i-1}}^{t_{n, i}} g(s, w_0^{t_{n, i-1}}, y_0^{(n), t_{n, i-1}})ds)^2}{2(t_{n, i}-t_{n, i-1})}} \right) d\mu_W(w),
$$
which further lead to
\begin{equation} \label{YgivenM}
\hspace{-2cm} f_{Y^{(n)}(\Delta_n)|W}(Y^{(n)}(\Delta_n)|W_0^T)=\prod_{i=1}^{n} \frac{1}{\sqrt{2\pi(t_{n, i}-t_{n, i-1})}}\exp \left({-\frac{(Y^{(n)}_{t_{n, i}}-Y^{(n)}_{t_{n, i-1}}- \int_{t_{n, i-1}}^{t_{n, i}} g(s, W_0^{t_{n, i-1}}, Y_0^{(n), t_{n, i-1}})ds)^2}{2(t_{n, i}-t_{n, i-1})}} \right),
\end{equation}
and
\begin{equation} \label{justY}
\hspace{-2.3cm} f_{Y^{(n)}(\Delta_n)}(Y^{(n)}(\Delta_n))=\int \prod_{i=1}^{n} \frac{1}{\sqrt{2\pi(t_{n, i}-t_{n, i-1})}}\exp \left({-\frac{(Y^{(n)}_{t_{n, i}}-Y^{(n)}_{t_{n, i-1}}- \int_{t_{n, i-1}}^{t_{n, i}} g(t_{n, i-1}, w_0^{t_{n, i-1}}, Y^{(n), t_{n, i-1}}_{0}) ds)^2}{2(t_{n, i}-t_{n, i-1})}} \right) d\mu_W(w).
\end{equation}
With (\ref{YgivenM}) and (\ref{justY}), we have
{\small \begin{align*}
\hspace{-3mm} I(W_0^T; Y^{(n)}(\Delta_n)) & =\E[\log f_{Y^{(n)}(\Delta_n)|W}(Y^{(n)}(\Delta_n)|W_0^T)]-\E[\log f_{Y^{(n)}(\Delta_n)}(Y^{(n)}(\Delta_n))] \\
&\hspace{-4.3cm}=\E\left[\log \prod_{i=1}^n \exp \left({-\frac{-2\int_{t_{n, i-1}}^{t_{n, i}} g(s, W_0^{t_{n, i-1}}, Y_0^{(n), t_{n, i-1}}) ds \; (Y^{(n)}_{t_{n, i}}-Y^{(n)}_{t_{n, i-1}})+ (\int_{t_{n, i-1}}^{t_{n, i}} g(s, W_0^{t_{n, i-1}}, Y_0^{(n), t_{n, i-1}})ds)^2}{2(t_{n, i}-t_{n, i-1})}} \right)\right]\\
& \hspace{-4.3cm} -\E\left[\log \int \prod_{i=1}^n \exp \left({-\frac{-2\int_{t_{n, i-1}}^{t_{n, i}} g(s, w_0^{t_{n, i-1}}, Y^{(n), t_{n, i-1}}_{0}) ds \; (Y^{(n)}_{t_{n, i}}-Y^{(n)}_{t_{n, i-1}})+ (\int_{t_{n, i-1}}^{t_{n, i}} g(s, w_0^{t_{n, i-1}}, Y^{(n), t_{n, i-1}}_{0}) ds)^2}{2(t_{n, i}-t_{n, i-1})}} \right) d\mu_W(w) \right] \\
&\hspace{-4.3cm} =\E\left[\sum_{i=1}^n \left({-\frac{-2\int_{t_{n, i-1}}^{t_{n, i}} g(s, W_0^{t_{n, i-1}}, Y^{(n), t_{n, i-1}}_{0}) ds \; (Y^{(n)}_{t_{n, i}}-Y^{(n)}_{t_{n, i-1}})+ (\int_{t_{n, i-1}}^{t_{n, i}} g(s, W_0^{t_{n, i-1}}, Y^{(n), t_{n, i-1}}_{0}) ds)^2}{2(t_{n, i}-t_{n, i-1})}} \right)\right]\\
&\hspace{-4.3cm} -\E\left[\log \int \exp \sum_{i=1}^n \left({-\frac{-2\int_{t_{n, i-1}}^{t_{n, i}} g(s, w_0^{t_{n, i-1}}, Y^{(n), t_{n, i-1}}_{0}) ds \; (Y^{(n)}_{t_{n, i}}-Y^{(n)}_{t_{n, i-1}})+ (\int_{t_{n, i-1}}^{t_{n, i}} g(s, w_0^{t_{n, i-1}}, Y^{(n), t_{n, i-1}}_{0}) ds)^2}{2(t_{n, i}-t_{n, i-1})}} \right) d\mu_W(w) \right].
\end{align*}}
On the other hand, it is well known (see, e.g.,~\cite{ih93}) that
\begin{align*}
I(W; Y_0^T) & =\E\left[\log \frac{d\mu_{Y|W}}{d\mu_B}(Y_0^T)\right]-\E\left[\log \frac{d\mu_{Y}}{d\mu_B}(Y_0^T)\right]\\
& =\E \left[\log \exp\left[ \int_0^T g(s, W_0^s, Y_0^s) dY(s)-\frac{1}{2} \int_0^T g^2(s, W_0^s, Y_0^s) ds \right] \right] \\
& -\E \left[\log \int \exp\left[ \int_0^T g(s, w_0^s, Y_0^s) dY(s)-\frac{1}{2} \int_0^T g^2(s, w_0^s, Y_0^s) ds d\mu_W(w) \right] \right] \\
& =\E \left[\int_0^T g(s, W_0^s, Y_0^s) dY(s)-\frac{1}{2} \int_0^T g^2(s, W_0^s, Y_0^s) ds\right] \\
& -\E \left[\log \int \exp\left[ \int_0^T g(s, w_0^s, Y_0^s) dY(s)-\frac{1}{2} \int_0^T g^2(s, w_0^s, Y_0^s) ds d\mu_W(w) \right] \right].
\end{align*}
Now, we compute
\begin{align*}
& \int_0^T g(s, W_0^s, Y_0^s) dY(s)-\sum_{i=1}^{n} \frac{\int_{t_{n, i-1}}^{t_{n, i}} g(s, W_0^{t_{n, i-1}}, Y_0^{(n), t_{n, i-1}}) ds (Y^{(n)}_{t_{n, i}}-Y^{(n)}_{t_{n, i-1}})}{t_{n, i}-t_{n, i-1}}\\
& =\int_0^T g(s, W_0^s, Y_0^s) dY(s)-\sum_{i=1}^{n} g(t_{n, i-1}, W_0^{t_{n, i-1}}, Y_0^{(n), t_{n, i-1}}) (Y^{(n)}_{t_{n, i}}-Y^{(n)}_{t_{n, i-1}})\\
& -\sum_{i=1}^{n} \frac{\int_{t_{n, i-1}}^{t_{n, i}} (g(s, W_0^{t_{n, i-1}}, Y_0^{(n), t_{n, i-1}})-g(t_{n, i-1}, W_0^{t_{n, i-1}}, Y_0^{(n), t_{n, i-1}})) ds (Y^{(n)}_{t_{n, i}}-Y^{(n)}_{t_{n, i-1}})}{t_{n, i}-t_{n, i-1}}.
\end{align*}
It can be easily checked that the second term of the right hand side of the above equality converges to $0$ in mean. For the first term, we have
\begin{align*}
&\int_0^T g(s, W_0^s, Y_0^s) dY(s)-\sum_{i=1}^{n} g(t_{n, i-1}, W_0^{t_{n, i-1}}, Y_0^{(n), t_{n, i-1}}) (Y^{(n)}_{t_{n, i}}-Y^{(n)}_{t_{n, i-1}})\\
&=\sum_{i=1}^n \int_{t_{n, i-1}}^{t_{n, i}} g(s, W_0^s, Y_0^s) dY(s)-\sum_{i=1}^n g(t_{n, i-1}, W_0^{t_{n, i-1}}, Y_0^{(n), t_{n, i-1}}) (Y_{t_{n, i}}-Y_{t_{n, i-1}})\\
&+\sum_{i=1}^n g(t_{n, i-1}, W_0^{t_{n, i-1}}, Y_0^{(n), t_{n, i-1}}) ((Y_{t_{n, i}}-Y_{t_{n, i-1}})-(Y^{(n)}_{t_{n, i}}-Y^{(n)}_{t_{n, i-1}}))\\
&=\sum_{i=1}^n \int_{t_{n, i-1}}^{t_{n, i}} g(s, W_0^s, Y_0^s) dY(s)-\sum_{i=1}^n \int_{i=1}^n g(t_{n, i-1}, W_0^{t_{n, i-1}}, Y_0^{(n), t_{n, i-1}}) dY(s)\\
&+\sum_{i=1}^n g(t_{n, i-1}, W_0^{t_{n, i-1}}, Y_0^{(n), t_{n, i-1}}) ((Y_{t_{n, i}}-Y_{t_{n, i-1}})-(Y^{(n)}_{t_{n, i}}-Y^{(n)}_{t_{n, i-1}}))\\
&=\sum_{i=1}^n \int_{t_{n, i-1}}^{t_{n, i}} (g(s, W_0^s, Y_0^s)-g(t_{n, i-1}, W_0^{t_{n, i-1}}, Y_0^{(n), t_{n, i-1}})) dY(s)\\
&+\sum_{i=1}^n g(t_{n, i-1}, W_0^{t_{n, i-1}}, Y_0^{(n), t_{n, i-1}}) ((Y_{t_{n, i}}-Y_{t_{n, i-1}})-(Y^{(n)}_{t_{n, i}}-Y^{(n)}_{t_{n, i-1}})).
\end{align*}
It then follows from Conditions (d) and (e), Lemmas~\ref{improved-liptser-1},~\ref{improved-liptser-2} and~\ref{Y-Y} that
\begin{equation}  \label{first-half}
\E\left[\left|\sum_{i=1}^{n} \frac{\int_{t_{n, i-1}}^{t_{n, i}} g(s, W_0^{t_{n, i-1}}, Y_0^{(n), t_{n, i-1}}) ds (Y^{(n)}_{t_{n, i}}-Y^{(n)}_{t_{n, i-1}})}{t_{n, i}-t_{n, i-1}} - \int_0^T g(s, W_0^s, Y_0^s) dY(s)\right| \right]=O(\delta^{\frac{1}{2}}_{\Delta_n}).
\end{equation}
And using a similar argument as above, we deduce that
\begin{equation}  \label{second-half}
\E\left[\left|\frac{1}{2} \sum_{i=1}^n \frac{(\int_{t_{n, i-1}}^{t_{n, i}} g(s, W_0^{t_{n, i-1}}, Y^{(n),t_{n, i-1}}_{0})ds)^2}{t_{n, i}-t_{n, i-1}}-\frac{1}{2} \int_0^T g(s, W_0^s, Y_0^s)^2 ds \right|\right] = O(\delta^{\frac{1}{2}}_{\Delta_n}).
\end{equation}
It then follows from (\ref{first-half}) and (\ref{second-half}) that as $n$ tends to infinity,
$$
\E\left[ \left|\sum_{i=1}^n \left({-\frac{-2\int_{t_{n, i-1}}^{t_{n, i}} g(s, W_0^{t_{n, i-1}}, Y_0^{(n), t_{n,i-1}}) ds(Y^{(n)}_{t_{n, i}}-Y^{(n)}_{t_{n, i-1}})+(\int_{t_{n, i-1}}^{t_{n, i}} g(s, W_0^{t_{n, i-1}}, Y_0^{(n), t_{n,i-1}}) ds)^2}{2(t_{n, i}-t_{n, i-1})}} \right) \right. \right.
$$
\begin{equation} \label{M-conv}
\left. \left. - \int_0^T g(s, W_0^s, Y_0^s) dY(s) + \frac{1}{2} \int_0^T g(s, W_0^s, Y_0^s)^2 ds \right|\right] = O(\delta^{\frac{1}{2}}_{\Delta_n}).
\end{equation}

We now establish the following convergence:
{\small \begin{equation} \label{m-conv}
\E\left[\log \int \exp A^{(n)}(w)d\mu_W(w) \right] \to \E\left[\log \int \exp A(w)d\mu_W(w) \right].
\end{equation}}
where
$$
\hspace{-1cm} A^{(n)}(w)=\sum_{i=1}^n \left({-\frac{-2\int_{t_{n, i-1}}^{t_{n, i}} g(s, w_0^{t_{n, i-1}}, Y^{(n), t_{n,i-1}}_0) ds(Y^{(n)}_{t_{n, i}}-Y^{(n)}_{t_{n, i-1}})+ (\int_{t_{n, i-1}}^{t_{n, i}} g(s, w_0^{t_{n, i-1}}, Y^{(n), t_{n,i-1}}_0) ds)^2}{2(t_{n, i}-t_{n, i-1})}} \right).
$$
and let
$$
A(w)={\int_0^T g(s, w_0^s, Y_0^s) dY(s)-\frac{1}{2} \int_0^T g^2(s, w_0^s, Y_0^s) ds}.
$$
Note that using a parallel argument as the derivation of (\ref{M-conv}), we can establish
\begin{equation}  \label{m-conv-1}
\E \int \left| A^{(n)}(w)-A(w) \right| d\mu_W(w) \to 0,
\end{equation}
as $n$ tends to infinity; and similarly as in the derivation of (\ref{S-prime}), from Conditions (d), (e) and (f), Lemmas~\ref{improved-liptser-1},~\ref{improved-liptser-2} and~\ref{Y-Y}, we deduce that
\begin{equation} \label{m-conv-2}
\E\left[\int \left| \exp A^{(n)}(w)-\exp A(w) \right| d\mu_W(w)\right] \to 0
\end{equation}
as $n$ tends to infinity. And note that we always have
\begin{equation} \label{m-conv-3}
\left| \log \int \exp A^{(n)}(w) d\mu_W(w) \right| \leq \int \exp A^{(n)}(w) d\mu_W(w)+ \left| \int A^{(n)}(w) d\mu_W(w) \right|.
\end{equation}
So, by the general Lebesgue dominated convergence theorem with (\ref{m-conv-1}), (\ref{m-conv-2}) and (\ref{m-conv-3}), we have
$$
\E\left[\log \int \exp A^{(n)}(w) d\mu_W(w)\right] \to \E\left[\log \int \exp A(w) d\mu_W(w)\right].
$$
So, under the condition (\ref{approximating-case-1}), we have established the theorem.

{\bf Step $\bf 2$.}  In this step, we will use the convergence in \textbf{Step} $\bf 1$ and establish the theorem without the condition (\ref{approximating-case-1}).

Defining the stopping $\tau_k$, $g_{(k)}$ and $Y_{(k)}$ as in the proof of Theorem~\ref{sampling-theorem-2}, we again have:
\begin{equation*}
Y_{(k)}(t)=\rho \int_0^t g_{(k)}(s, W_0^s, Y_{(k), 0}^s) ds +B(t), \quad t \in [0, \tau_k \wedge T].
\end{equation*}
For any fixed $k$, applying the Euler-Maruyama approximation as in (\ref{Euler-Maruyama-Sequence}) and (\ref{linear-interpolation}) to the above channel with respect to $\Delta_n$, we obtain the process $Y_{(k)}^{(n)}(\cdot)$.

Now, by the fact that
\begin{align*}
I(W_0^T; Y^{(n)}(\Delta_n))& =\E[\log f_{Y^{(n)}(\Delta_n)|W}(Y^{(n)}(\Delta_n)|W_0^T)]-\E[\log f_{Y^{(n)}(\Delta_n)}(Y^{(n)}(\Delta_n))] \\
&= \E[A^{(n)}(W)] - \E\left[\log \int \exp A^{(n)}(w) d\mu_W(w)\right]\\
& \geq 0,
\end{align*}
we deduce that
\begin{align*}
& \hspace{-0.5cm}\E[\log f_{Y^{(n)}(\Delta_n)}(Y^{(n)}(\Delta_n))] \\
& \hspace{-0.5cm} \leq \E[\log f_{Y^{(n)}(\Delta_n)|W}(Y^{(n)}(\Delta_n)|W_0^T)]\\
&\hspace{-0.5cm} =\E\left[\sum_{i=1}^n \left({-\frac{-2\int_{t_{n, i-1}}^{t_{n, i}} g(s, W_0^{t_{n, i-1}}, Y^{(n), t_{n, i-1}}_{0}) ds \; (Y^{(n)}_{t_{n, i}}-Y^{(n)}_{t_{n, i-1}})+ (\int_{t_{n, i-1}}^{t_{n, i}} g(s, W_0^{t_{n, i-1}}, Y^{(n), t_{n, i-1}}_{0}) ds)^2}{2(t_{n, i}-t_{n, i-1})}} \right)\right] \\
& \hspace{-0.5cm}  = \E\left[\sum_{i=1}^n {\frac{(\int_{t_{n, i-1}}^{t_{n, i}} g(s, W_0^{t_{n, i-1}}, Y^{(n), t_{n, i-1}}_{0}) ds)^2}{2(t_{n, i}-t_{n, i-1})}}\right]\\
& \hspace{-0.5cm}  \leq \E\left[\sum_{i=1}^n {\frac{(\int_{t_{n, i-1}}^{t_{n, i}} g(s, W_0^{t_{n, i-1}}, Y^{(n), t_{n, i-1}}_{(k), 0})^2 ds)\int_{t_{n, i-1}}^{t_{n, i}} ds}{2(t_{n, i}-t_{n, i-1})}}\right]\\
& \hspace{-0.5cm} = \E\left[\sum_{i=1}^{n} \int_{t_{n, i-1}}^{t_{n, i}} g(s, W_0^{t_{n, i-1}}, Y^{(n), t_{n, i-1}}_{(k), 0})^2 ds\right]\\
& \hspace{-0.5cm} = \E\left[\int_{0}^{T} g(s, W_0^{\lfloor s \rfloor_{\Delta_n}}, Y^{(n), \lfloor s \rfloor_{\Delta_n}}_{0})^2 ds\right],
\end{align*}
where $\lfloor s \rfloor_{\Delta_n}$ denote the unique number $n_0$ such that $t_{n, n_0} \leq s < t_{n, n_0+1}$.
Now, using the easily verifiable fact that
$$
\frac{1}{\int \exp A^{(n)}(w) d\mu_W(w)} =\E[\exp(-A^{(n)}(W))|Y_0^T],
$$
and Jensen's inequality, we deduce that
$$
\E\left[\log \frac{1}{\int \exp A^{(n)}(w) d\mu_W(w)}\right] = \E \left[\log \E[\exp(-A^{(n)}(W))|Y_0^T] \right] \leq \log \E[\exp(-A^{(n)}(W))] \leq 0,
$$
where for the last inequality, we have applied Fatou's lemma as in deriving (\ref{Fatou-Lemma}). It then follows that
$$
0 \leq \E[\log f_{Y^{(n)}(\Delta_n)}(Y^{(n)}(\Delta_n))] \leq \E\left[\sum_{i=1}^{n} \int_{t_{n, i-1}}^{t_{n, i}} g(s, W_0^{t_{n, i-1}}, Y^{(n), t_{n, i-1}}_{0})^2 ds\right],
$$
which further implies that
$$
I(W_0^T; Y^{(n)}(\Delta_n)) \leq \E\left[\sum_{i=1}^{n} \int_{t_{n, i-1}}^{t_{n, i}} g(s, W_0^{t_{n, i-1}}, Y^{(n), t_{n, i-1}}_{0})^2 ds\right].
$$
Now, using the fact that $Y^{(n)}$ and $Y^{(n)}_{(k)}$ coincide over $[0, \tau_k \wedge T]$, one verifies that for any $\varepsilon > 0$,
\begin{align*}
I(W_0^T; Y^{(n)}(\Delta_n))-I(W_0^T; Y^{(n)}_{(k), \Delta_n}) & \leq \E\left[\int_{\tau_k}^{T} g(s, W_0^{\lfloor s \rfloor_{\Delta_n}}, Y^{(n), \lfloor s \rfloor_{\Delta_n}}_{0})^2 ds\right]\\
& \hspace{-5cm} \leq \E\left[\int_{\tau_k}^{T} g(s, W_0^{\lfloor s \rfloor_{\Delta_n}}, Y^{(n), \lfloor s \rfloor_{\Delta_n}}_{0})^2 ds; T-\tau_k \leq \varepsilon \right]+\E\left[\int_{\tau_k}^{T} g(s, W_0^{\lfloor s \rfloor_{\Delta_n}}, Y^{(n), \lfloor s \rfloor_{\Delta_n}}_{0})^2 ds; T-\tau_k > \varepsilon \right]\\
& \hspace{-5cm} \leq \int_{T-\varepsilon}^{T} \E\left[ g(s, W_0^{\lfloor s \rfloor_{\Delta_n}}, Y^{(n), \lfloor s \rfloor_{\Delta_n}}_{0})^2 \right] ds +\E\left[\int_{\tau_k}^{T} g(s, W_0^{\lfloor s \rfloor_{\Delta_n}}, Y^{(n), \lfloor s \rfloor_{\Delta_n}}_{0})^2 ds; T-\tau_k > \varepsilon \right].
\end{align*}
Using the easily verifiable fact that $\{\tau_k\}$ converges to $T$ in probability uniformly over all $n$ and the fact that $\varepsilon$ can be arbitrarily small, we conclude that as $k$ tends to infinity, uniformly over all $n$,
\begin{equation} \label{absolute-value-bounded}
I(W_0^T; Y^{(n)}_{(k)}(\Delta_n)) \to I(W_0^T; Y^{(n)}(\Delta_n)).
\end{equation}
Next, an application of the monotone convergence theorem, together with the fact that $\tau_k \to T$ as $k$ tends to infinity, yields that monotone increasingly
$$
I(W_0^T; Y_0^{\tau_n})=\frac{1}{2} \E\left[\int_0^{\tau_n} (g(s)-\hat{g}(s))^2 ds\right] \to I(W_0^T; Y_0^T)=\frac{1}{2} \E\left[\int_0^T (g(s)-\hat{g}(s))^2 ds\right]
$$
as $n$ tends to infinity. By \textbf{Step} $\bf 1$, for any fixed $k_i$,
$$
\lim_{n \to \infty} I(W_0^T; Y_{(k_i)}^{(n)}(\Delta_n))=I(W_0^T; Y_0^{\tau_{k_i}}),
$$
which means that there exists a sequence $\{n_i\}$ such that, as $i$ tends to infinity,
$$
I(W; Y_{(k_i)}^{(n_i)}(\Delta_n)) \to I(W; Y_{0}^{T}).
$$
Moreover, by (\ref{absolute-value-bounded}),
$$
\lim_{i \to \infty } I(W_0^T; Y_{(k_i)}^{(n_i)}(\Delta_{n_i})) = \lim_{i \to \infty } I(W_0^T; Y^{(n_i)}(\Delta_n)),
$$
which further implies that
$$
\lim_{i \to \infty} I(W_0^T; Y^{(n_i)}(\Delta_n))=I(W_0^T; Y_0^T).
$$
The theorem then follows from a usual subsequence argument as in the proof of Theorem~\ref{sampling-theorem-2}.
\end{proof}

\begin{rem} \label{mmse-2}
Parallel to Remark~\ref{mmse-1}, the arguments in the proof of Theorem~\ref{approximation-theorem-1} can be adapted to yield an approximation theorem in estimation theory.

More precisely, consider the following continuous-time Gaussian feedback channel under the assumptions in Theorem~\ref{approximation-theorem-1}:
$$
Y(t)=\int_0^t X(s, M, Y_0^{s}) ds + B(t), \quad t \in [0, T].
$$
The MMSE is the limit of the approximated MMSE, namely,
$$
\int_0^T \E[(X(s)-\E[X(s)|Y_0^T])^2] ds=\lim_{n \to \infty} \int_0^T \E[(X^{(n)}(s)-\E[X^{(n)}(s)|Y^{(n)}(\Delta_n)])^2] ds.
$$
In more detail, the above-mentioned convergence follows from the fact that
$$
\int_0^T \E[(X^{(n)}(s))^2] ds \to \int_0^T \E[(X(s))^2] ds
$$
and the fact that
$$
\E[\E^2[X(s)|Y_0^T]]=\E\left[\left(\frac{\int X(s, m_0^s, Y_0^s) \exp(A(m)) d\mu_M(m)}{\int \exp(A(m)) d\mu_M(m)}\right)^2\right],
$$
and the fact that
$$
\E[\E^2[X^{(n)}(s)|Y^{(n)}(\Delta_n)]]=\E\left[\left(\frac{\int X^{(n)}(s, m_0^s, Y_0^s) \exp(A^{(n)}(m)) d\mu_M(m)}{\int \exp(A^{(n)}(m)) d\mu_M(m)}\right)^2\right].
$$
Then, using a similar argument as in the proof of Theorem~\ref{approximation-theorem-1}, we can show
$$
\lim_{n \to \infty} \E[\E^2[X^{(n)}(s)|Y^{(n)}(\Delta_n)]]=\E[\E^2[X(s)|Y_0^T]],
$$
which implies the claimed convergence.

Similarly, we can also conclude that with the assumptions in Theorem~\ref{approximation-theorem-1}, the causal MMSE is the limit of the approximated causal MMSE, namely,
$$
\int_0^T \E[(X(s)-\E[X(s)|Y_0^s])^2] ds=\lim_{n \to \infty} \int_0^T \E[(X^{(n)}(s)-\E[X^{(n)}(s)|Y^{(n)}(\Delta_n \cap [0, s])^2] ds.
$$
\end{rem}

\section{Proof of Theorem~\ref{Theorem-MAC}} \label{proof-Theorem-MAC}

In the section, we give the proof of Theorem~\ref{Theorem-MAC}. For notational convenience only, we will assume $m=2$, the case with a generic $m$ being completely parallel. We will first need the following lemma, which is a key component in our treatment of both continuous-time Gaussian MACs.
\begin{lem} \label{independent-OUs}
For any $\eps > 0$, there exist two independent Ornstein-Uhlenbeck processes $\{X_i(s): s \geq 0\}$, $i=1, 2$, satisfying the following power constraint:
\begin{equation} \label{a-p-c}
\mbox{for $i=1, 2$, there exists $P_i > 0$ such that for all $t > 0$,  } \frac{1}{t} \int_0^t E[X_i^2(s)] ds = P_i,
\end{equation}
such that for all $T$,
\begin{equation} \label{comma-sum}
|I_T(X_1, X_2; Y)/T-(P_1+P_2)/2| \leq \eps,
\end{equation}
and
\begin{equation} \label{conditional}
|I_T(X_1; Y|X_2)/T-P_1/2| \leq \eps, \quad |I_T(X_2; Y|X_1)/T-P_2/2| \leq \eps,
\end{equation}
moreover,
\begin{equation} \label{treat-as-noise}
|I_T(X_1; Y)/T-P_1/2| \leq \eps, \quad |I_T(X_2; Y)/T-P_2/2| \leq \eps,
\end{equation}
where
\begin{equation}  \label{to-be-interpreted-as-channel}
Y(t)=\int_0^t X_1(s) ds + \int_0^t X_2(s) ds + B(t), \quad t \geq 0.
\end{equation}
Here (and often in the remainder of the paper) the subscript $T$ means that the (conditional) mutual information is computed over the time period $[0, T]$.
\end{lem}

\begin{proof}
For $a > 0$, consider the following two independent Ornstein-Uhlenbeck processes $X_i(t)$, $i=1, 2$, given by
$$
X_i(t)=\sqrt{2a P_i} \int_{-\infty}^t e^{-a(t-s)} dB_i(s),
$$
where $B_i$, $i=1, 2$, are independent standard Brownian motions. Obviously, for $X_i$ defined as above, (\ref{a-p-c}) is satisfied. A parallel version of the proof of Theorem $6.2.1$ of~\cite{ih93} yields that
$$
I_T(X_1, X_2; Y) = I_T(X_1+X_2; Y)= \frac{1}{2} \int_0^T E[(X_1(t)+X_2(t)-E[X_1(t)+X_2(t)|Y_0^t])^2] dt.
$$
It then follows from Theorem $6.4.1$ in~\cite{ih93} (applied to the Ornstein-Uhlenbeck process $X_1(t)+X_2(t)$ that as $a \rightarrow \infty$,
$$
I_T(X_1, X_2; Y)/T = I_T(X_1+X_2; Y)/T \rightarrow (P_1+P_2)/2,
$$
uniformly in $T$, which establishes (\ref{comma-sum}).

For $i=1, 2$, define
$$
\tilde{Y}_i(t)=\int_0^t X_i(s) ds + B(t), \quad t > 0.
$$
As in the proof of Theorem $6.4.1$ in~\cite{ih93}, we deduce that for $i=1, 2$, $I_T(X_i; \tilde{Y}_i)/T$ tend to $P_i/2$ uniformly in $T$. Now, since $X_1$ and $X_2$ are independent, we have for any fixed $T$,
$$
I_T(X_1; Y|X_2)=I_T(X_1; \tilde{Y}_1|X_2)=I_T(X_1; \tilde{Y}_1),
$$
and
$$
I_T(X_2; Y|X_1)=I_T(X_2; \tilde{Y}_2|X_1)=I_T(X_2; \tilde{Y}_2),
$$
which immediately implies (\ref{conditional}).

Now, by the chain rule of mutual information,
$$
I_T(X_1, X_2; Y)=I_T(X_1; Y)+ I_T(X_2; Y|X_1)=I_T(X_2; Y)+ I_T(X_1; Y|X_2),
$$
which, together with (\ref{comma-sum}) and (\ref{conditional}), implies (\ref{treat-as-noise}).

\end{proof}

\begin{rem}  \label{xianming-explanation}
With $X_i$, $i=1, 2$, regarded as channel inputs, (\ref{to-be-interpreted-as-channel}) can be reinterpreted as a white Gaussian MAC. For $i \neq j$, $I(X_i; Y)$, the reliable transmission rate of  $X_i$ when $X_j$ is not known can be arbitrarily close to $I(X_i; Y|X_j)$, the reliable transmission rate of $X_i$ when $X_j$ is known. In other words, for white Gaussian MACs, knowledge about other user's inputs will not help to achieve faster transmission rate, and therefore, they can be simply treated as noises. An more intuitive explanation of this result is as follows: for the Ornstein-Uhlenbeck process $X_i$ as specified in the proof, its power spectral density can be computed as
$$
f_i(\lambda)=\frac{2 a P_i}{2 \pi (\lambda^2+a^2)},
$$
which is ``negligible'' compared to that of the white Gaussian noise (which is the constant $1$) as $a$ tends to infinity. Lemma~\ref{independent-OUs} is a key ingredient for deriving the capacity regions of white Gaussian MACs.
\end{rem}

We also need some result on the information stability of continuous-time Gaussian processes. Let $(U, V)  = \{(U(t), V(t)), t \geq 0\}$ be a continuous Gaussian system (which means $U(t), V(t)$ are pairwise Gaussian stochastic processes). Define
$$
\varphi^{(T)}(u, v)=\frac{d\mu_{UV}^{(T)}}{d\mu_U^{(T)} \times \mu_V^{(T)}}(u, v), \qquad (u, v) \in C[0, T] \times C[0, T],
$$
where $\mu_U^{(T)}$, $\mu_V^{(T)}$ and $\mu_{UV}^{(T)}$ denote the probability distributions of $U_0^T$, $V_0^T$ and their joint distribution, respectively. For any $\varepsilon > 0$, we denote by $\mathcal{T}_{\varepsilon}^{(T)}$ the $\varepsilon$-typical set:
$$
\mathcal{T}^{(T)}_{\varepsilon}=\left\{(u, v) \in C[0, T] \times C[0, T]; \frac{1}{T} |\log \varphi^{(T)}(u, v)-I_T(U, V)| \leq \varepsilon \right\}.
$$
The pair $(U, V)$ is said to be {\it information stable}~\cite{pi64} if for any $\varepsilon > 0$,
$$
\lim_{T \to \infty} \mu^{(T)}_{UV}(\mathcal{T}_{\varepsilon}) = 1.
$$

The following theorem is a rephrased version of Theorem 6.6.2. in~\cite{ih93}.
\begin{lem} \label{information-stability}
The Gaussian system $(U, V)$ is information stable provided that
$$
\lim_{T \rightarrow \infty} \frac{I_T(U; V)}{T^2}=0.
$$
\end{lem}
\noindent Lemma~\ref{information-stability} will be used in the proof of Theorem~\ref{Theorem-MAC} to establish, roughly speaking, that almost all sequences are jointly typical.

With Lemmas~\ref{independent-OUs} and~\ref{information-stability}, Theorem~\ref{Theorem-MAC} largely follows from a lengthy yet almost routine argument, which is included below due to a number of technical challenges in the proof.
\begin{proof}[Proof of Theorem~\ref{Theorem-MAC}]

\textbf{The converse part.} In this part, we will show that for any sequence of $(T, (e^{T R_1}, e^{T R_2}), (P_1, P_2))$-codes with $P_e^{(T)} \rightarrow 0$ as $T \rightarrow \infty$, the rate pair $(R_1, R_2)$ will have to satisfy
$$
R_1 \leq P_1/2, \qquad R_2 \leq P_2/2.
$$

Fix $T$ and consider the above-mentioned $(T, (e^{T R_1}, e^{T R_2}), (P_1, P_2))$-code. By the code construction, it is possible to estimate the messages $(M_1, M_2)$ from the channel output $Y_0^T$ with a low probability of error. Hence, the conditional entropy of $(M_1, M_2)$ given $Y_0^T$ must be small; more precisely, by Fano's inequality,
$$
H(M_1,M_2|Y_0^T) \leq T(R_1+R_2)P^{(T)}_e +H(P^{(T)}_e) = T \varepsilon_T,
$$
where $\varepsilon_T \rightarrow 0$ as $T \rightarrow \infty$. Then, we have $$
H(M_1|Y^T) \leq H(M_1, M_2|Y^T) \leq T\varepsilon_T, \quad
H(M_2|Y^T) \leq H(M_1, M_2|Y^T) \leq T\varepsilon_T.
$$
Now, we can bound the rate $R_1$ as follows:
\begin{align*}
T R_1 & = H(M_1)\\
      & = I(M_1; Y_0^T)+H(M_1|Y_0^T)\\
      & \leq I(M_1; Y_0^T) + T \varepsilon_T\\
      & \leq H(M_1)-H(M_1|Y_0^T) + T \varepsilon_T \\
      & \leq H(M_1|M_2)-H(M_1|Y_0^T, M_2)+T \varepsilon_T \\
      & = I(M_1; Y_0^T|M_2) + T \varepsilon_T.
\end{align*}

Conditioning on $M_2$ and applying Theorem $6.2.1$ in~\cite{ih93}, we have
\begin{align*}
I(M_1; Y_0^T|M_2) &=\frac{1}{2} E\left[\int_0^T E[(X_1(t)+X_2(t)-\hat{X}_1(t)-\hat{X}_2(t))^2|M_2] dt \right]\\
&=\frac{1}{2} \int_0^T E[(X_1(t)+X_2(t)-\hat{X}_1(t)-\hat{X}_2(t))^2] dt,
\end{align*}
where $\hat{X}_i(t)=E[X_i(t)|Y_0^T, M_2]$, $i=1, 2$. Noticing that $X_2=\hat{X}_2$, we then have
$$
I(M_1; Y_0^T|M_2)=\frac{1}{2} \int_0^T E[(X_1(t)-\hat{X}_1(t))^2] dt,
$$
which, together with (\ref{PowerConstraint-MAC}), implies that $R_1 \leq P_1/2$. A completely parallel argument will yield that  $R_2 \leq P_2/2$.

\textbf{The achievability part.} In this part, we will show that as long as $(R_1, R_2)$ satisfying
\begin{equation} \label{strictly-less-than}
0 \leq R_1 < P_1/2, \quad 0 \leq R_2 < P_2/2,
\end{equation}
we can find a sequence of $(T, (e^{T R_1}, e^{T R_2}), (P_1, P_2))$-codes with $P_e^{(T)} \rightarrow 0$ as $T \rightarrow \infty$. The argument consists of several steps as follows.

\emph{Codebook generation}: For a fixed $T > 0$ and $\varepsilon > 0$, assume that $X_1$ and $X_2$ are independent Ornstein-Uhlenbeck processes over $[0, T]$ with respective variances $P_1-\varepsilon$ and $P_2-\varepsilon$, and that $(R_1,R_2)$ satisfying (\ref{strictly-less-than}). Generate $e^{T R_1}$ independent codewords $X_{1, i}$, $i \in \{1, 2, \ldots, e^{T R_1}\}$, of length $T$, according to the distribution of $X_1$. Similarly, generate $e^{T R_2}$ independent codewords $X_{2, j}$, $j \in \{1, 2, \ldots, e^{T R_2}\}$, of length $T$, according to the distribution of $X_2$. These codewords (which may not satisfy the power constraint in (\ref{PowerConstraint-MAC})) form the codebook, which is revealed to the senders and the receiver.

\emph{Encoding:} To send message $i \in \mathcal{M}_1$, sender $1$ sends the codeword $X_{1, i}$. Similarly, to send $j \in \mathcal{M}_2$, sender $2$ sends $X_{2, j}$.

\emph{Decoding:} For any fixed $\varepsilon > 0$, let $\mathcal{T}_{\varepsilon}^{(T)}$ denote the set of {\it jointly typical} $(x_1, x_2, y)$ sequences, which is defined as follows:
$$
\mathcal{T}_{\varepsilon}^{(T)}=\{(x_1,x_2, y) \in C[0, T] \times C[0, T] \times C[0, T]:|\log
\varphi_1(x_1, x_2, y)-I_T(X_1,X_2; Y)|\leq  T \varepsilon,
$$
$$
|\log
\varphi_2(x_1, x_2, y)-I_T(X_1; X_2, Y)|\leq  T \varepsilon, |\log \varphi_3(x_1, x_2, y)-I_T(X_2; X_1,Y)|\leq  T \varepsilon \},
$$
where
$$
\varphi_1(x_1, x_2, y)=\frac{d\mu_{X_1 X_2 Y}}{d\mu_{X_1 X_2} \times \mu_{Y}}(x_1, x_2, y),
$$
$$
\varphi_2(x_1, x_2, y)=\frac{d\mu_{X_1 X_2 Y}}{d\mu_{X_1} \times \mu_{X_2 Y}}(x_1,x_2,y),
$$
$$
\varphi_3(x_1, x_2, y)=\frac{d\mu_{X_1 X_2 Y}}{d\mu_{X_2} \times \mu_{X_1 Y}}(x_1, x_2, y).
$$
Here we remark that it is easy to check that the above Radon-Nykodym derivatives are all well-defined; see, e.g., Theorem $7.7$ of~\cite{li01} for sufficient conditions for their existence. Based on the received output $y \in C[0, T]$, the receiver chooses the pair $(i, j)$ such that
$$
(x_{1, i}, x_{2, j}, y) \in \mathcal{T}_{\varepsilon}^{(T)},
$$
if such a pair $(i, j)$ exists and is unique; otherwise, an error is declared. Moreover, an error will be declared if the chosen codeword does not satisfy the power constraint in (\ref{PowerConstraint-MAC}).

\emph{Analysis of the probability of error:} Now, for fixed $T, \varepsilon > 0$, define
$$
E_{ij}=\{(X_{1, i} ,X_{2, j} ,Y) \in \mathcal{T}_{\varepsilon}^{(T)}\}.
$$
By symmetry, we assume, without loss of generality, that (1,1) was
sent. Define $\pi^{(T)}$ to be the event that
$$
\int_0^T (X_{1, 1}(t))^2 dt > P_1 T, \quad \int_0^T (X_{2, 1}(t))^2 dt > P_2 T.
$$
Then, $\hat{P}_e^{(T)}$, the error probability for the above coding scheme (where codewords violating the power constraint are allowed), can be upper bounded as follows:
$$
\hat{P}_e^{(T)} = P(\pi^{(T)} \cup E_{11}^c \bigcup \cup_{(i, j) \neq (1, 1)} E_{ij})
$$
$$
\leq P(\pi^{(T)})+P(E_{11}^c)+\sum_{i \neq 1, j =1} P(E_{i1})+\sum_{i=1, j \neq 1} P(E_{1j})+\sum_{i \neq 1, j \neq 1} P(E_{ij}).
$$
So, for any $i, j \neq 1$, we have
$$
\hat{P}^{(T)}_e \leq P(\pi^{(T)})+ P(E^c_{11})+e^{T R_1} P(E_{i1})+e^{T R_2} P(E_{1j})+ e^{T R_1 + T R_2} P(E_{ij})
$$
Using the well-known fact that an Ornstein-Uhlenbeck process is ergodic~\cite{le08, ku04}, we deduce that $P(\pi^{(T)}) \to 0$ as $T \to \infty$. And by Lemma~\ref{information-stability} and Theorem $6.2.1$ in~\cite{ih93}, we have
$$
\lim_{T \rightarrow \infty} P((X_{1, 1}, X_{2, 1}, Y) \in \mathcal{T}_{\varepsilon}^{(T)})=1 \mbox{ and thus } \lim_{T\rightarrow \infty} P(E^c_{11})=0.
$$
Now, we have for any $i \neq 1$,
\begin{align*}
P(E_{i1}) & = P((X_{1, i}, X_{2, 1}, Y) \in \mathcal{T}^{(T)}_{\varepsilon}) \\
&= \int_{(x_1, x_2, y) \in \mathcal{T}_{\varepsilon}^{(T)}} d\mu_{X_1}(x_1) d\mu_{X_2 Y}(x_2, y) \\
&= \int_{\mathcal{T}_{\varepsilon}^{(T)}}\frac{1}{\varphi_1(x_1, x_2, y)}d\mu_{X_1 X_2 Y}(x_1, x_2, y) \\
&\leq \int_{\mathcal{T}_{\varepsilon}^{(T)}}e^{-I_T(X_1; X_2, Y)+\varepsilon T}d\mu_{X_1 X_2 Y}(x_1, x_2, y) \\
&= e^{-I_T(X_1; Y|X_2)+\varepsilon T},
\end{align*}
where we have used the independence of $X_1$ and $X_2$, and the consequent fact that
$$
I_T(X_1; X_2, Y)=I_T(X_1; X_2)+I_T(X_1; Y| X_2)=I_T(X_1; Y|X_2).
$$
Similarly, we have, for $j \neq 1$,
$$
P(E_{1j}) \leq e^{-I_T(X_2; Y|X_1)+\varepsilon T},
$$
and for $i, j \neq 1$,
$$
P(E_{ij}) \leq e^{-I_T(X_1, X_2; Y)+\varepsilon T}.
$$
It then follows that
$$
\hat{P}^{(T)}_e \leq P(\pi^{(T)})+P(E_{11}^c)+ e^{T R_1+ \varepsilon T-I_T(X_1; Y|X_2)}+e^{T R_2+ \varepsilon T-I_T(X_2; Y|X_1)}+e^{T R_1+T R_2+ \varepsilon T-I_T(X_1, X_2; Y)}.
$$
By Lemma~\ref{independent-OUs}, one can choose independent OU processes $X_1, X_2$ such that
$I_T(X_1; Y|X_2)/T \rightarrow (P_1-\eps)/2$, $I_T(X_2; Y|X_1)/T \rightarrow (P_2-\eps)/2$ and $I_T(X_1, X_2; Y)/T \rightarrow (P_1+P_2-2 \eps)$ uniformly in $T$. This implies that with $\eps$ chosen sufficiently small, we have $\hat{P}^{(T)}_e \rightarrow 0$, as $T \rightarrow \infty$. In other words, there exists a sequence of good codes (which may not satisfy the power constraint) with low average error probability. Now, from each of the above codes, we delete the worse half of the codewords (any codeword violating the power constraint will be deleted since it must have error probability $1$). Then, with only slightly decreased transmission rate, the remaining codewords will satisfy the power constraint and will have small maximum error probability (and thus small average error probability $P_e^{(T)}$), which implies that the rate pair $(R_1,R_2)$ is achievable.
\end{proof}

\begin{rem}
The achievability part can be proven alternatively, which will be roughly described as follows: for arbitrarily small $\eps > 0$, by Lemma~\ref{independent-OUs}, one can choose independent Ornstein-Uhlenbeck processes $X_i$ with respective variances $P_i-\eps$, $i=1, 2$, such that $I_T(X_i; Y)/T$ approaches $(P_i-\eps)/2$. Then, a parallel random coding argument with $X_j$, $j \neq i$, being treated as noise at receiver $i$ shows that the rate pair $((P_1-\eps)/2, (P_2-\eps)/2)$ can be approached, which yields the achievability part.
\end{rem}

\section{Proof of Theorem~\ref{Theorem-IC-Without-Feedback}}  \label{proof-Theorem-IC-Without-Feedback}

For notational convenience only, we only prove the case when $n=2$; the case when $n$ is generic is similar.

{\bf The converse part.} In this part, we will show that for any sequence of $(T, (e^{T R_1}, e^{T R_2}), (P_1, P_2))$ codes with
$P_e^{(T)} \rightarrow 0$, the rate pair $(R_1, R_2)$ will have to satisfy
\begin{equation} \label{less-than-squared}
R_1 \leq a_{11}^2 P_1/2, \quad R_2 \leq a_{22}^2 P_2/2.
\end{equation}

Fix $T$ and consider the above-mentioned $(T, (e^{T R_1}, e^{T R_2}), (P_1, P_2))$ code. By the code construction, for $i=1, 2$, it is possible to estimate the messages $M_i$ from the channel output $Y_{i, 0}^T$ with an arbitrarily low probability of error. Hence, by Fano's inequality, for $i=1, 2$,
$$
H(M_i|Y_{i, 0}^T) = T \varepsilon_{i, T},
$$
where $\varepsilon_{i, T} \to 0$ as $T \to \infty$. We then have
$$
T R_1=H(M_1)=H(M_1|M_2)=I(M_1; Y_1|M_2)+H(M_1|M_2, Y_1) \leq I(M_1; Y_1|M_2)+ T \varepsilon_{1, T},
$$
As in the proof of Theorem~\ref{Theorem-MAC}, we have
$$
I(M_1; Y_{1, 0}^T|M_2)= \frac{a^2_{11}}{2} \int_0^T E[(X_1(s)-E[X_1(s)|M_2, Y_{1, 0}^s])^2] ds.
$$
It then follows that
$$
T R_1 \leq \frac{a_{11}^2}{2} \int_0^T E[(X_1(s)-E[X_1(s)|M_2, Y_{1, 0}^s])^2] ds + T \varepsilon_{1, T},
$$
which implies that $R_1 \leq a_{11}^2 P_1/2$. With a parallel argument, one can derive that $R_2 \leq a_{22}^2 P_2/2$. The proof for the converse part is then complete.

{\bf The achievability part.} We only sketch the proof of this part. For arbitrarily small $\eps > 0$, by Lemma~\ref{independent-OUs}, one can choose independent Ornstein-Uhlenbeck processes $X_i$ with respective variances $P_i-\eps$, $i=1, 2$, such that $I_T(X_i; Y)/T$ approaches $a^2_{ii}(P_i-\eps)/2$. Then, a parallel random coding argument as in the proof of Theorem~\ref{Theorem-MAC} with $X_j$, $j \neq i$, being treated as noise at receiver $i$ shows that the rate pair $(a^2_{11}(P_1-\eps)/2, a^2_{22}(P_2-\eps)/2)$ can be approached, which yields the achievability part.

\section{Proof of Theorem~\ref{Theorem-BC-Without-Feedback}} \label{proof-Theorem-BC-Without-Feedback}

One of the important tools that plays a key role in discrete-time network information theory is the entropy power inequality~\cite{co2006, el11}, which can be applied to compare information-theoretic quantities involving different users.
The following lemma, which, despite its strikingly different form, serves the typical function of a discrete-time entropy power inequality.
\begin{lem} \label{Xianming-Lemma}
Consider a continuous-time white Gaussian channel characterized by the following equation
$$
Y(t)=\sqrt{snr} \int_0^t X(s) ds + B(t), \quad t \geq 0,
$$
where $snr \geq 0$ denotes the signal-to-noise ratio in the channel and $M$ is the message to be transmitted through the channel. Then, for any fixed $T$, $I_T(M; Y)/snr$ is a monotone decreasing function of $snr$.
\end{lem}

\begin{proof}
For notational convenience, in this proof, we write $I_T(M; Y)$ as $I_T(snr)$.
By Theorem $6.2.1$ in~\cite{ih93}, we have
$$
I_T(snr)=\frac{snr}{2} \int_0^T E[(X(s)-E[X(s)|Y_0^s])^2] ds,
$$
and Theorem $6$ in~\cite{gu05}, we have (the derivative is with respect to $snr$)
$$
I'_T(snr)=\frac{1}{2} \int_0^T E[(X(s)-E[X(s)|Y_0^T])^2] ds.
$$
It then follows that
\begin{align*}
\left(\frac{I_T(snr)}{snr}\right)' &=\frac{1}{snr}\left(I'_T(snr)-\frac{I_T(snr)}{snr}\right) \\
&=\frac{1}{2 snr} \left(\int_0^T E[(X(s)-E[X(s)|Y_0^T])^2] ds-\int_0^T E[(X(s)-E[X(s)|Y_0^s])^2] ds \right) \leq 0,
\end{align*}
which immediately implies the lemma.
\end{proof}

We are now ready for the proof of Theorem~\ref{Theorem-BC-Without-Feedback}.
\begin{proof}[Proof of Theorem~\ref{Theorem-BC-Without-Feedback}]
For notational convenience only, we prove the case when $n=2$, the case when $n$ is generic being parallel.

{\bf The converse part.} Without loss of generality, we assume that
$$
snr_1 \geq snr_2.
$$
We will show that for any sequence of $(T, (e^{T R_1}, e^{T R_2}), P)$ codes with
$P_e^{(T)} \rightarrow 0$ as $T \rightarrow \infty$, the rate pair $(R_1, R_2)$ will have to satisfy
\begin{equation} \label{divided-by-snr}
\frac{R_1}{snr_1}+\frac{R_2}{snr_2} \leq \frac{P}{2}.
\end{equation}

Fix $T$ and consider the above-mentioned $(T, (e^{T R_1}, e^{T R_2}), P)$-code. By the code construction, for $i=1, 2$, it is possible to estimate the messages $M_i$ from the channel output $Y_{i, 0}^T$ with an arbitrarily low probability of error. Hence, by Fano's inequality, for $i=1, 2$,
$$
H(M_i|Y_{i, 0}^T) \leq T R_i P^{(T)}_e +H(P^{(T)}_e) = T \varepsilon_{i, T},
$$
where $\varepsilon_{i, T} \rightarrow 0$ as $T \rightarrow \infty$. It then follows that
\begin{equation} \label{eq-1}
T R_1 = H(M_1) = H(M_1|M_2) \leq I(M_1; Y_{1, 0}^T|M_2) + T \varepsilon_{1, T},
\end{equation}
\begin{equation} \label{eq-2}
T R_2 = H(M_2) \leq I(M_2; Y_{2, 0}^T) + T \varepsilon_{2, T}.
\end{equation}
By the chain rule of mutual information, we have
\begin{equation} \label{eq-3}
I(M_1, M_2; Y_{2, 0}^T)=I(M_2; Y_{2, 0}^T)+I(M_1; Y_{2, 0}^T|M_2) \geq I(M_2; Y_{2, 0}^T) + \frac{snr_2}{snr_1} I(M_1; Y_{1, 0}^T|M_2),
\end{equation}
where, for the inequality above, we have applied Lemma~\ref{Xianming-Lemma}.
Now, by Theorem $6.2.1$ in~\cite{ih93}, we have
$$
I(M_1, M_2; Y_{2, 0}^T) = \frac{snr_2}{2} \int_0^T E[(X(s)-E[X(s)|Y_{2, 0}^s])^2] ds \leq \frac{snr_2}{2} \int_0^T E[X^2(s)] ds,
$$
which, together with (\ref{eq-1}), (\ref{eq-2}), (\ref{eq-3}) and (\ref{PowerConstraint-BC}), immediately implies the converse part.

{\bf The achievability part.} We only sketch the proof of this part. For an arbitrarily small $\eps > 0$, by Theorem $6.4.1$ in~\cite{ih93}, one can choose an Ornstein-Uhlenbeck processes $\tilde{X}$ with variance $P-\eps$, such that $I_T(\tilde{X}; Y_i)/T$ approaches $snr_i (P-\eps)/2$. For any $0 \leq \lambda \leq 1$, let
$$
X(t)=\sqrt{\lambda} X_1(t)+ \sqrt{1-\lambda} X_2(t), \quad t \geq 0,
$$
where $X_1$ and $X_2$ are independent copies of $\tilde{X}$. Then, by a similar argument as in the proof of Lemm~\ref{independent-OUs}, we deduce that $I_T(X_1; Y_1)/T, I_T(X_2; Y_2)/T$ approach $snr_1 \lambda (P-\eps)/2$, $snr_2 (1-\lambda)(P-\eps)/2$, respectively. Then, a parallel random coding argument as in the proof of Theorem~\ref{Theorem-MAC} such that
\begin{itemize}
\item when encoding, $X_i$ only carries the message meant for receiver $i$;
\item when decoding, receiver $i$ treats $X_j$, $j \neq i$, as noise,
\end{itemize}
shows that the rate pair $(snr_1 \lambda (P-\eps)/2, snr_2 (1-\lambda) (P-\eps)/2)$ can be approached, which immediately establishes the achievability part.
\end{proof}

\begin{rem}
For the achievability part, instead of using the power sharing scheme as in the proof, one can also employ the following time sharing scheme: set $X$ to be $X_1$ for $\lambda$ fraction of the time, and $X_2$ for $1-\lambda$ fraction of the time. Then, it is straightforward to check this scheme also achieves the rate pair $(snr_1 \lambda (P-\eps)/2, snr_2 (1-\lambda) (P-\eps)/2)$. This, from a different perspective, echoes the observation in~\cite{la03} that time sharing achieves the capacity region of a white Gaussian BC as the bandwidth limit tends to infinity.
\end{rem}

\end{document}